\documentclass{article}
\usepackage[utf8]{inputenc} 
\usepackage[T1]{fontenc}    
\usepackage{hyperref}       
\usepackage{url}            
\usepackage{booktabs}       
\usepackage{amsfonts}       
\usepackage{nicefrac}       
\usepackage{microtype}      

\usepackage{cite}
\usepackage{graphicx}
\usepackage{amssymb}
\usepackage{amsthm}
\usepackage{dsfont}
\usepackage{amsmath}
\usepackage{xcolor}
\usepackage{thmtools, thm-restate}
\usepackage{hyperref}
\usepackage{bbold}
\usepackage{xspace}

\pagecolor{white}
\DeclareMathAlphabet\mathbfcal{OMS}{cmsy}{b}{n}

\title{Concerning Iterative Graph Normalization\\ and Maximum Weight Independent Sets}
\author{Laurent Guigues \\ Amazon \\guigues@amazon.com}
\date{December, 14, 2020}

\newcommand{\IS}{\mathcal{I}}
\newcommand{\MIS}{\overline{\mathcal{I}}}
\newcommand{\SUPP}{supp}
\newcommand{\DEG}{deg}
\newcommand{\N}[1]{\mathcal{V}_{#1}}
\newcommand{\R}{\mathcal{N}}
\newcommand{\Rh}{\mathcal{N}_h}
\newcommand{\NORMG}{\mathsf{N}} 
\newcommand{\DEN}[1]{dens_G(#1)}
\newcommand{\EM}[1]{\emph{\bf #1}}
\newcommand{\norm}[1]{\left\lVert#1\right\rVert}
\newcommand{\DOM}{\mathcal{D}}

\newcommand{\ONE}{\mathbb{1}}
\newcommand{\HD}{\oslash}
\newcommand{\HM}{\odot}
\newcommand{\DIAG}{diag}
\newcommand{\IMN}[1]{\mathcal{I_{#1}}}

\newcommand{\AP}{\textsf{AP}\xspace}
\newcommand{\SK}{\textsf{SK}\xspace}
\newcommand{\MWIS}{\textsf{MWIS}\xspace}

\newcommand{\IGN}{\textsf{IGN}\xspace}
\newcommand{\ICN}{\textsf{ICN}\xspace}
\newcommand{\KAKO}{\textsf{WG}\xspace}
\newcommand{\SA}{\textsf{SA}\xspace}

\newcommand{\PERM}{\mathcal{P}}
\newcommand{\BIRK}{\mathcal{B}}
\newcommand{\RoN}{\mathcal{R}}
\newcommand{\CoN}{\mathcal{C}}
\newcommand{\CrN}{\mathcal{X}}

\newenvironment{sloppypar*}
 {\sloppy\ignorespaces}
 {\par}

\newtheorem{proposition}{Proposition}[section]
\newtheorem{theorem}{Theorem}[section]
\newtheorem{conjecture}{Conjecture}[section]
\newtheorem{corollary}{Corollary}[theorem]

\theoremstyle{definition}
\newtheorem{definition}{Definition}[section]
\theoremstyle{remark}

\begin{document}
\maketitle

\begin{abstract}
We consider a very simple operation on a weighted graph which consists 
in normalizing in parallel all the weights of its nodes by the sum of the weights on their neighborhoods 
plus their own weight.
We call this operation Graph Normalization (GN). We study its iteration (IGN) 
and a variant in which we apply a non-linear activation function to the weights after each normalization.

We show that the indicator vectors of the 
maximal independent sets of $G$ are the only binary fixed points of GN, that they 
are attractive under simple conditions on the activation function and we characterize their basins of attraction. 
We enumerate a number of other fixed points and we prove repulsivity for some classes.
Based on extensive experiments and different theoretical arguments we conjecture 
that \IGN always converges and converges to a binary solution for non-linear activations.
If our conjectures are correct, \IGN would thus be 
a differentiable approximation algorithm for the Maximum Weight 
Independent Set problem (MWIS), a central NP-hard 
optimization problem with numerous applications.

IGN is closely related to a greedy approximation algorithm of \MWIS
by Kako et al. \cite{kako2005approximation} 
which has a proven approximation ratio.
Experimental results show that \IGN provides solutions of very close total weight 
to those found by the algorithm of Kako et al.

An important special case of \MWIS is the optimal Assignment Problem. 
In this context, \IGN corresponds to an iterative matrix normalization scheme 
which is closely related to the Sinkhorn-Knopp algorithm (\SK). 
The difference is that \SK projects the initial weight matrix to a doubly stochastic matrix, 
hence finds a soft assignment,
while \IGN projects to a permutation matrix, i.e. a crisp solution of the assignment problem.
We relate our scheme to the Softassign algorithm, or equivalently to entropy-regularized \SK,
and provide comparative experimental results.

As Graph Normalization is differentiable, its iterations can be embedded 
into a machine learning framework and used to train end-to-end 
any model which includes a graphical optimization step 
which can be cast as a maximum weight independent set problem.
This includes problems such as graph and hypergraph matching, sequence alignment, 
clustering, ranking, etc. with applications in multiple domains such as computer vision, 
mobile networks optimization or manufacturing.

\end{abstract}

\tableofcontents

\section{Introduction}
Given a set of items with associated values
and a set of incompatibility constraints between pairs of items, 
the maximum weight independent set problem (\MWIS) is the problem 
of selecting a subset of the items maximizing the sum of the values of the selected items 
and respecting the pairwise incompatibilities.
The graphical formulation of \MWIS represents the items as the weighted nodes 
of an undirected graph and the conflicts as its edges, 
and asks for finding a set of non-adjacent nodes, 
called an independent set or stable set of the graph, 
of maximum total weight.

A number of combinatorial optimization problems are special cases of \MWIS, 
such as the assignment problem and its multi-dimensional generalization, 
the min cost flow, shortest path, max flow, max weighted clique, 
or the general set packing problems. 
\MWIS is also at the heart of graph coloring problems \cite{mehrotra1996column} and 
it has been proven that any MAP estimation for probability distributions over finite domains 
can be reduced to a \MWIS problem \cite{sanghavi2009message}.

In practice, solving these problems has a wide range of applications, 
such as in economy for combinatorial auctions \cite{de2003combinatorial},
in data clustering \cite{li2012clustering, grover2019stochastic},
coding theory and error-correcting codes \cite{butenko2002finding},
interval selection problems arising in manufacturing \cite{spieksma1998approximating}, map labelling problems \cite{verweij1999optimisation}, frequency assignment problems in wireless networks \cite{malesinska1997graph}, 
or in computer vision for image segmentation \cite{brendel2010segmentation}, multi-object tracking \cite{papageorgiou2009maximum, brendel2011multiobject}, stereo-vision \cite{horaud1989stereo}, 
action recognition or robotics \cite{vento2013graph} - 
including recent deep learning-based approaches \cite{caetano2009learning, wang2019learning, zanfir2018deep}.

However, \MWIS is NP-hard \cite{garey1979computers} and hard to approximate.
In general, \MWIS is Poly-APX-complete, which is the class of the hardest problems that can be approximated efficiently to within a factor polynomial in the input size.
Even for bounded degree graphs, finding the maximal independent set (without considering weighted graphs) is MAXSNP-complete, implying that, for some constant c (depending on the degree) it is NP-hard to find an approximate solution that comes within a factor of c of the optimum \cite{papadimitriou1991optimization}.

Due to the importance of the problem, finding approximation algorithms to \MWIS have received 
much attention in the combinatorial optimization community. 
Different heuristic or greedy algorithms have been proposed \cite{busygin2002heuristic, sakai2003note, kako2005approximation} as well as many linear or quadratic integer programming-based methods, such as branch-and-price \cite{warrier2005branch} 
or branch-and-bound \cite{warren2006combinatorial}. 
A recent branch-and-reduce algorithm is able to exactly solve \MWIS on large graphs 
of millions of nodes \cite{lamm2019exactly}.
Many algorithms dedicated to sub-problems or sub-classes of graphs - some of which are known to be $P$ - 
have also been proposed.

By nature, all these combinatorial algorithms are not differentiable. 
The only differentiable approach which we are aware of for the general \MWIS problem 
is the specialization of the message passing / belief propagation framework to the particular case of \MWIS \cite{shah2005max, sanghavi2009message}.
Despite some proofs of optimality when they converge, message passing algorithms, such as max-product, have no general proof of convergence on general graphs with loops 
and are known in practice to diverge on some instances.

Differentiable approximation algorithms to combinatorial problems have regained attention in the 
recent years with the advent of machine learning and especially deep learning, because 
they can be embedded into end-to-end learning frameworks trained by gradient descent-based loss 
minimization techniques.

In this context, the Sinkhorn-Knopp approximation algorithm for the Assignment problem\cite{sinkhorn1967concerning}
has received much attention in the recent years \cite{patrini2018sinkhorn, mena2018learning, tay2020sparse}. 
The Assignment Problem (\AP) is a classical graphical optimization problem in which 
one has to assign a set of meals to a set of guests, 
maximizing the sum of preferences of each guest for each meal,
with the constraint that each guest can at most enjoy one meal and meals cannot be shared (1-1 assignment).
In the balanced \AP problem, one further assumes the same number of guests and of meals.
\AP can be cast as a maximum weight bipartite graph matching problem, 
in which the nodes correspond to the guests and the meals, and the weights of the edges to the preferences.
The goal is then to find a set of non intersecting edges of maximum total weight. 
One sees that by exchanging the role of the nodes and the edges 
- i.e. taking the dual of the graph in the hypergraph duality sense - \AP is a special instance 
of a \MWIS problem.
Now, this instance has $P$-time combinatorial algorithms, 
such as Kuhn's Hungarian algorithm \cite{kuhn1955hungarian}.

\AP can also be approached from the linear optimization point of view. 
The preferences are then arranged in a square matrix $W$
and one looks for a \emph{permutation} matrix $P$ such that $\sum_{ij} W_{ij}P_{ij}$ is maximal.
The constraint that $P$ must be a permutation matrix, i.e. represents a 1-1 matching, 
can be expressed by the fact that $P$ must be a doubly stochastic matrix, 
i.e. a matrix whose rows and columns are all normalized (sum up to $1$), 
and whose entries are \emph{binary}.
In their seminal work of 1967, Sinkhorn and Knopp have shown that alternating rows and columns normalization 
of a square nonnegative matrix $W$ converges to a doubly stochastic matrix if and only if $W$ has support \cite{sinkhorn1967concerning}. 
The Sinkhorn-Knopp algorithm (\SK) thus gives a \emph{relaxed} solution to the balanced 
assignment problem in the linear programming sense: it doesn't provide a binary 
solution to the problem but a real valued solution respecting the same row/column normalization constraints.
The properties of the \SK algorithm have been largely studied, see for example a review in \cite{knight2008sinkhorn}.

Using \SK to enforce the doubly stochastic condition,
and embedding the optimization into a deterministic annealing framework, 
Kosowsky and Yuille \cite{kosowsky1994invisible} 
have proposed the Softassign algorithm (\SA), which is proven to 
converge at zero temperature to an optimal solution of \AP \cite{rangarajan1997convergence}.
\SA has been used to approximate the traveling salesman problem, the graph matching problem,
and the graph partitioning problem \cite{gold1996softmax}.
\SK also plays a central role in the theory of optimal transport 
and in this context, it has been shown that \SA (which is not called like that in this community), 
corresponds to an entropy-regularized problem, 
the temperature corresponding to the intensity of the regularization \cite{mena2018learning}.

The iterative graph normalization algorithm that we introduce here is closely related to the \SK algorithm 
but applies to the more general problem of \MWIS and not only to \AP.
In fact, we initially discovered it by modifying \SK.
In the context of a graph matching application, 
we needed to use a fixed number of iterations of \SK and unrolling the \SK algorithm loop,
we asked ourselves "should we start by row or column normalization? does it matter?" and then 
"what if we normalized \emph{simultaneously} the rows and the columns?".
We tried that out - dividing each entry of the matrix by the sum of the values in its row and column, 
without counting twice the element itself - and we had the surprise that it was then converging to 
a \emph{binary} solution, a permutation matrix, instead of a soft assignment.
We called this operation "matrix cross normalization", 
as we were normalizing the matrix by the sum of its elements in crosses 
rather than alternating normalization in rows and in columns.
We then realized that it was a special case of a more general dynamical system on weighted graphs 
and that the constraint on the solution to be a permutation for the case of the assignment problem 
corresponded in the general case to the constraint of being a maximal independent set of the initial graph.

As usual, in this paper we present our findings the other way around.
We start from the general graph formulation which we ended with, 
and motivate it by the most general problem, i.e. the 
maximum weight independent set problem, 
to specialize it to the assignment problem 
and relate it to the Sinkhorn-Knopp and Softassign algorithms.

We prove a number of properties of Iterative Graph Normalization (\IGN),
in particular local attractivity of maximal independent sets 
and repulsivity of non-maximal ones under suitable activation functions.
We also prove convergence for complete graphs.

However, we couldn't prove general convergence.
Based on multiple experiments, we anyway formulate the following conjectures:

\begin{conjecture}\label{conj_convergence}
Iterative graph normalization always converges.
\end{conjecture}
\begin{conjecture}\label{conj_convergence_to_MIS}
For any node-weighted graph with all distinct weights, iterative graph normalization
with a suitable non-linear activation converges to a maximal independent set of the graph.
\end{conjecture}

The proper definitions and conditions under which we claim that these conjectures hold will be 
clarified below. 

\subsection{Outline of the paper}
The organization of the paper is as follows.
The section \ref{sec_def} establishes definitions and notations used throughout the paper, and defines our 
main objects of interest.
The section \ref{sec_basic_properties} starts by providing some basic properties of graph normalization.
We then look in section \ref{sec_Kn} at the special case of iterative normalization of complete graphs.
The section \ref{sec_general} then walks the reader through our general results on 
arbitrary graphs. We illustrate our results through the special case of the path graph of order $3$, 
of which we can make a complete study and provide insightful graphics thanks to its small dimensionality.
The first subsection concerns the geometrical aspect of the graph normalization map and the 
next one the fixed points of graph normalization, their identification, their stability and their basins of attractions.
The section \ref{sec_kako} then provides experimental results on the approximation error 
of \MWIS by iterated graph normalization.
The section \ref{sec_SK} gets back to the special case of the assignment problem.
Before conclusion, the section \ref{sec_discuss} proposes a discussion.
Most proofs are reported to appendix in order to privilege fluidity.

\section{Definitions and notations}\label{sec_def}
Throughout the paper, we consider a simple undirected graph over a set of $n$ nodes 
$V=\{1\dots n\}$.
The adjacency matrix $A$ of such a graph is binary (only simple edges and thus $A\in \mathcal{M}_n(\{0,1\})$) 
and with a zero diagonal (no loops).
In the following, the set $V$ is most of the time implicit, and we  
assimilate the datum of a graph to its adjacency matrix $A$.

A weighted graph is a couple $G=(A,x)$, where $A$ is the adjacency matrix of $G$ and $x \in {\mathbb{R}^{+}}^{n}$ is a vector of nonnegative weights on the nodes of $G$.

For a node $i$, $\N{i} = \{ j \in $V$ | A_{ij} = 1\}$ denotes its neighborhood and 
$\DEG(i) = |\N{i}|$ its degree.

An independent set $S$ of $G$ (or stable set)  
is a set of nodes which are not adjacent in $G$, i.e. such that 
$\forall (i,j)\in S^2 : A_{ij} = 0$. We denote by $\IS(G)$ the set of the independent sets of a graph $G$.
An independent set $S$ is maximal if adding any new node to $S$ breaks the independency of $S$, i.e. $S\in \IS (G)$ and $\forall T\ne S | S\subset T : T \notin \IS (G)$, which means that 
any node not in $S$ is adjacent to at least one node of $S$.
$\MIS(G)$ denotes the set of the maximal independent sets of $G$.

Given a nonnegative vector $x = (x_1 \dots x_n) \in {\mathbb{R}^{+}}^{n}$, its support -
which we denote by $\SUPP(x)$ - is the 
set of the indices $i$ such that $x_i > 0$\footnote{Please note that this definition of the support of a vector is unrelated to the notion of a matrix with support which is involved in the convergence of the Sinkhorn-Knopp algorithm mentioned above.}.

If $S$ is a subset of $V$, its indicator vector is the binary vector 
$ind(S) = (i_1 \dots i_n) \in \{0,1\}^n$ such that $i_j = 1$ if $j \in S$ and $i_j = 0$ otherwise.
In the following, we will assimilate a subset $S$ and its indicator vector as 
$S = \SUPP(ind(S))$.

If $A$ is a matrix or a vector, we will write $A>0$ when all its components are strictly positive.
Throughout the paper, $\HM$ and $\HD$ denote resp. the Hadamard (element-wise) product and division of either matrices or vectors.

For a graph $G=(A,x)$ and a subset of its nodes $S\subset V$, we will denote by $G[S] = (A[S], x[S])$ the subgraph of $G$ induced by $S$. It adjacency matrix $A[S]$ is obtained from $A$ by only keeping the rows and columns which are in $S$ and similarly for the induced weight vector $x[S]$.

The Maximum Weight Independent Set (\MWIS) problem is the problem of finding an independent set 
of total maximum weight in $G=(A,w)$, i.e. a solution of:
\[
	\max_{x\in \IS(G)} w^t x
\]

\MWIS can be formulated as a binary linear program (BLP):
\begin{eqnarray*}
BLP &:&\max w^t x \\
&&x\in\{0,1\}^n \\
&&x_i + x_j \le 1 \quad \forall (i,j)\in S^2 | A_{ij} = 1
\end{eqnarray*}

or as a binary quadratic program (BQP):
\begin{eqnarray*}
BQP &:& \max w^t x \\
&&x\in\{0,1\}^n \\
&&x^t A x = 0
\end{eqnarray*}

\begin{definition} Let $G=(A,x)$ be a simple weighted graph such that $(A+I)x > 0$.
The \EM{normalization} of $G$ is the operation $\R$ on the weights $x$ 
of $G$ defined by:
\begin{equation}
\R(x) = x \oslash (A+I)x.
\end{equation}

In other words, the components of $\R(x)$ are given by:
\begin{equation*}
{\R}_i(x) = \frac{x_i}{x_i + \sum_j A_{ij}x_j}
\end{equation*}

\end{definition}

Remark that graph normalization doesn't change the structure of the graph but just its weights.
It is actually a vector function from ${\mathbb{R}^+}^n$ to ${\mathbb{R}^+}^n$ parametrized 
by the structure of the graph, i.e. the adjacency matrix $A$.
It would have been maybe more accurate to call it something like ``graph-based vector normalization'' 
but we have privileged concision.
Also remark that the underlying graph upon which a weight vector $x$ is normalized is absent from 
the notation $\R(x)$. We have privileged an uncluttered notation, leaving the underlying graph 
implicit. 

Normalization is only defined on graphs $G=(A,x)$ such that $(A+I)x > 0$, 
in order to be able to divide each weight $x_i$ by $x_i+\sum A_{ij} xj$. 
We say that the weighted graphs $G=(A,x)$ such that $(A+I)x>0$ are \EM{normalizable}.
We denote by $\NORMG$ the set of normalizable graphs.
For a given graph structure $A$, we also denote by 
$\NORMG_A = \{ x \in {\mathbb{R}^+}^n | (A+I)x>0\}$ the set of weight vectors such that $G= (A,x)$ 
is normalizable.

Normalizable graphs can be characterized by introducing the notion of \EM{density} of a subset of the 
nodes of a graph:

\begin{definition}\label{def_density}
Given a graph $G$ with vertex set $V$ and adjacency matrix $A$, and a subset of its 
nodes $S\subset V$, the density of $S$ in $G$ is defined by 
\begin{eqnarray}
\DEN{S}  &=& \min_{i \notin S} | \N{i} \cap S | \\
&=& \min_{i \notin S} \sum_{j \in S} A_{ij}
\end{eqnarray}
\end{definition}

This definition is illustrated in figure \ref{fig_example_density}.

\begin{figure}
    \centering
    \includegraphics[width=0.8\linewidth]{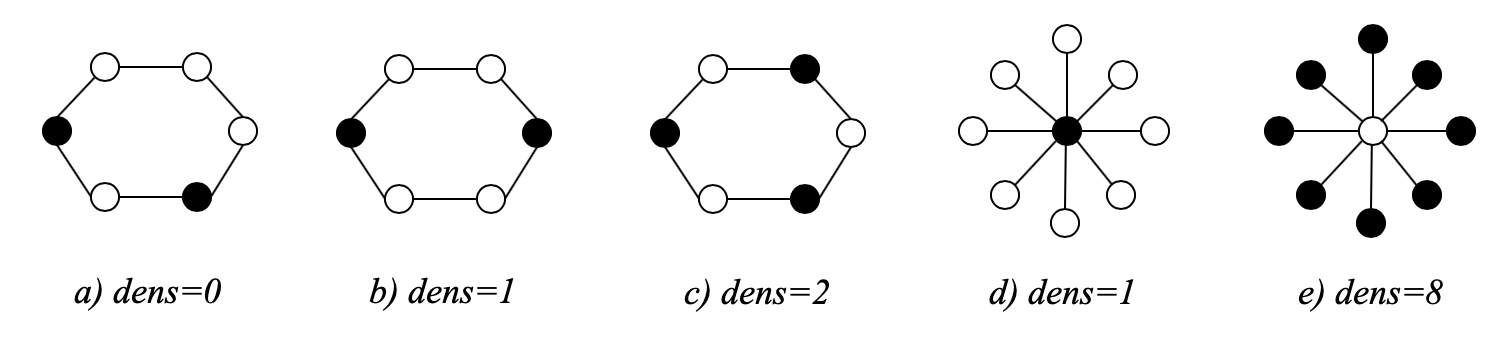}
    \caption{Examples of densities of a set in a graph. 
    In each graph, the black nodes represent the elements of the set $S$ considered. 
    The density is the minimum number of neighbors belonging to $S$ taken over the nodes which do not belong to $S$. 
    Remark that all the sets represented in this figure are independent and maximal except for a).}
    \label{fig_example_density}
\end{figure}

Obviously, $G=(A,x)$ is normalizable if and only if $\DEN{\SUPP(x)}>0$, i.e. 
if and only if every zero weight node has at least one non zero neighbor.
We will see below that this notion of density is central in graph normalization and also appears in 
the condition of stability of the binary fixed points of GN.

We say that a weighted graph $G=(A,x)$ is normalized when $\R(x) = x$, i.e. when $x$ is a fixed point of normalization on $G$.

We also consider a variant of GN in which after normalization, the weights are transformed by a non-linear activation function\footnote{Please note that normalization is already a non-linear operation, as it involves an element-wise division.}. As we will detail below, the role of this activation is to ensure convergence to binary fixed points 
and in particular to indicators of independent sets.
In order to play this role, the activation function must verify certain properties whose importance will become 
clear later on.
\begin{definition} \label{def_act_fun}
In our context, an \EM{activation function} is a real-valued function $h: [0,1]\rightarrow [0,1]$ which verifies:
\begin{enumerate}
\item $h(0)=0$ and $h(1)=1$.
\item Either:
\begin{enumerate}
\item $h$ is strictly convex on $[0,1]$, or
\item $h(1/2)=1/2$ and $h$ is strictly convex on $[0,1/2]$ and strictly concave on $[1/2,1]$. 
\end{enumerate}
\end{enumerate}

\end{definition}

To fix ideas, the following activation functions verify the above properties. Their graphs are illustrated in figure \ref{fig_activations}. Please note that other valid activation functions can be considered.
\begin{enumerate}
\item \EM{Power activation}: $p_{a,t}(x) = k_1(x+t)^a + k_2$  with $a > 1$ and $t\ge 0$, where $k_1$ and $k_2$ are $2$ coefficients which are set so that $p_{a,t}(0)=0$ and $p_{a,t}(1)=1$. $t$ is a shift from the origin, whose role will appear clear later on.
\item \EM{Sigmoid activation}.  Let $t_a(x) = \frac{1}{1+e^{-a x}}$ be the standard sigmoid function of parameter $a$. We use a translated and rescaled version $s_a$ of $t_a$ such that $s_a(0)=0$ and $s_a(1)=1$, defined by:
\begin{equation}
s_a(x) = \frac{1}{2}\frac{t_a(x-\frac{1}{2})-\frac{1}{2}}{t_a(\frac{1}{2})-\frac{1}{2}}+\frac{1}{2}
\end{equation}
$s_1$ is very close to the identity and as $a$ increases, $s_a$ converges to 
a Heavyside step function centered on $\frac{1}{2}$.

\end{enumerate}

\begin{figure}
    \centering
    \begin{tabular}{ccc}
    \includegraphics[width=0.4\linewidth]{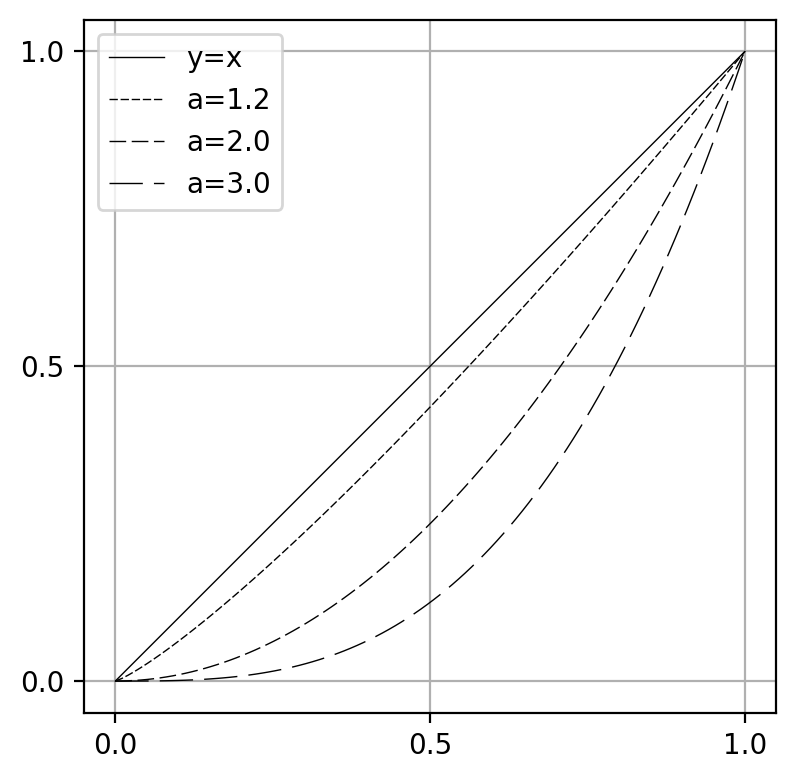} &
    \includegraphics[width=0.4\linewidth]{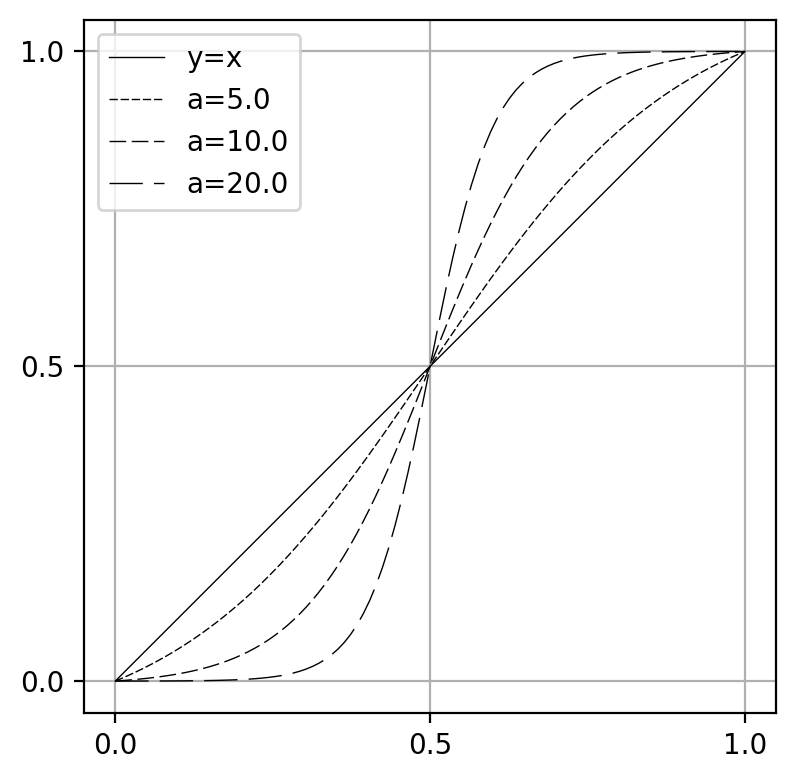}\\
    (a) Power activation &
    (b) Sigmoid activation 
    \end{tabular}
    \caption{Possible activation functions matching the conditions of definition \ref{def_act_fun}.
	}
    \label{fig_activations}
\end{figure}

\begin{definition} Given an activation function $h$ defined as above, 
the \EM{activated normalization} of $G=(A,x)$ is defined by:
\begin{equation}
\Rh(x) = h\left( \R(x) \right)
\end{equation}
where $h$ is applied component-wise, i.e. :
\begin{equation*}
{\Rh}_i(x) = h\left( \R_i(x) \right)
\end{equation*}
\end{definition}

In this paper, we are interested in the iteration of graph normalization (IGN), 
i.e. in the discrete dynamical system:
\begin{eqnarray}
x^0 &=& x \\
x^{k+1} &=& \R_h(x^k)  \quad=\quad \R_h^{k+1} (x)
\end{eqnarray}

where $\R_h^{k+1} = \R_h \circ  \R_h \dots  \circ \R_h$ represents $k+1$ compositions of $\R_h$.

\section{Basic properties of graph normalization}\label{sec_basic_properties}
The following properties have simple proofs, which we leave to the reader:
\begin{sloppypar*}
\begin{enumerate}
\item $\Rh$ maps the weights into $[0,1]$, i.e. $\forall G=(A,x) \in \NORMG : \R_h(x) \in [0,1]^n$.

\item Positiveness of the weights is conserved by normalization, i.e. ${\SUPP(\Rh(x)) = \SUPP(x)}$ (or componentwise: ${\R_h}_i(x) = 0 \Leftrightarrow x_i = 0$). \label{prop_positivity}

\item $\R_h$ is invariant by rescaling the weights of $G$, 
i.e. $\forall G=(A,x) \in \NORMG, \forall \alpha>0 : \R_h(\alpha x) = \R_h(x)$. 

\item $\R_h(x)$ only depends on the ratios of the weights between neighboring nodes the graph. Indeed, if $x_i=0$ then ${\R_h}_i(x)=0$ and if $x_i>0$, then 
\[
{\R_h}_i(x) = h \left(\frac{1}{1 + \sum A_{ij}\frac{x_j}{x_i}}\right).
\]

\item $\R$ commutes with the automorphisms of $G$ : for any permutation matrix $M$ such that $MA = AM$ : $\R(Mx) = M \R(x)$. This is not true for $\R_h$ if $h$ is nonlinear.

\item If $x$ is an independent set of $A$ then $G=(A,x)$ is normalizable if and only if $x$ is maximal. 
This follows from the fact that non maximal independent sets have a null density.

\item The normalization of $G=(A,x)$ coincides on $\SUPP(x)$ with the normalization of $G[\SUPP(x)]$. 
In words, the dynamics of $\Rh$ on the support of $x$ 
is independent of the connectivity of the support with the other nodes. \label{prop_norm_induced_supp}

\item A graph $G=(A,x)$ is normalized, i.e. is a fixed point of $\R$, if and only if for any node $i$ either $x_i = 0$ or $\sum A_{ij}x_j + x_i = 1$. Remark that normalization to $1$ on neighborhoods must
only hold at nodes of nonzero weight, which is a critical aspect of $\R$ as will be clarified below.

\end{enumerate}
\end{sloppypar*}

The properties \ref{prop_positivity} and \ref{prop_norm_induced_supp}
(zeros of $x$ are stable and do not influence the dynamics on the rest of the graph) 
imply that whenever $G=(A,x)$ is normalizable, 
the dynamics of $\Rh$ can be studied on the graph induced by the 
support of $x$.
Equivalently, we can assume that $x$ is strictly positive.  

Similarly, the dynamics of normalization on the different connected components of $G$ 
are independent from each other.
We thus can assume that $G$ is connected.

Following these last two properties, we will restrict our study of graph normalization 
to connected graphs with weights belonging to the domain $(0,1]^n$, 
which we denote by $\DOM$.
Please note however that iterative normalization, if it converges, 
can converge to weight vectors belonging to the closure of $\DOM$, hence which have null components 
(and it actually does in general, which is the main point of the paper).

\section{Complete graphs}\label{sec_Kn}
We start by looking at the particular case of normalization of complete graphs.

If $G=(A,x)$ is the complete graph $K_n$ with $n$ nodes, then $\forall i\in V$:
$\R_i(x) = x_i / \sum_{1}^{n} x_j$, which is stable in $1$ iteration. 
$\R(x)$ is normalized in the sense that its components sum up to $1$.
$\R$ thus corresponds to the standard normalization of the vector $x$ to a stochastic vector, 
which is obtained by projecting $x$ onto the unit $(n-1)$-simplex 
$\Delta_{n-1}$.

Now, if a suitable non-linear activation function is introduced, 
then the iterations converge to the indicator vector of $\arg\max x$:

\begin{restatable}{theorem}{thKnconvergence}\label{ref_thKnconvergence}
If $G=(A,x)$ is a complete graph and $x$ has a unique nonzero maximum component, then 
for any non-linear activation verifying the conditions of definition \ref{def_act_fun},
$\R^k_h(x)$ converges to the indicator vector of $\arg\max x$.
\end{restatable}

The figure \ref{fig_Kn_examples} provides some examples of the dynamics of 
graph normalization on complete graphs.

Note that the theorem \ref{ref_thKnconvergence} can be easily extended to the case when 
$x$ has multiple maximal components, i.e. when $m = |\arg\max x| > 1$. 
In this case, all the maxima follow the same dynamics, are strictly increasing after the 
first normalization, and converge to $h(1/m)$, while the other weights converge to $0$.
Hence in this case the limit weight vector is not binary, 
but its support is still $\arg\max x$.

\begin{figure}[!ht]
    \centering
    \begin{tabular}{ccc}    
    \includegraphics[width=0.25\linewidth]{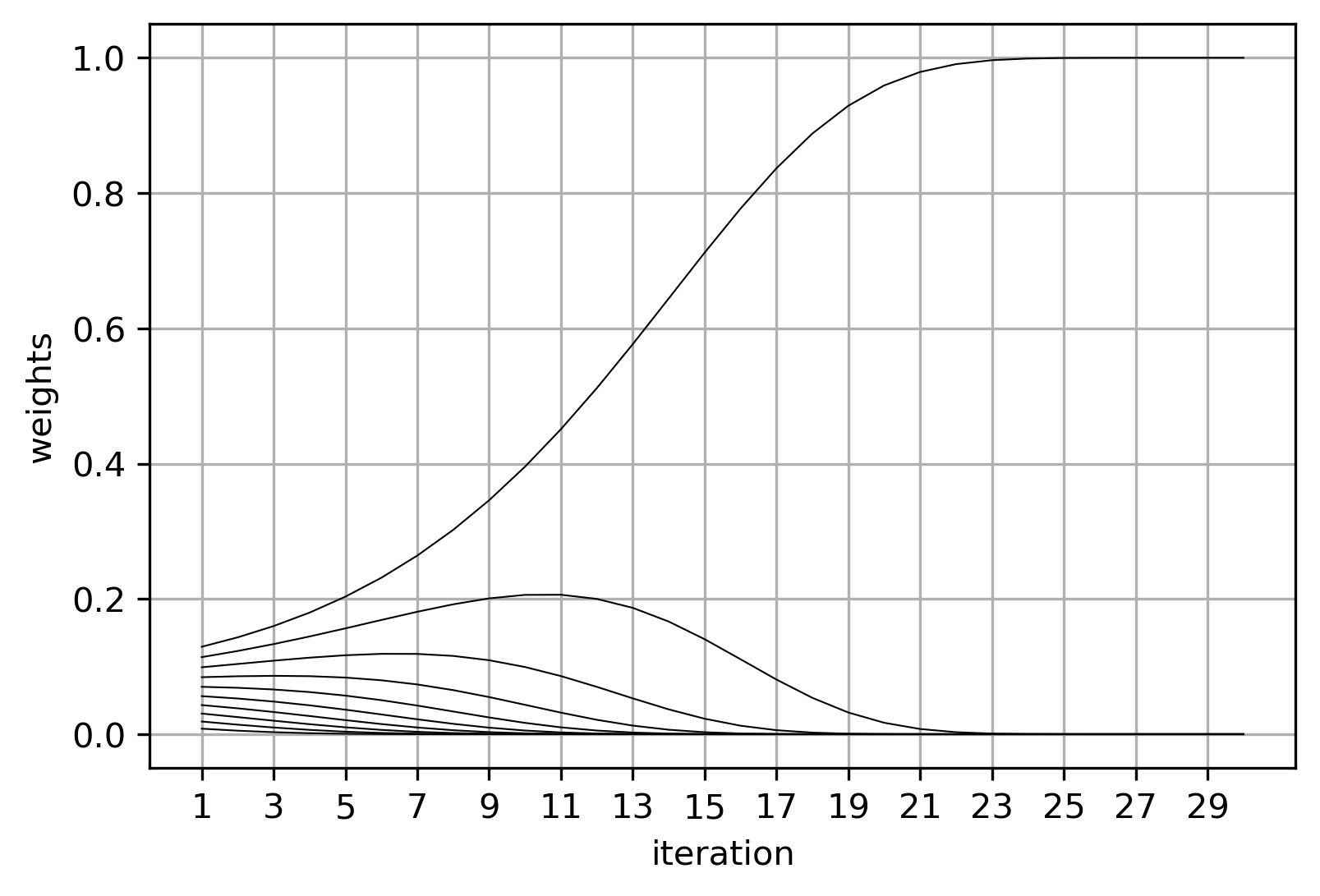} &
    \includegraphics[width=0.25\linewidth]{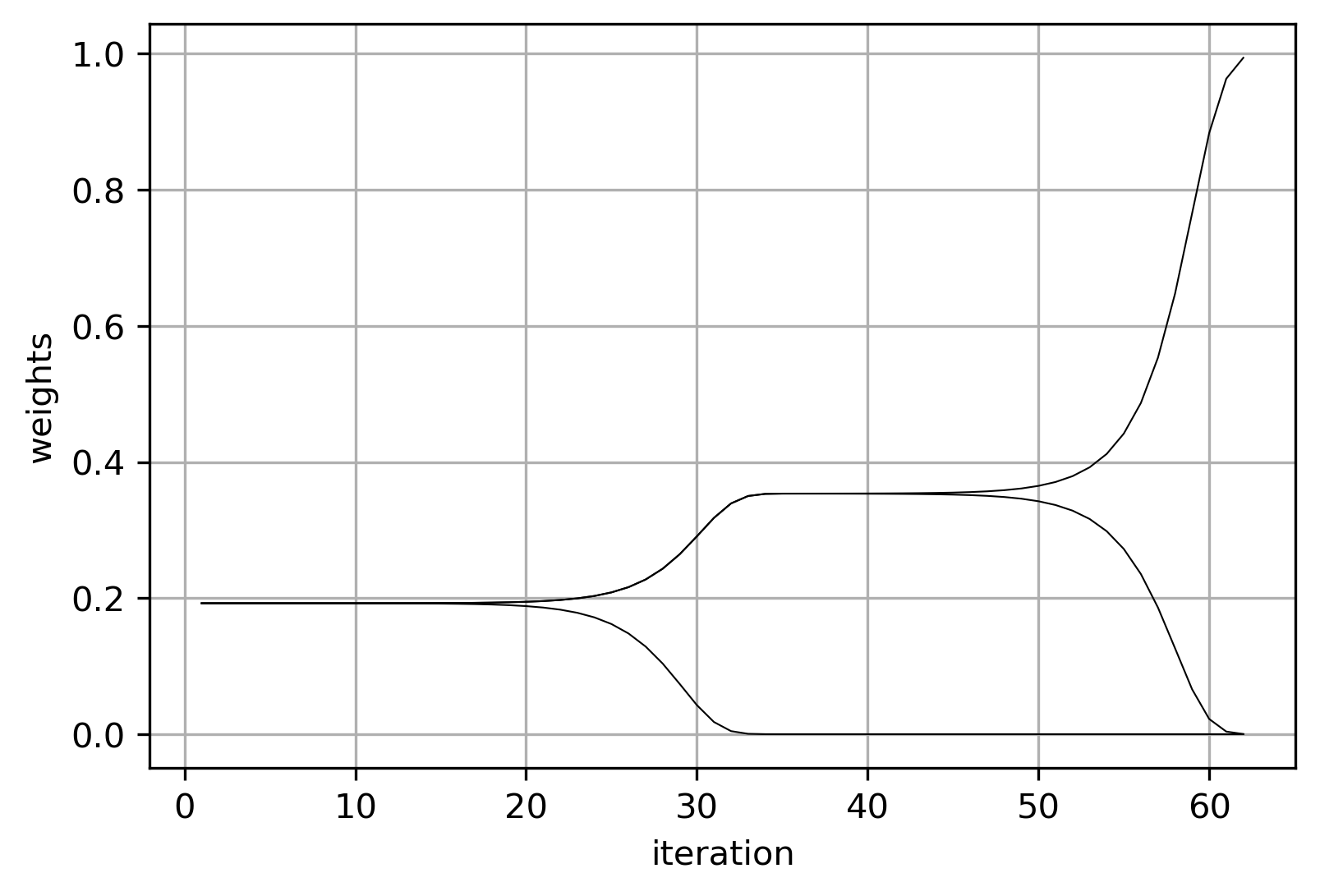} &
    \includegraphics[width=0.25\linewidth]{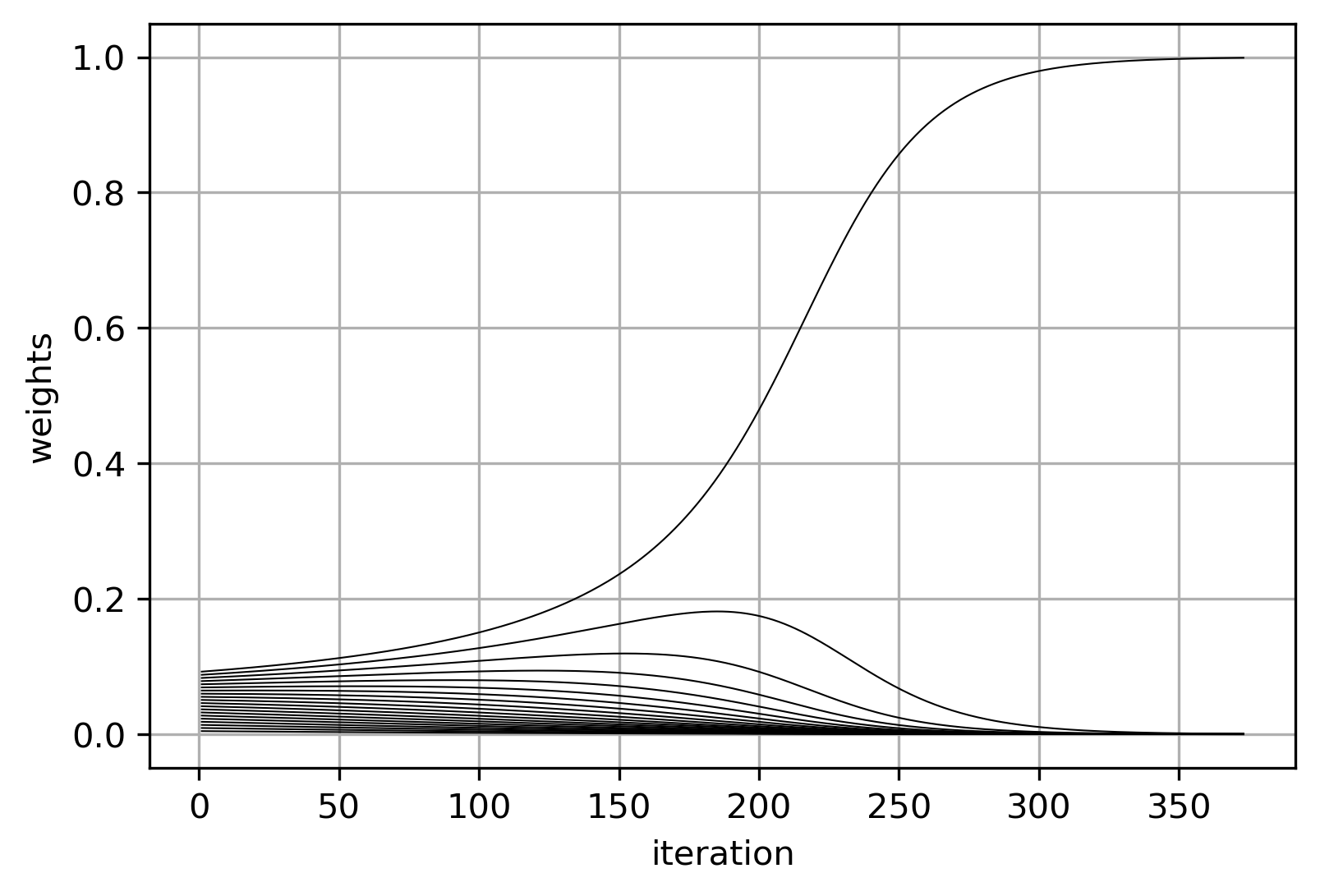} \\
    (a) \parbox{0.25\linewidth}{\small$K_{10}$. \\ $x=(10, 9, \dots, 1)$. \\ $h=p_{1.2, 0}$} &
    (b) \parbox{0.3\linewidth}{\small$K_{3}$. \\ $x=(1, 1-10^{-10}, 1-10^{-5})$. \\$h=p_{1.5, 0}$} &
    (c) \parbox{0.25\linewidth}{\small$K_{20}$. \\ $x=(20, 19, \dots, 1)$. \\ $h=s_{1}$} \\
    \includegraphics[width=0.25\linewidth]{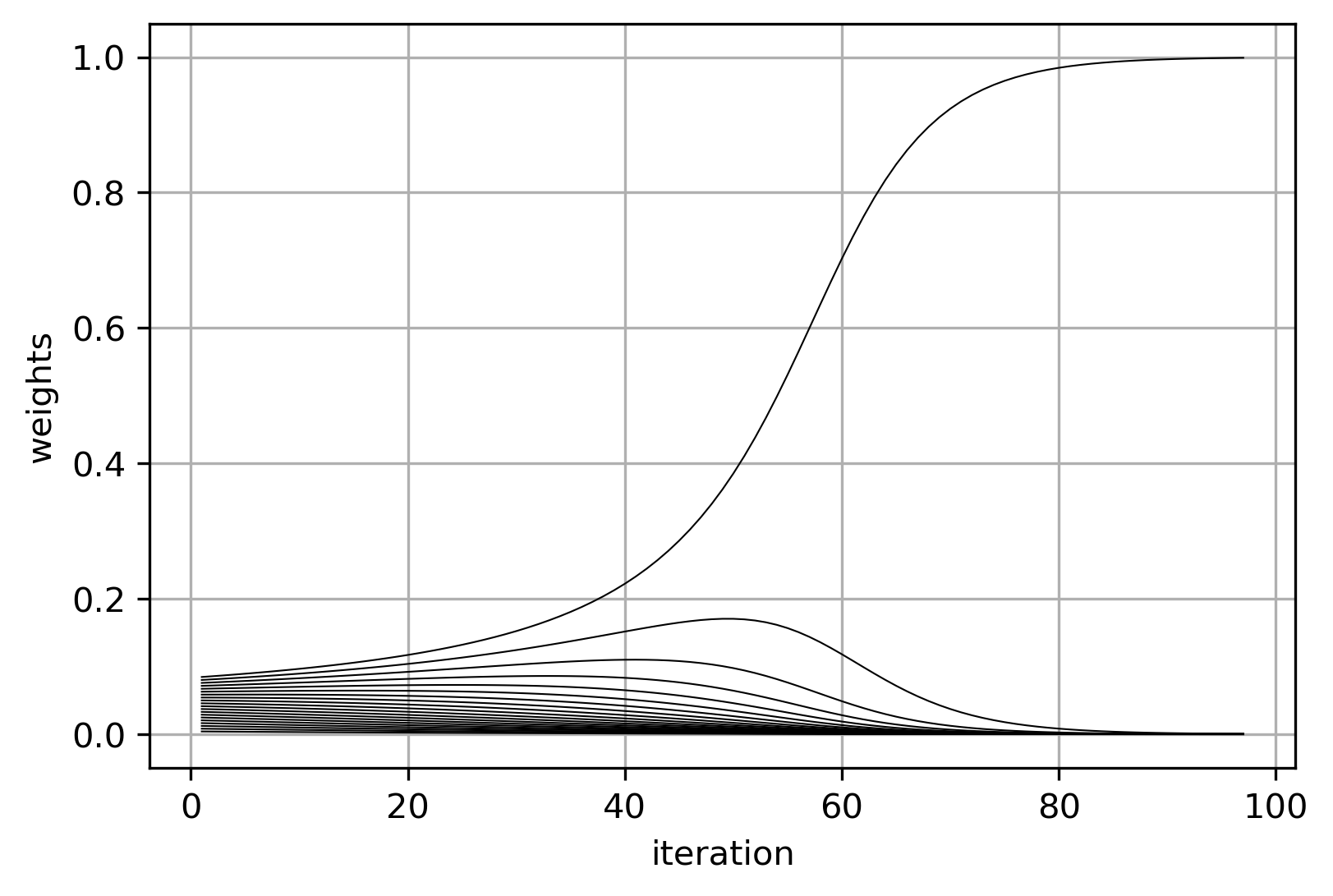} &
    \includegraphics[width=0.25\linewidth]{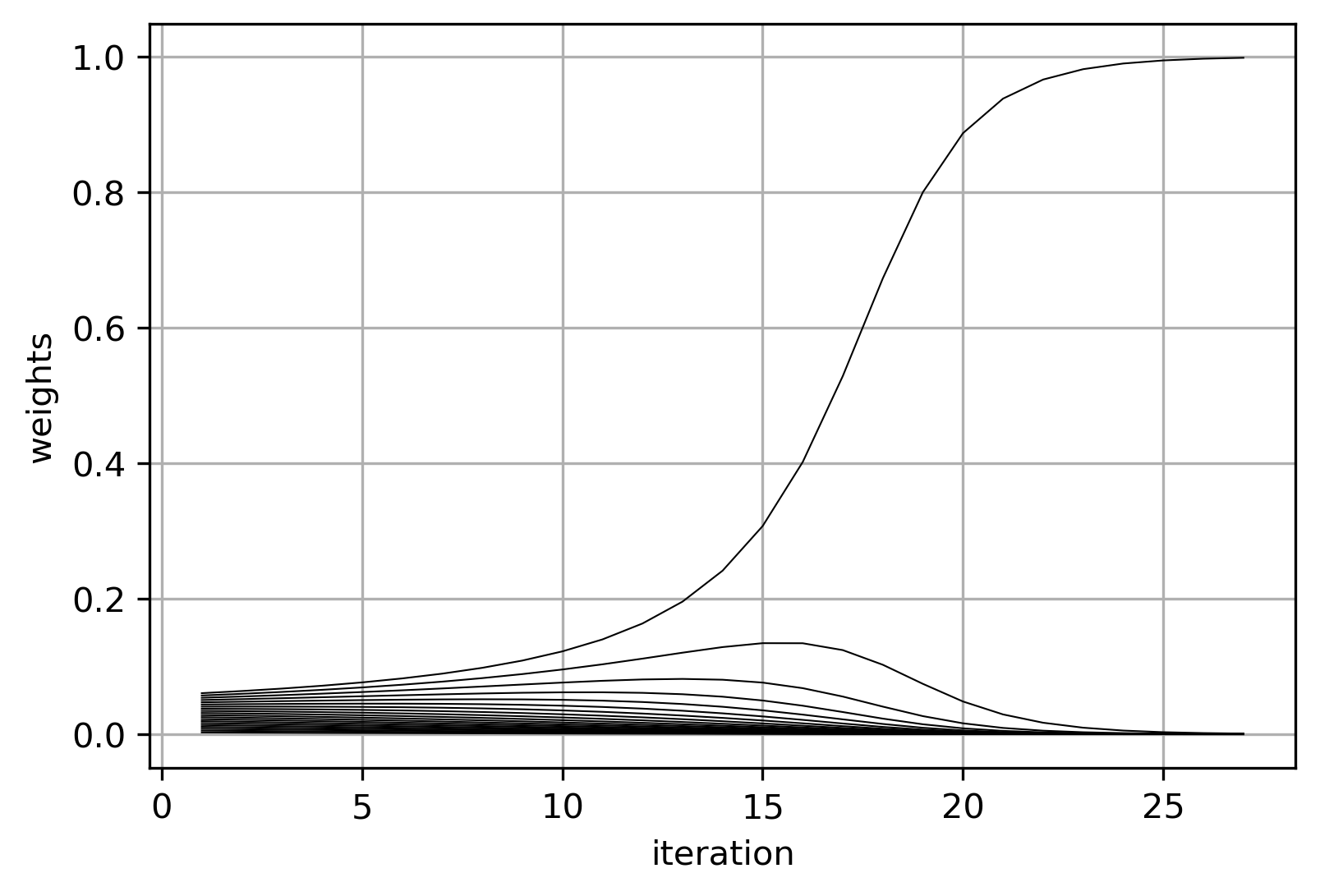} &
    \includegraphics[width=0.25\linewidth]{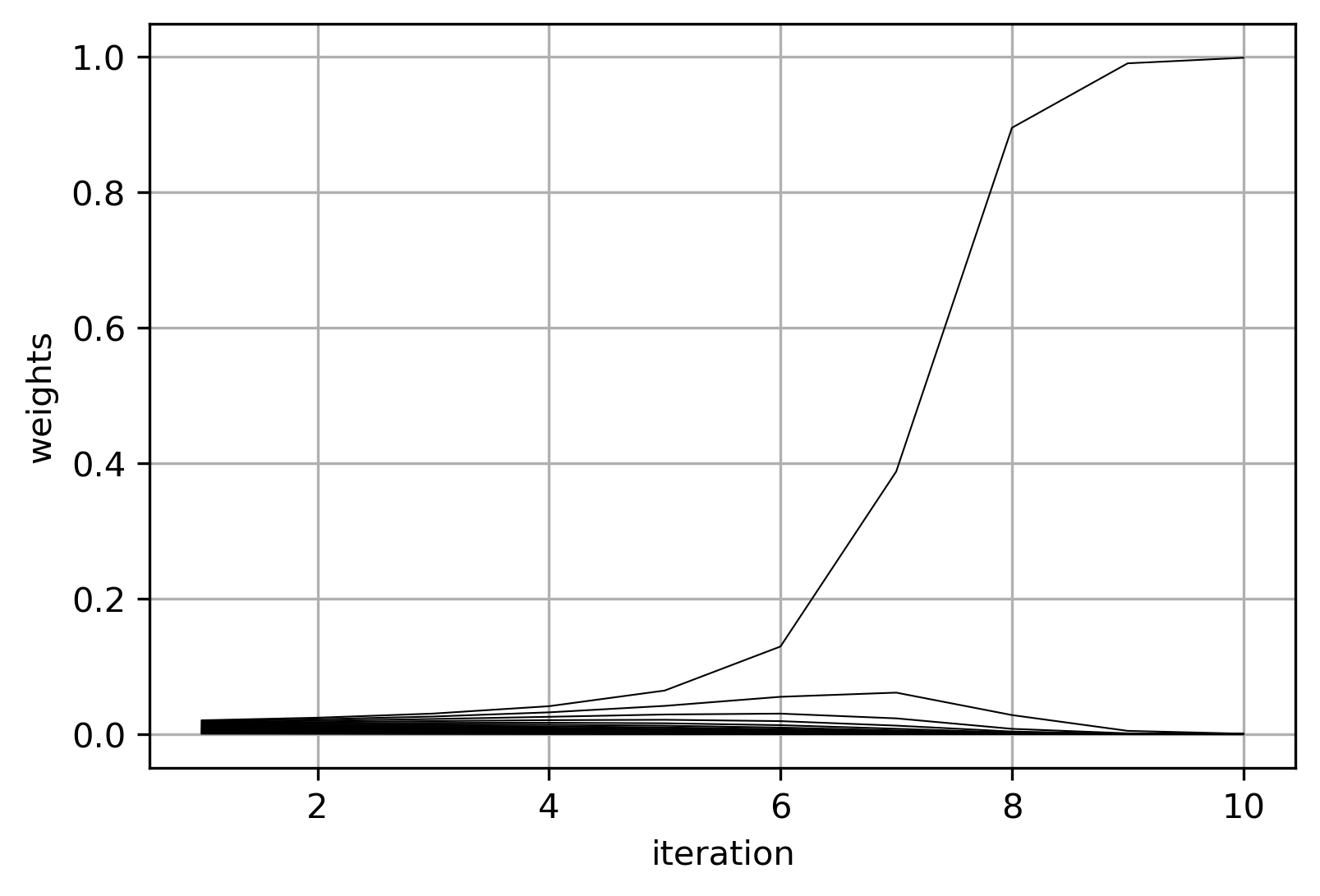} \\
    (d) \parbox{0.25\linewidth}{\small$K_{20}$. \\ $x=(20, 19, \dots, 1)$. \\ $h=s_{2}$} &
    (e) \parbox{0.28\linewidth}{\small$K_{20}$. \\ $x=(20, 19, \dots, 1)$. \\ $h=s_{4}$} &
    (f) \parbox{0.25\linewidth}{\small$K_{20}$. \\ $x=(20, 19, \dots, 1)$. \\ $h=s_{8}$}

    \end{tabular}
    \caption{Examples of dynamics of activated normalization 
    of complete graphs. Each curve represents the evolution of the weight of a node through normalization iterations.
    Remark that the maximum weight is strictly increasing (which is shown in the proof of the theorem \ref{ref_thKnconvergence}) 
    but that the other weights are not necessarily decreasing at the first iterations (except for the minimum weight(s)).
    (a) Illustrates that even a very slight non-linearity (a power of $1.2$) is sufficient to binarize the input vector in tens of iterations. (b) Shows the dynamics for very small relative differences between the input weights. 
    (c)-(f) Show that increasing the amount of non-linearity of a sigmoid activation speeds up convergence (the iterations were stopped when $x_1 > 0.999$) but that the global profile of the dynamics remain very similar. 
    }
    \label{fig_Kn_examples}
\end{figure}

The behavior of iterative normalization on complete graphs is reminiscent of the Softmax operator:
\[
{Softmax_{\tau}}_i (x) = \frac{e^{x_i / \tau}}{\sum e^{x_j / \tau}}
\]

which converges to the one-hot vector 
of $\arg\max x$ as $\tau$ goes to $0$, again if there are not ties \cite{gold1996softmax}.
The fundamental difference is that Softmax provides $\arg\max x$ as the limit of a continuous 
parameter $\tau$, which is not accessible in practice, whereas $\IGN$ provides it 
as the limit of the discrete iterations of a stable dynamical system.
We will find a similar parallel between $\IGN$ applied to the assignment problem 
and the Softassign algorithm in section \ref{sec_SK}.

\section{General graphs}\label{sec_general}
\subsection{The path graph of order $3$}
A path graph $P_n$ with $n$ nodes 
has an adjacency matrix of the form:
\[
A = 
\begin{bmatrix} 
0 & 1 & 0 & \dots & 0 & 0 \\
1 & 0 & 1 & 0       & \dots & 0 \\
0 & 1 & 0 & 1       & \dots & 0 \\
\vdots &&&&& \vdots \\
0 & 0 & \dots & 0 &  1 & 0 \\
\end{bmatrix}
\]
which represents a graph made up of a line of nodes:
\begin{center}
\includegraphics[width=1.5in]{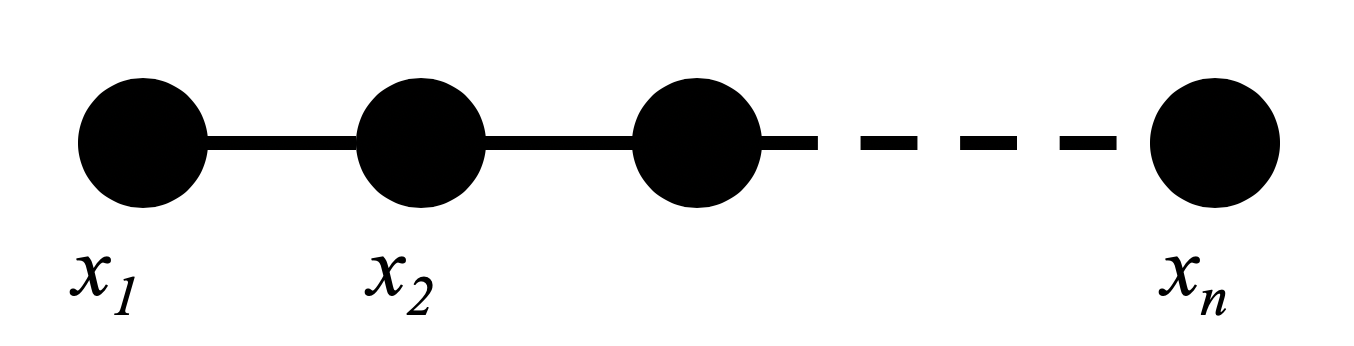}
\end{center}

$P_3$ is the smallest connected graph which is not complete hence the smallest interesting graph for normalization.
The study of iterated normalization on $P_3$ already illustrates 
most of the general properties of graph normalization.
Furthermore - as we will see in a moment - 
the dynamics of the system actually lie on a surface, 
and so can be easily visualized, which helps getting insight into the behavior of graph normalization.
In this section, we will use the special case of $P_3$ to 
walk the reader through our general results on graph normalization on arbitrary graphs.

\subsection{Geometry of the graph normalization map}
Let's first look at the case of $P_3$.
As we are in $3$ dimensions, we use the notation $(x,y,z)$ for the $3$ weights on $P_3$.
Normalization of $P_3$ is the function $(u,v,w) = \R((x,y,z)) $ defined by:
\[
\begin{array}{lcl} 
u &=& x / (x+y)\\ 
v &=& y / (x+y+z)\\
w &=& z / (y+z)
\end{array}
\]
It is defined everywhere on ${\mathbb{R}^+}^3$ except on the $x$ axis (the line such that $y=z=0$), the $z$ axis ($x=y=0$) and at the origin $(0,0,0)$, where the graph is not normalizable.

Our colleague Paul Munger 
remarked that the image of the unit cube $(0,1]^3$ by $\R$ 
is an algebraic variety $V$ of $(0,1]^3$.
Indeed, rewriting the normalization equations gives:
\[
\begin{array}{lcl} 
u (x+y) - x &=& 0\\
v (x+y+z) - y &=& 0\\
w (y+z) - z &=& 0
\end{array}
\]
which is the intersection of $3$ quadratic polynomials in the variables $u, v, w, x, y, z$ hence an algebraic variety in 
$\mathbb{R}^6$. 
The image $V$ of the unit cube by normalization is then the intersection of this variety with the subset of 
$\mathbb{R}^6$ spanned by $\{u, v, w\}$, which is still an algebraic variety.
 
Using Groebner basis, Paul computed that $V$ is actually given by the single polynomial equation:
\[
uvw + uw - u - v - w + 1 = 0
\] 

This equation represents a \emph{surface} in $(0,1]^3$ which is shown in figure \ref{fig_taco}.
Paul called it the ``Taco''. 

\begin{figure}[!ht]
    \centering
    \includegraphics[width=0.5\linewidth]{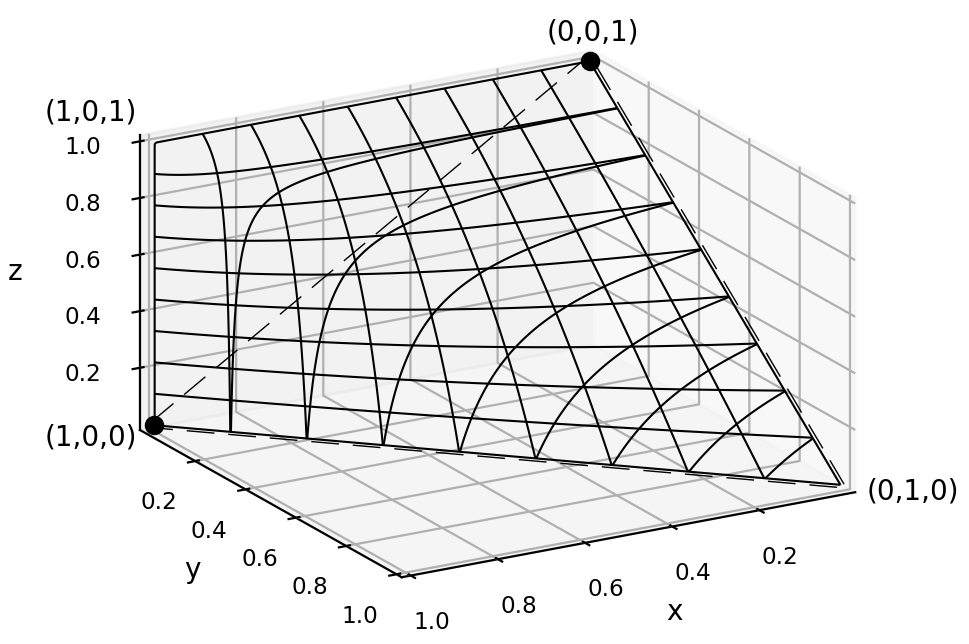} 
    \caption{The ``Taco''. This surface is the image of the unit cube by normalization of the path graph $P_3$.
    }
    \label{fig_taco}
\end{figure}

At the first iteration, any normalizable weight vector is projected on the Taco and then 
evolves on this surface. Remark that the Taco is sitting on the triangle $\Delta_2 = ((1,0,0),(0,1,0),(0,0,1))$ 
which is dashed in figure \ref{fig_taco}. 
$\Delta_2$ is the unit 2-simplex - or probability 2-simplex - i.e. the triangle such that $u+v+w=1$, 
hence the surface onto which the stochastic 3-vectors live.

Points on $\Delta_2$ are in a $1$ to $1$ mapping with the lines of ${\mathbb{R}^+}^3$, 
i.e. the points of the $2$-projective space $\mathbb{P}_2(\mathbb{R}^+)$. 
Now, the Taco also intersect any line of ${\mathbb{R}^+}^3$ only once, hence 
can be reparametrized as a function $d=f(a,b)$, where $(a,b)$ parametrizes the points on $\Delta_2$, or 
equivalently the lines of ${\mathbb{R}^+}^3$, and $d$ is the $L_1$ norm of the point on the surface, i.e. $u+v+w$.
Recalling that $\R$ is invariant by rescaling its input, i.e. that $\R(kx) = \R(x)$ for any positive $k$, 
it means that $\R$ is actually a mapping from $\mathbb{P}_2(\mathbb{R}^+)$ (minus the non normalizable points) to $\mathbb{P}_2(\mathbb{R}^+)$, i.e. a mapping from lines to lines.
For $P_3$, the non normalizable projective points correspond to the points $(1,0,0)$ and $(0,0,1)$ on $\Delta_2$.
They are represented with black dots in figure \ref{fig_taco}. 
Remark that these two non normalizable points, together with $(0,0,0)$, are the 
indicators of the three independent sets of $P_3$ which are not maximal.

For a general graph $A$, denote by 
$\IMN{A} = \{ \R(x); x\in \NORMG_A\}$ the image by normalization on $A$ of the set the 
normalizable vectors.
The following theorem shows that the above property of the Taco is general.

\begin{restatable}{theorem}{thprojective}\label{ref_thprojective}
For any graph $A$ of size $n$, 
$\IMN{A}$ is a hypersurface of 
$\mathbb{R}^n$ which intersects each line of $\mathbb{R}^n$ at most once.
\end{restatable}

Hence $\R$ is a mapping from $\mathbb{P}_{n-1}(\mathbb{R}^+) \backslash \NORMG_A$ to $\mathbb{P}_{n-1}(\mathbb{R}^+)$, i.e. a mapping from lines to lines.
The theorem \ref{ref_thprojective} also implies that $\IMN{A}$ can always be expressed as the zero-set 
of a unique polynomial equation in $n$ variables.

Furthermore, the next theorem shows that if the graph is a tree then 
the normalization mapping is injective for the interior points of the simplex, 
which means that $\R$ is reversible inside the simplex.

\begin{restatable}{theorem}{thtreeinjective}\label{ref_thtreeinjective}
For any tree $T$ of size $n$ and positive weights $x$, graph normalization of $(T, x)$ is 
an injection from $\mathbb{P}_{n-1}({\mathbb{R}^+}^*)\backslash \NORMG_T$ to $\IMN{T}$,
i. e. $\forall (x,y) \in{\NORMG_T}^2$ such that $x>0$ and $y>0$ : 
$\R(x) = \R(y) \Rightarrow \exists k>0 | y=kx$. 
\end{restatable}

We don't know whether $\R$ is still injective for general graphs with cycles. 
Remark that $P_3$ is a tree hence that its dynamics are reversible in the interior of $\Delta_2$. 

\subsection{Fixed points}
\subsubsection{Identification}

Getting back to the Taco, 
one easily verifies that apart from the $3$ non-normalizable/non-maximal binary points mentioned above,
$\R$ has $2$ binary fixed points on $P_3$: $(1,0,1)$ and $(0,1,0)$ which are the 
indicators of the $2$ maximal independent sets of $P_3$.
The proposition \ref{refthbinaryfparemis} shows that for any graph the only binary fixed points of 
normalization are its maximal independent sets.

\begin{restatable}{proposition}{thbinaryfparemis}\label{refthbinaryfparemis}
For any graph $G$ and any activation function $h$, 
a normalizable binary vector $x$ is a fixed point of $\Rh$ if and only if it is a 
maximal independent set of $G$.
\end{restatable}

The other fixed points on $P_3$ are of the form $(a,1-a,0)$ and $(0,a,1-a)$ with $a\in(0,1[$ and 
correspond to the points on the segments $((1,0,0),(0,1,0))$ and $((0,1,0),(0,0,1))$.

These fixed points belong to a general class of fixed points of graph normalization which we call
``complete clusters''.
\begin{definition} For a weighted graph $G=(A,x)$, the subgraphs induced by the 
connected components of the support of $x$ are called the \EM{clusters} of $G$.
\end{definition} 
For example, for $A=P_3$ and $x=(a,1-a,0)$, we have $\SUPP(x)=\{1,2\}$ 
and the nodes $1$ and $2$ are connected in $P_3$ 
so that the connected components of the support of $x$ are $\{\{1,2\}\}$ 
(if the nodes would have been disconnected this set would have been $\{\{1\},\{2\}\}$).
The unique cluster is thus the induced subgraph $G[\{1,2\}]$ 
which is $(A[\{1,2\}], x[\{1,2\}]) = (\begin{bmatrix}0 & 1 \\ 1 & 0\end{bmatrix}, (a, 1-a))$. 
Now this graph is the complete graph $K_2$ and it is normalized (as $a+(1-a)=1$) hence 
it is stable by normalization.

Obviously, in general, for any graph, a normalizable vector $x$ is a fixed point of normalization 
if and only if it is made up of fixed clusters.
Clearly also, any normalized complete cluster (a cluster $C=(K_d, x)$ such that $\sum x_i = 1$) is fixed.
Remark that the maximal independent sets of a graph are the 
fixed points made up of clusters which are singletons of weight $1$,
that is normalized complete graphs of dimension $1$.

Now, there are many other fixed clusters. 
In general, a cluster is fixed if and only if it is a connected graph $C=(A,x)$ such that $x>0$ 
verifies $x = x\HD (A+I)x$. 
This is equivalent to $x\HM (A+I)x = x$ 
and as $x>0$, we can divide each equation by $x_i$, hence $x$ is a solution of 
\begin{equation}\label{eqn_fixed_cluster}
(A+I)x = \ONE \quad\textrm{with}\quad x>0
\end{equation}

This equation simply expresses that the cluster is normalized in the sense that 
for each node, its weight added to those of its neighbors sum up to $1$.
It is critical to mention here that the condition of normalization to $1$ over neighborhoods 
only holds for nodes with \emph{strictly positive} weights, 
hence that this equation doesn't hold in general on every node of a fixed point,
because fixed points have in general many nodes with zero weight. 
All the nodes of a cluster are normalized to $1$ because we precisely required clusters 
to be only made of nodes with strictly positive weights.

One verifies that there is a systematic solution to this equation when the cluster is regular, i.e. such that 
the degrees of its nodes are all equal to a number $d$:

\begin{proposition} 
Any $d$-regular cluster $C=(A,x)$ such that $\forall i : x_i = \frac{1}{d+1}$ is fixed.
\end{proposition}

Note however that a $d$-regular cluster can also be fixed without having all its weights identical.
This is for example the case for the cycle of size $6$, 
which is $2$-regular, 
when it is weighted by positive vectors of the form $(a, a, b, a, a, b)$ verifying $2a+b=1$.

Nevertheless, we have that
\begin{restatable}{proposition}{propfixedregular}\label{ref_propfixedregular}
The only fixed clusters with all-identical weights are the regular clusters.
\end{restatable}

Working in the other direction, some graphs or some graphs containing certain sub-structures cannot 
be fixed clusters.
\begin{proposition}
A fixed cluster cannot contain a node of degree $1$, except if it is $K_2$. In particular, trees other than $K_1$ and $K_2$ cannot be fixed clusters.
\end{proposition}
\begin{proof}
Consider a cluster $C$ with a dangling node of weight $a$, its unique neighbor of weight $b$ and 
the nodes connected to this neighbor, whose sum of weights is $c$. 
As $C$ is a cluster $a>0$, $b>0$ and $c>0$.
The normalization condition at the dangling node is $a+b=1$ and 
the equation for its neighbor is $a+b+c=1$ hence $c=0$ which contradicts the fact that $C$ is a cluster.
\end{proof}

Other structures are impossible for a cluster, such as a dangling triangle attached by a bridge to the rest of the graph. 
We let the reader verify it.

This restricts the set of possible fixed clusters, however, 
finding general solutions to the equation \ref{eqn_fixed_cluster} is 
not easy. 
For a given graph structure $A$, knowing whether a solution $x$ exists is a linear feasibility problem however 
with a \emph{strict} positivity condition on $x$.
One way to solve such a problem numerically is to transform it into the linear program:
\begin{eqnarray}
\min \ONE^t y \label{eqn_fixed_cluster_LP}\\
(A+I)x + y &=& \ONE \nonumber\\
x &\ge& 0 \nonumber\\ 
 y &\ge& 0 \nonumber
\end{eqnarray}
Where $y$ is an additional variable of same size than $x$. 
One sees that this LP is always feasible by setting $x=0$ and $y=\ONE$.
If the equation \ref{eqn_fixed_cluster} has no positive solution then 
the LP \ref{eqn_fixed_cluster_LP} only has solutions in which $y$ has at least a strictly positive component. 
On the other hand, if \ref{eqn_fixed_cluster} has a positive solution $x^*$, 
then the LP \ref{eqn_fixed_cluster_LP} will be able to attain an objective of $0$ by setting $y=0$ and $x=x^*$.
As a conclusion, if the LP has a solution such that $y=0$ and $x>0$ then it is a solution of equation 
 \ref{eqn_fixed_cluster} and if the LP is unable to find a solution achieving a zero objective then 
 there is no solution to equation \ref{eqn_fixed_cluster}.

We have run this linear program on the graphs of Read and Wilson atlas of graphs \cite{read2005atlas}.
The atlas comprises $1253$ graphs of maximal size $7$, among which $569$ are connected, not regular and do not have dangling nodes, i.e. are candidate non trivial fixed clusters. 
We have found $22$ non trivial fixed clusters. They are represented in figure \ref{fig_RW_fixed_clusters}.
We found $1$ fixed cluster of size $5$, $2$ fixed clusters of size $6$ and $19$ fixed clusters of size $7$.

\begin{figure}[!ht]
    \centering
    \begin{tabular}{cccc}    
    \includegraphics[width=0.22\linewidth]{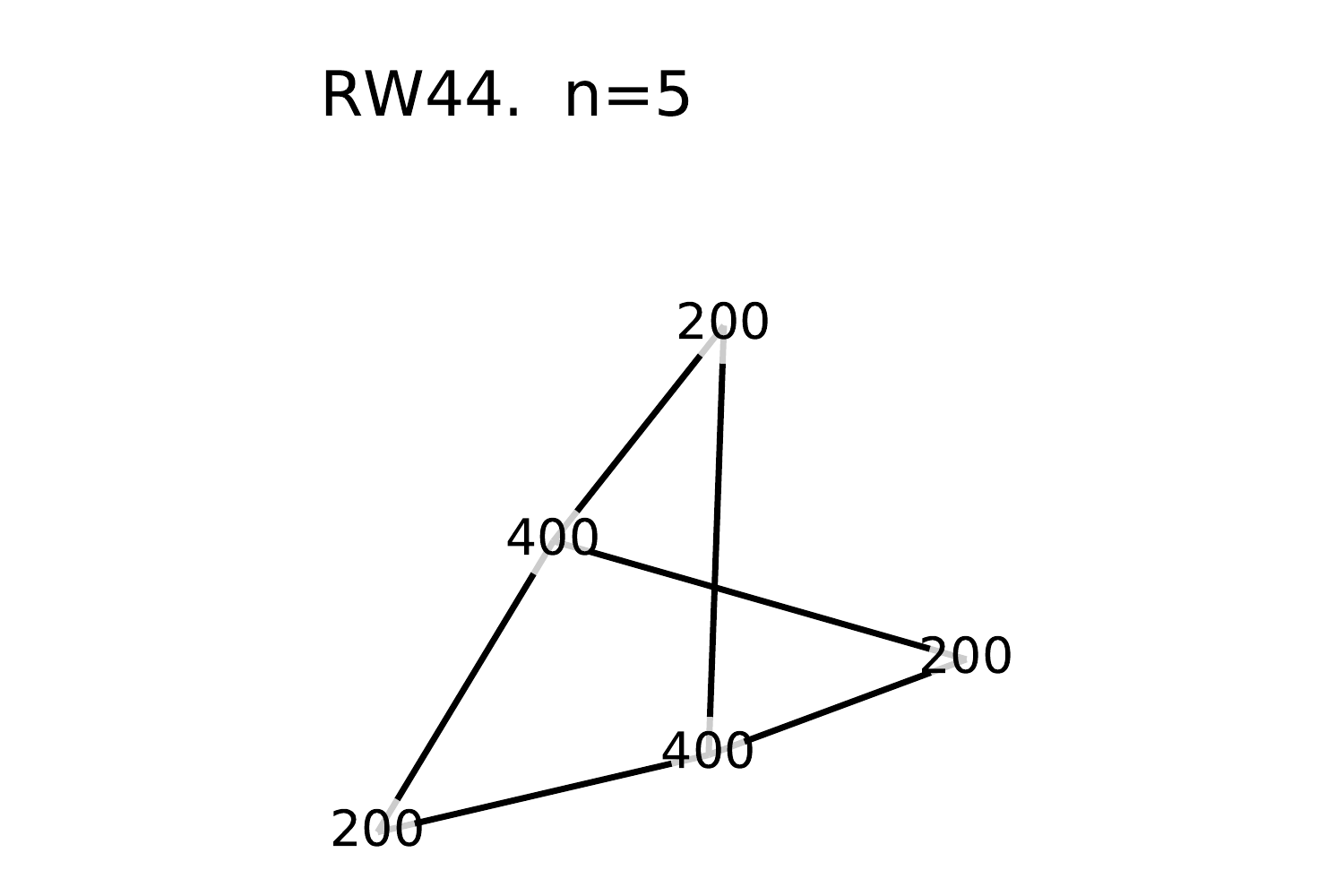} &
    \includegraphics[width=0.22\linewidth]{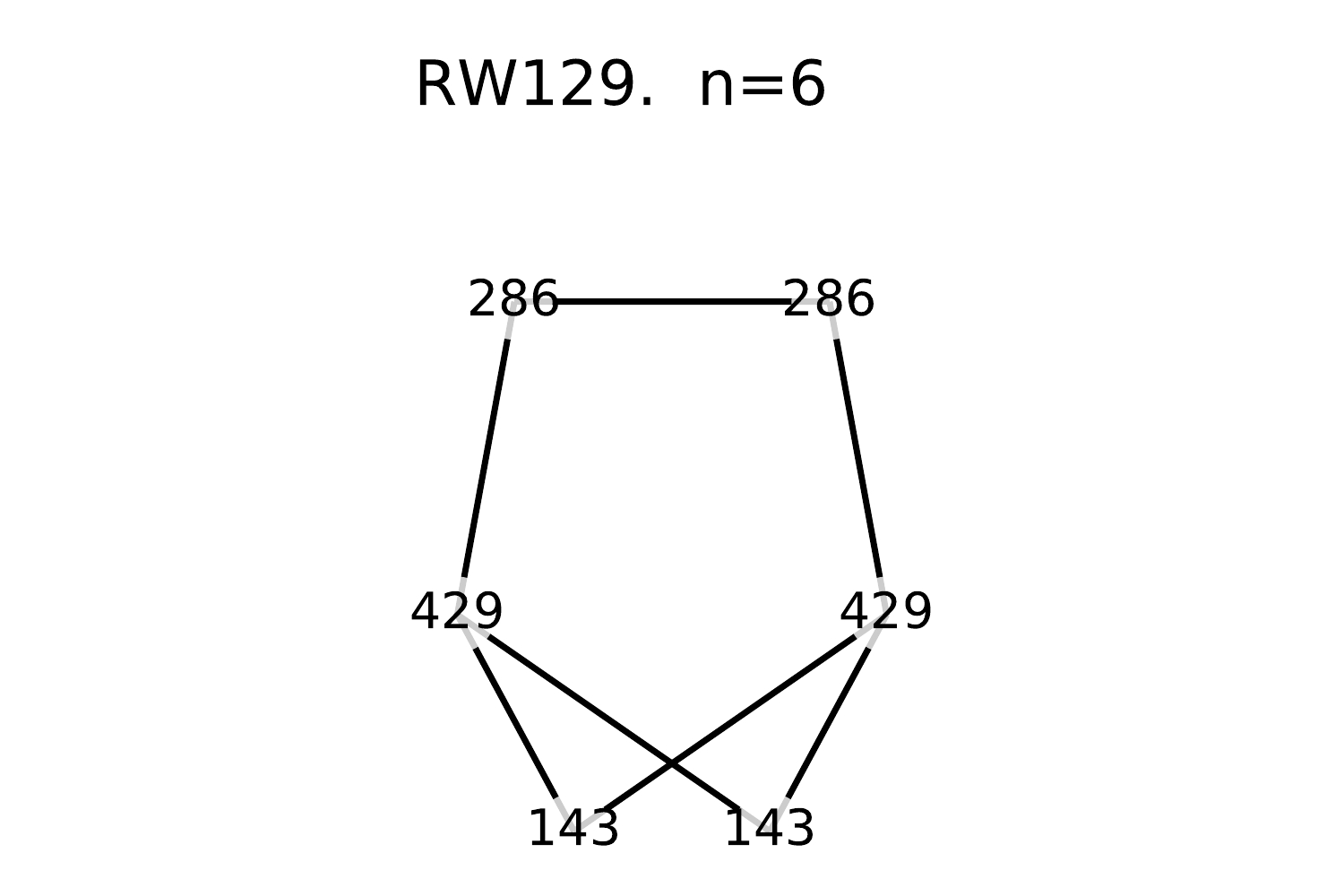} &
    \includegraphics[width=0.22\linewidth]{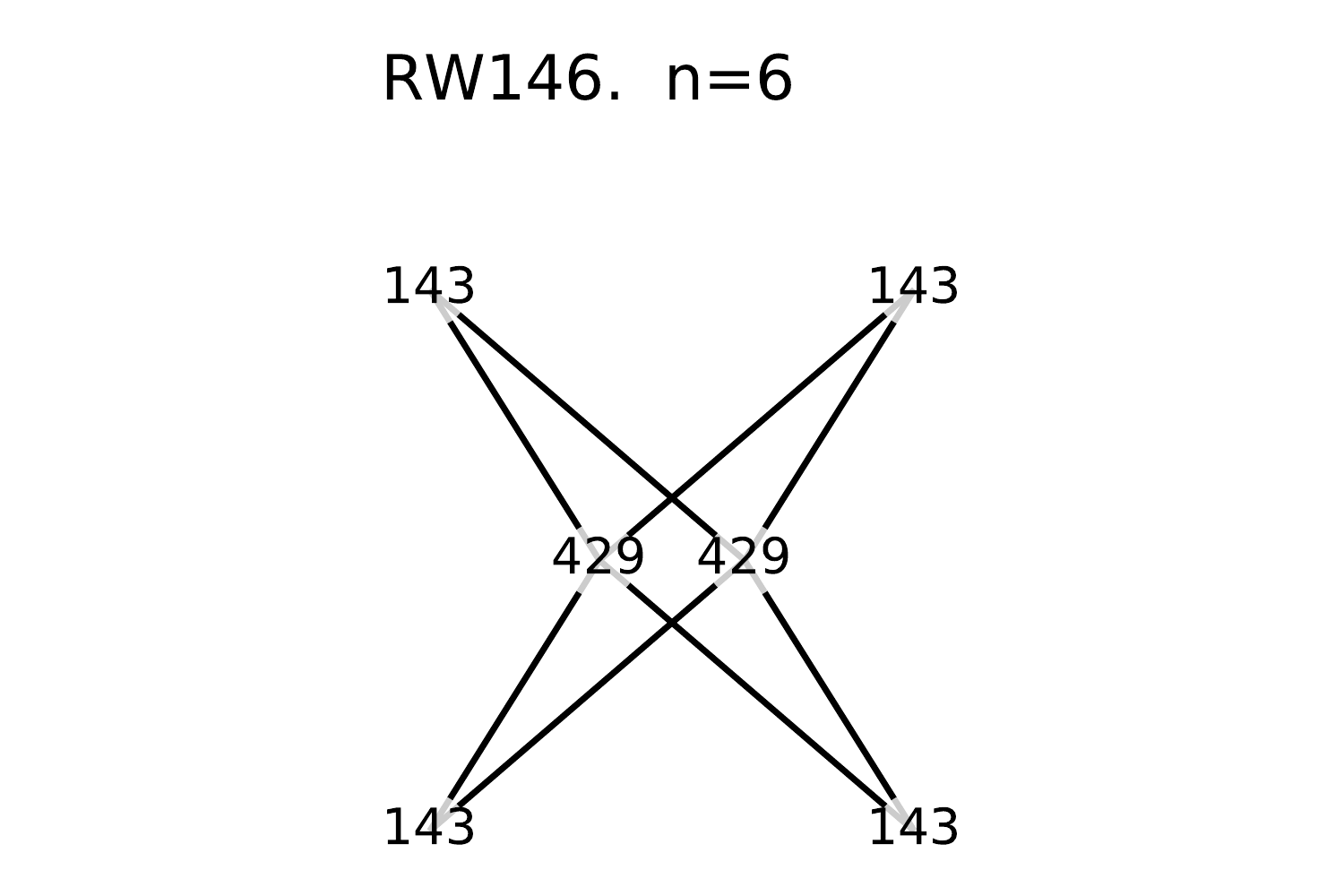} & 
    \includegraphics[width=0.22\linewidth]{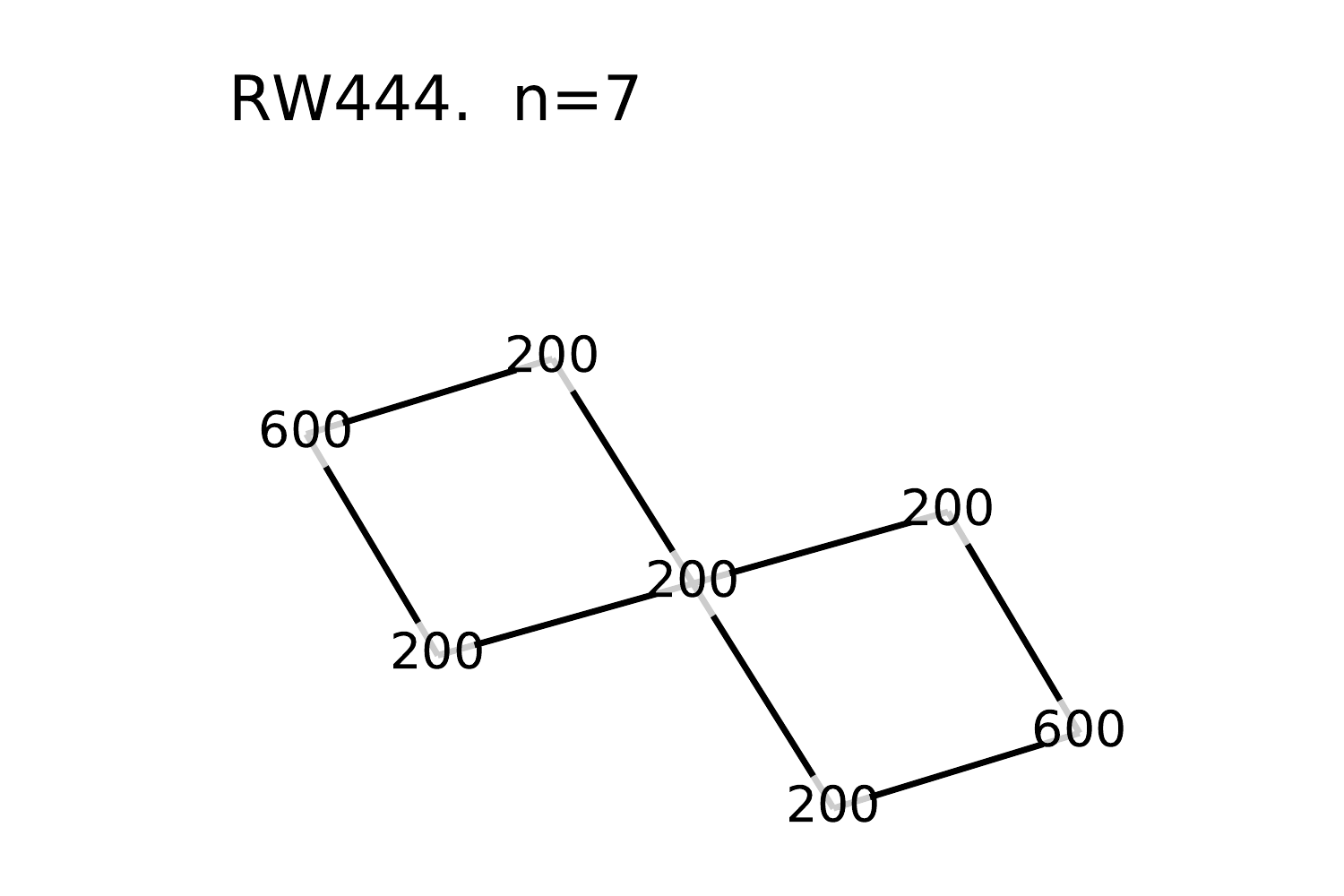} \\
    \includegraphics[width=0.22\linewidth]{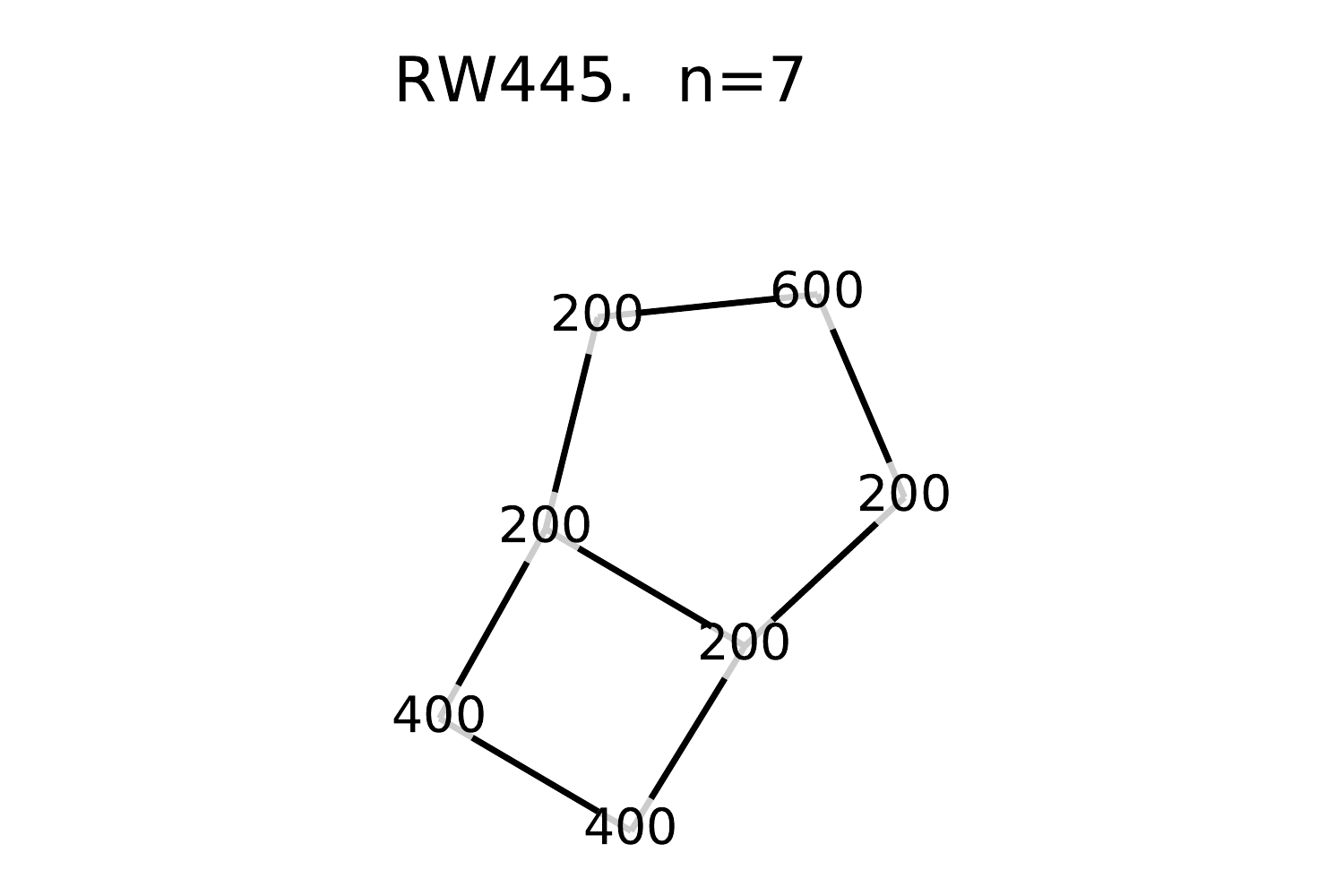} &
    \includegraphics[width=0.22\linewidth]{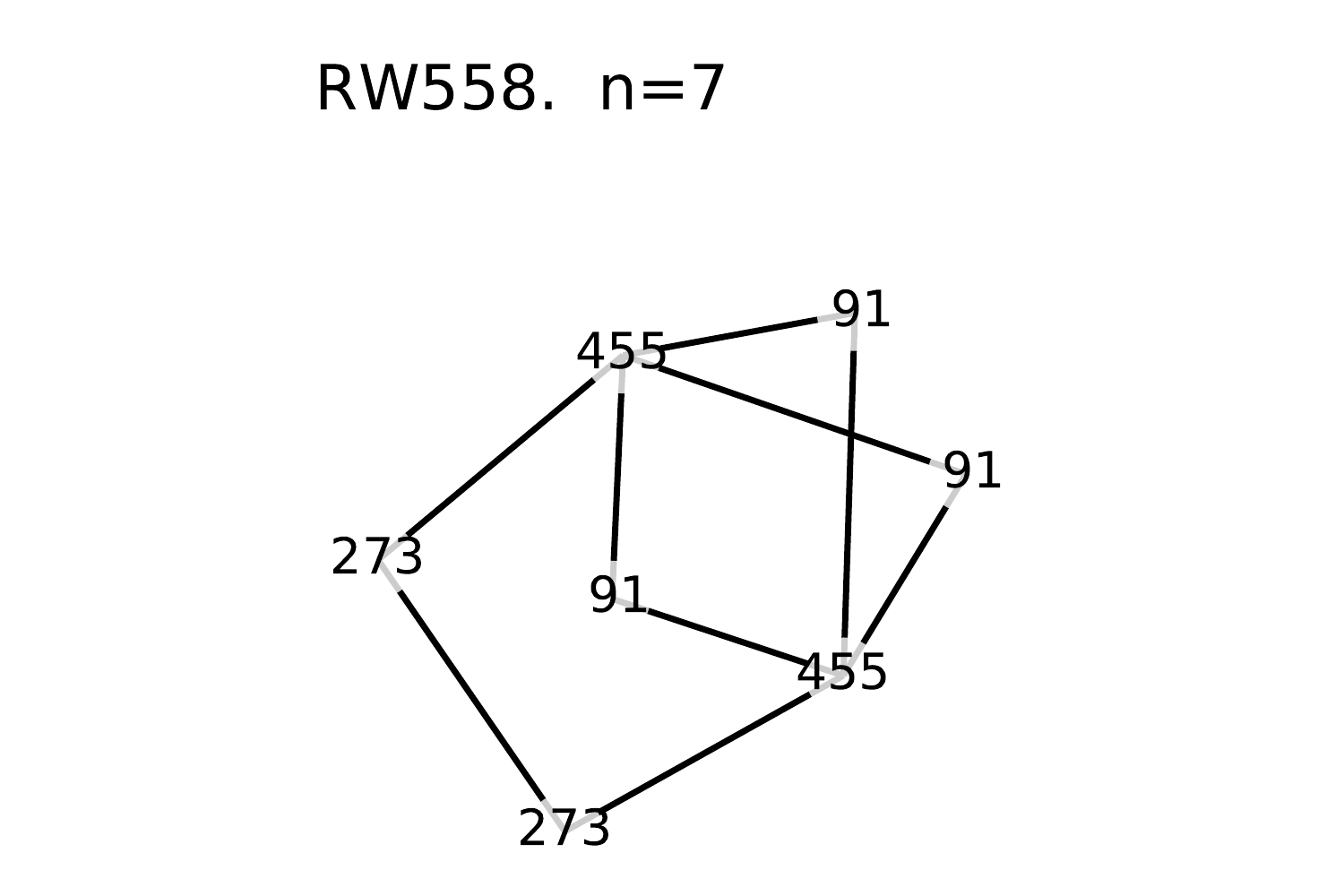} & 
    \includegraphics[width=0.22\linewidth]{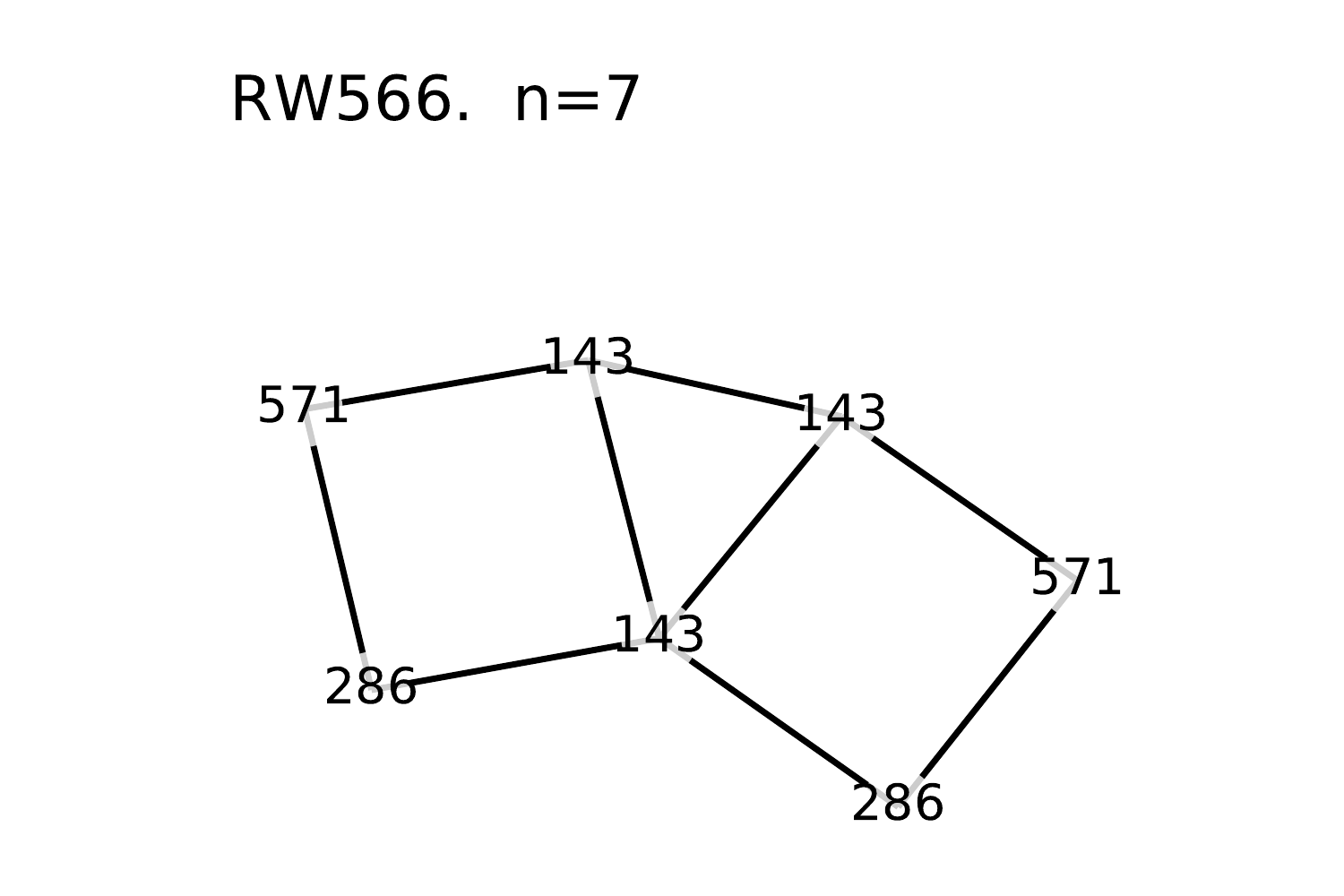} &
    \includegraphics[width=0.22\linewidth]{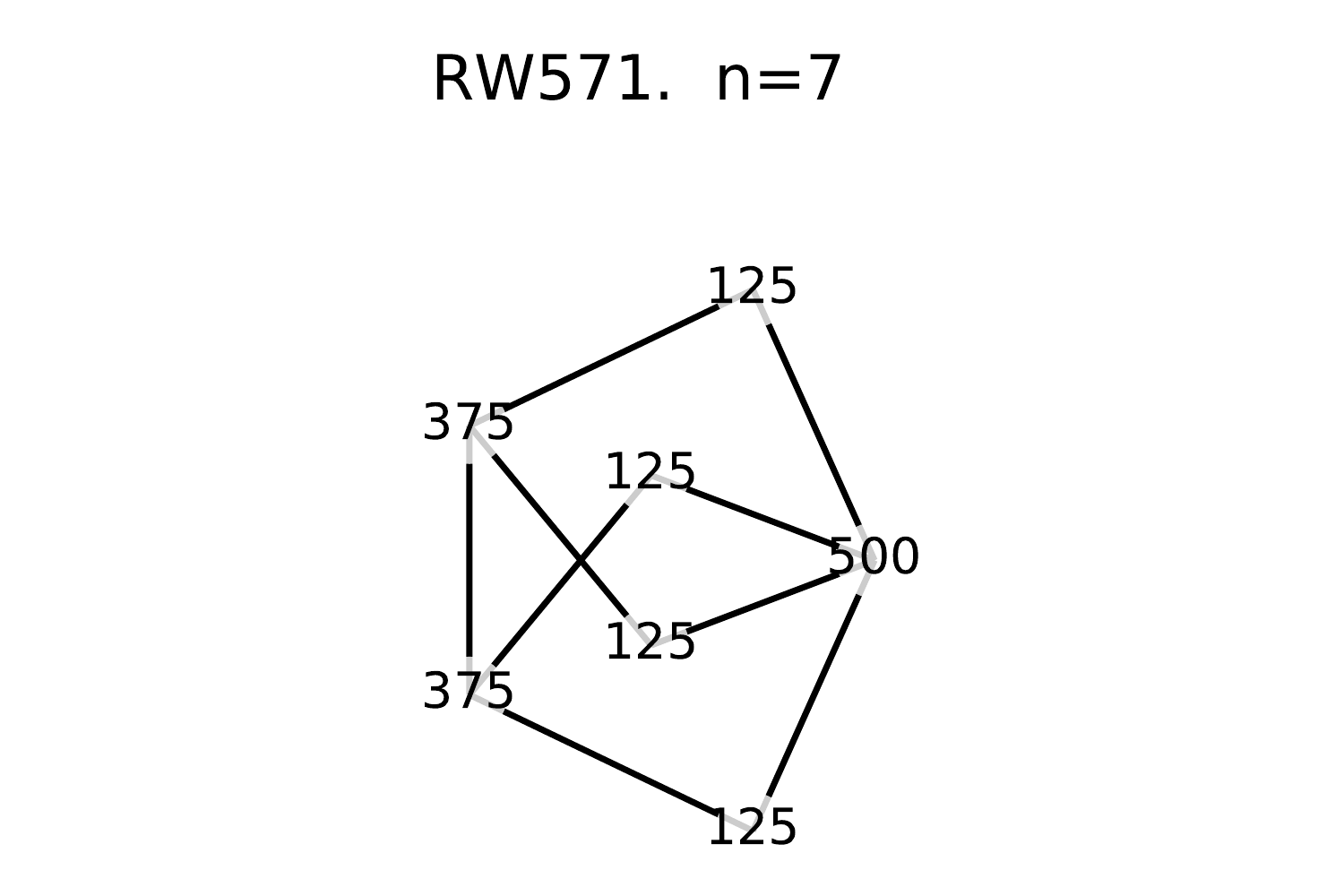} \\
    \includegraphics[width=0.22\linewidth]{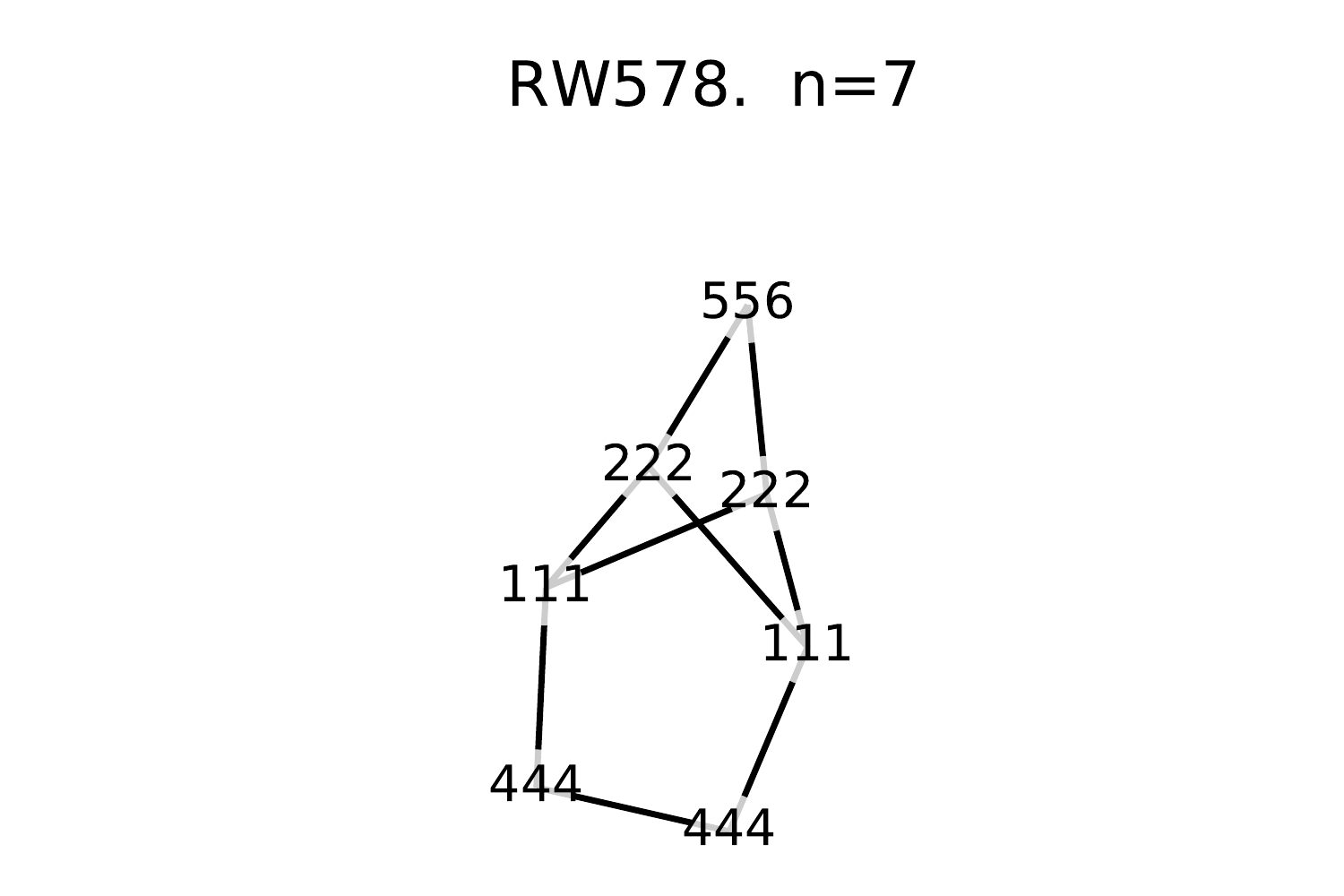} & 
    \includegraphics[width=0.22\linewidth]{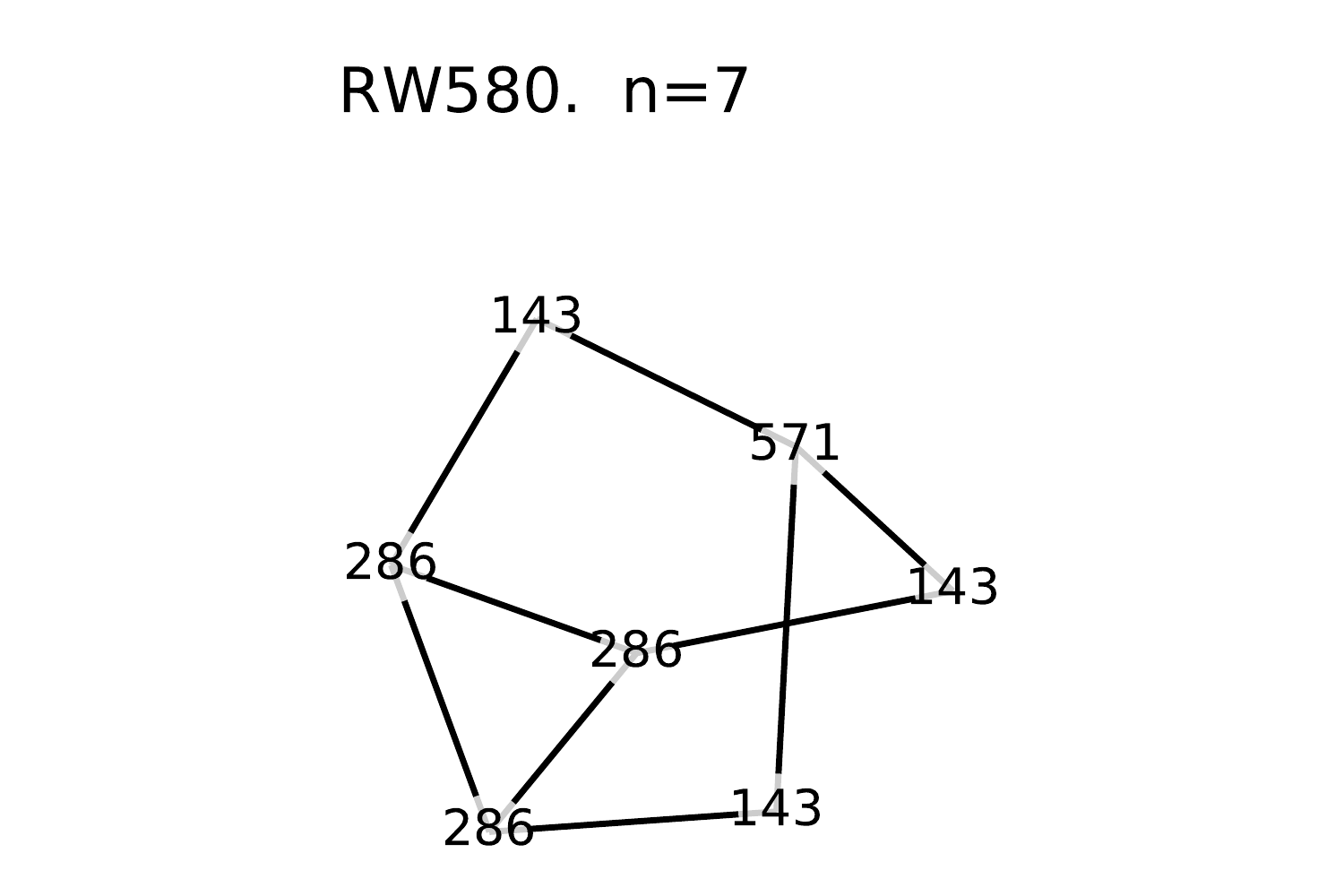} &
    \includegraphics[width=0.22\linewidth]{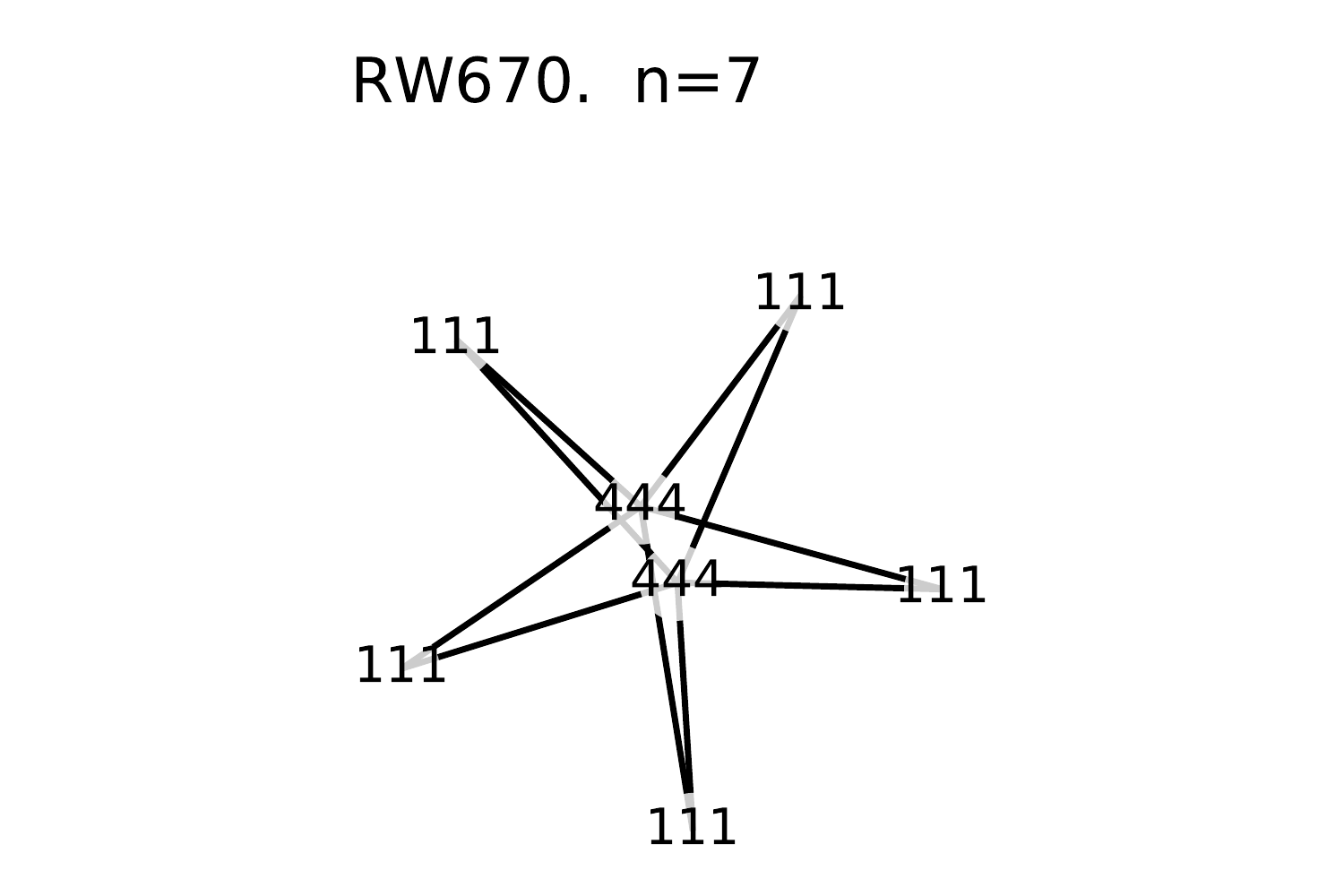} &
    \includegraphics[width=0.22\linewidth]{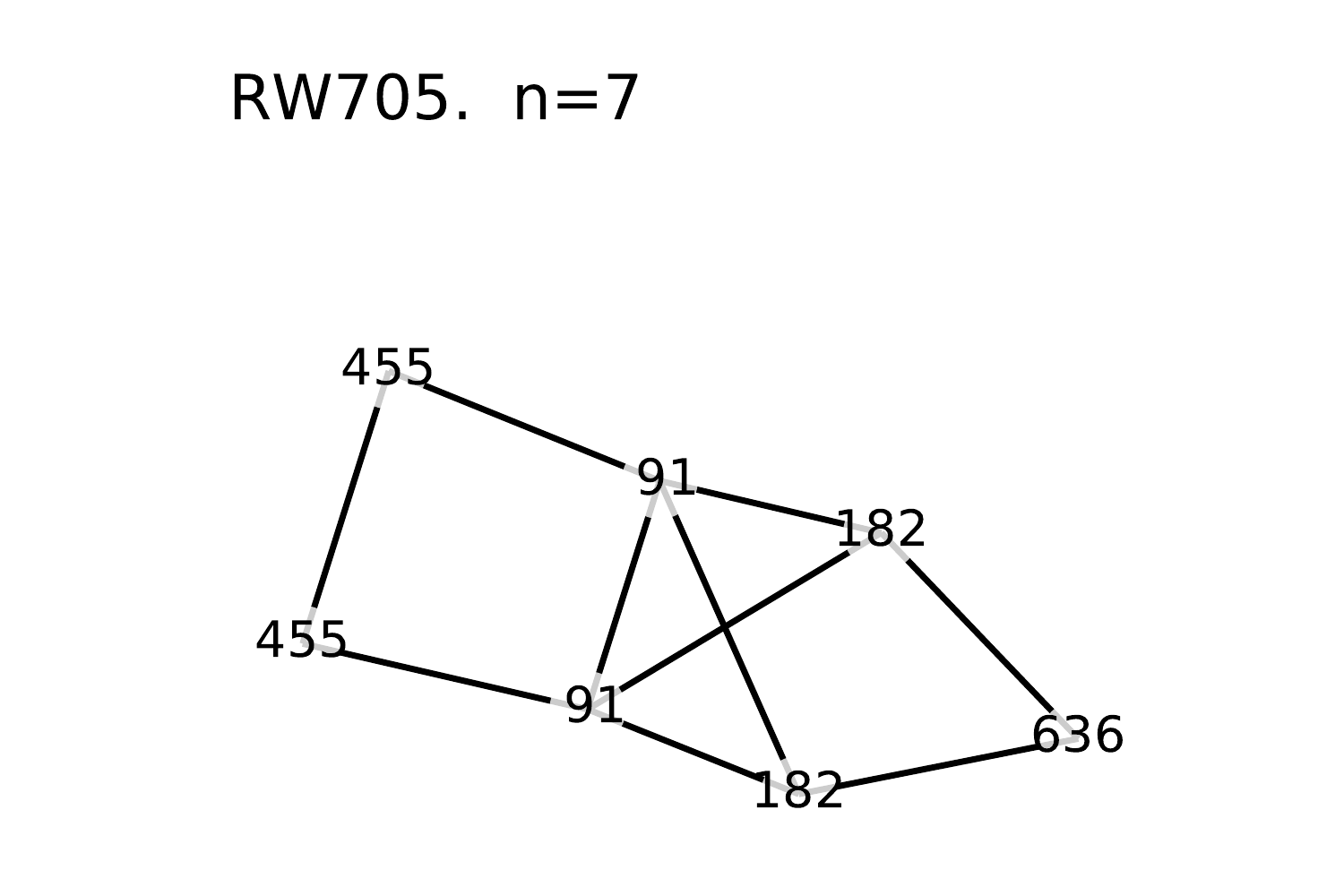}  \\ 
    \includegraphics[width=0.22\linewidth]{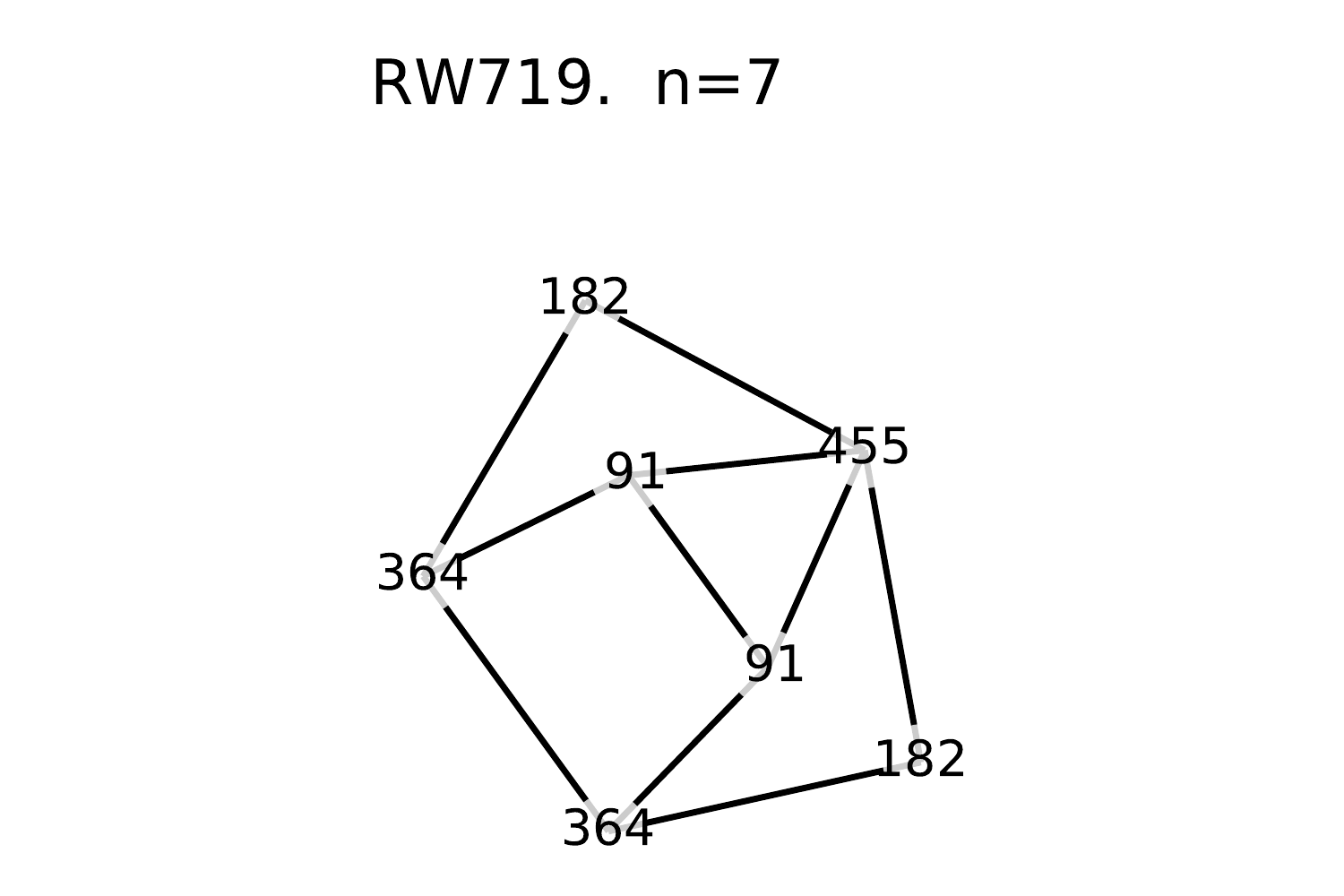} &
    \includegraphics[width=0.22\linewidth]{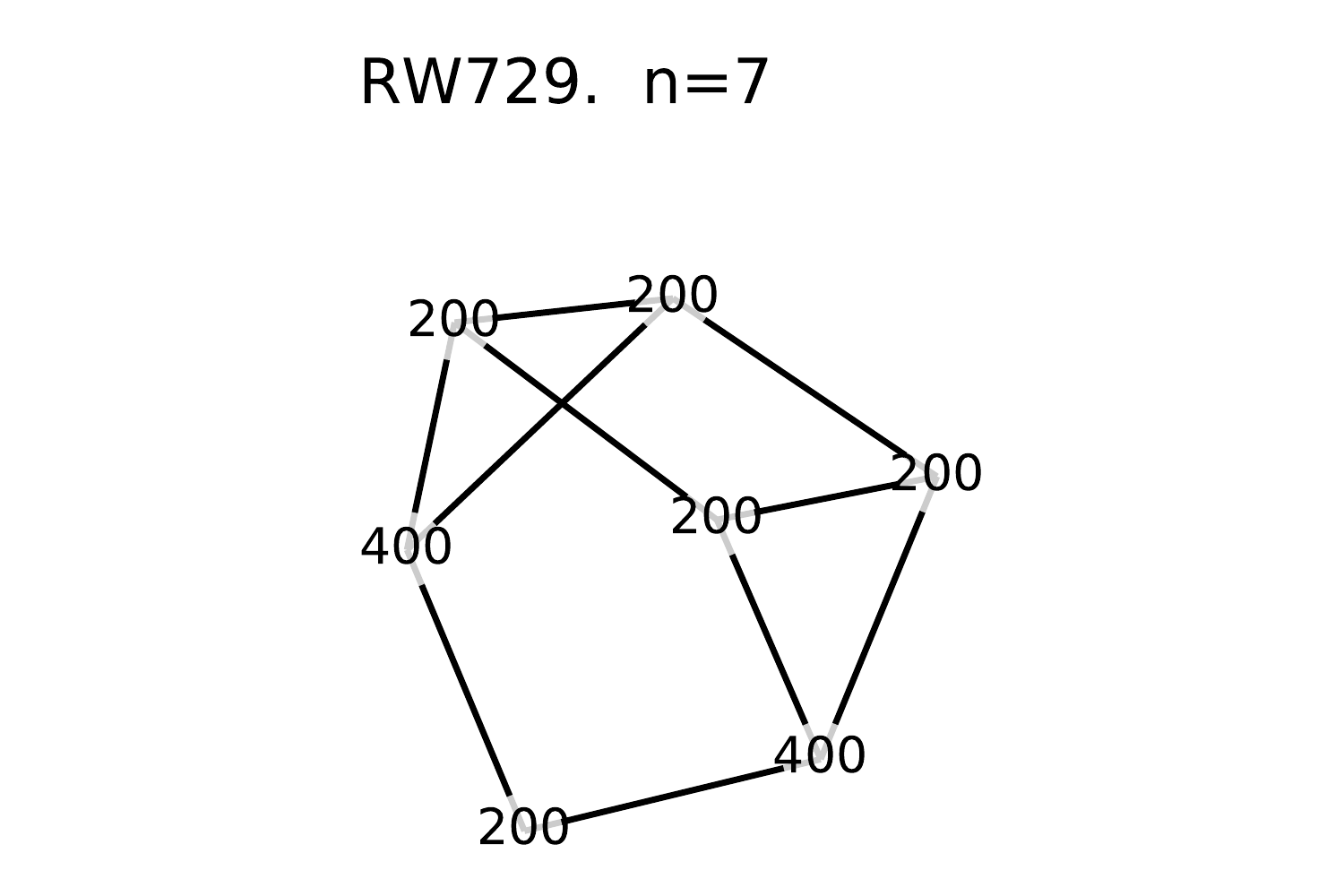} &
    \includegraphics[width=0.22\linewidth]{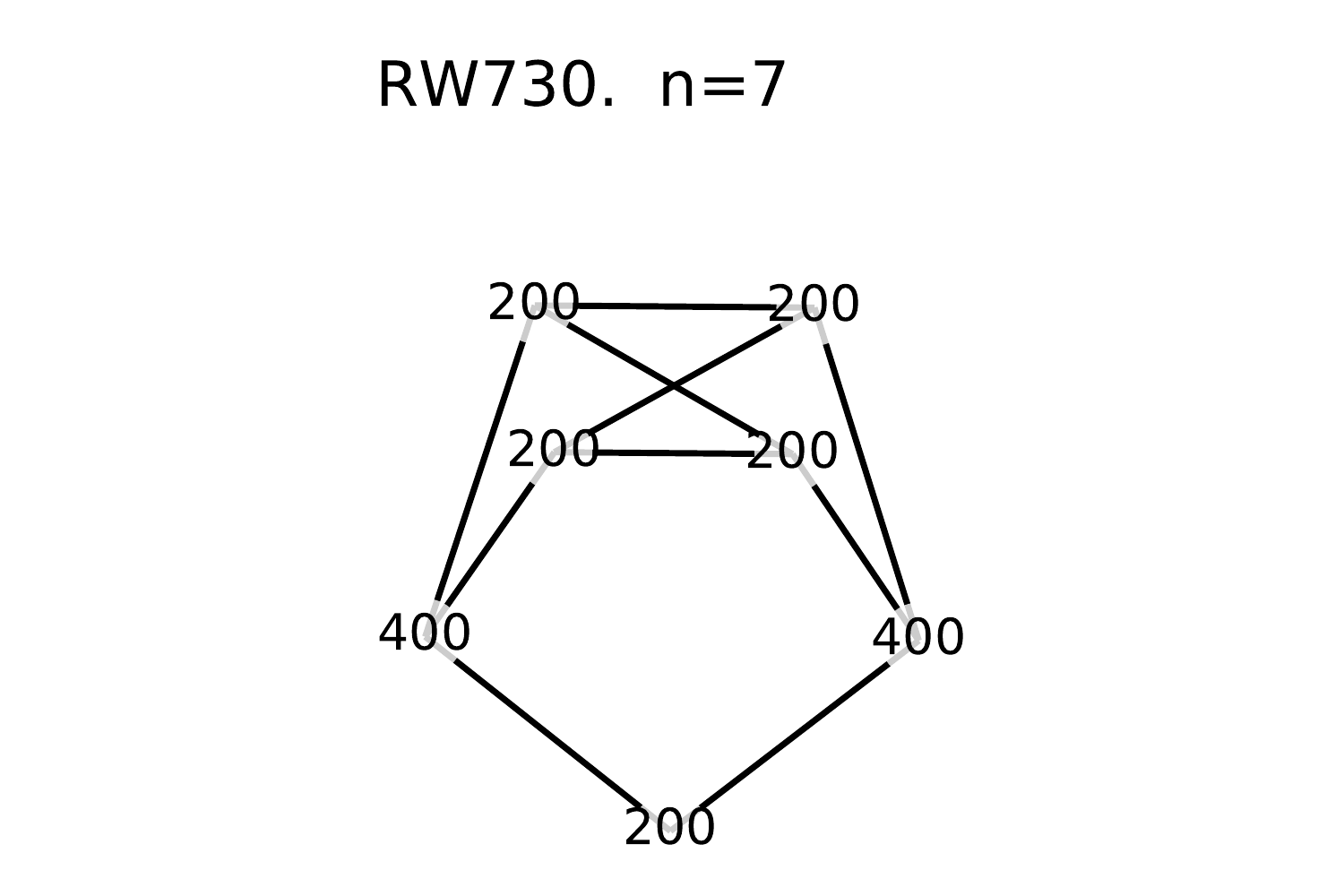}  &
    \includegraphics[width=0.22\linewidth]{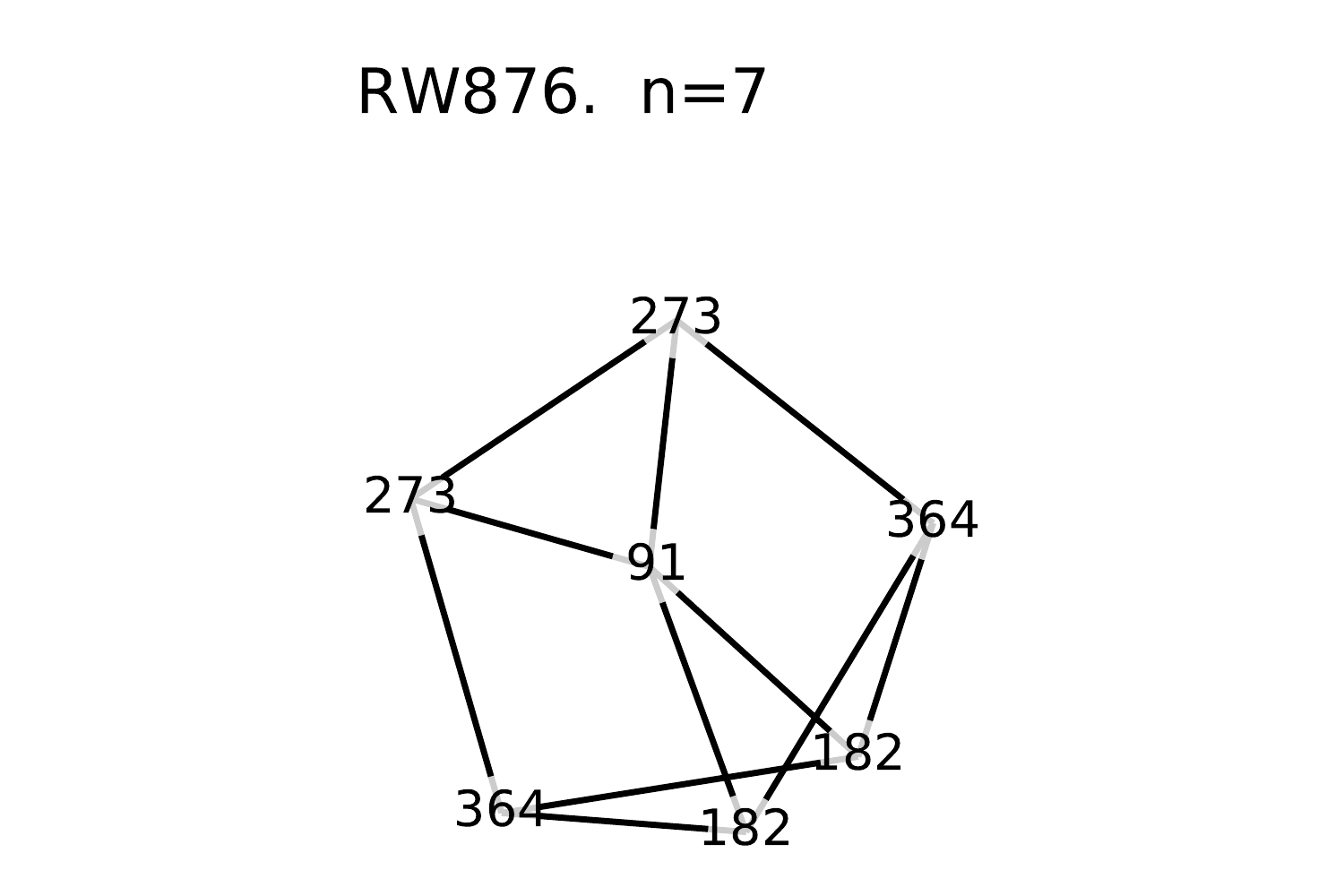}  \\
    \includegraphics[width=0.22\linewidth]{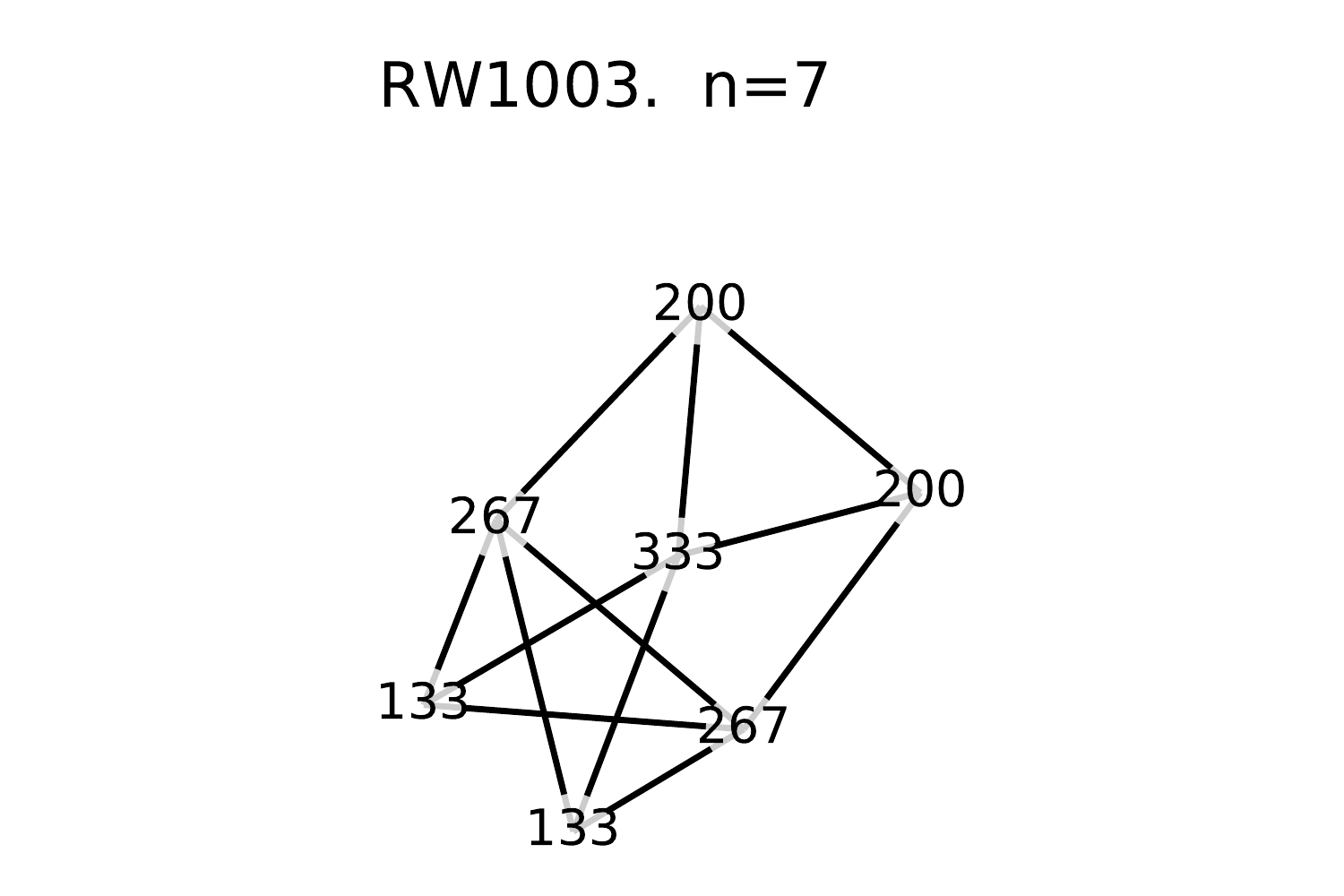} &
    \includegraphics[width=0.22\linewidth]{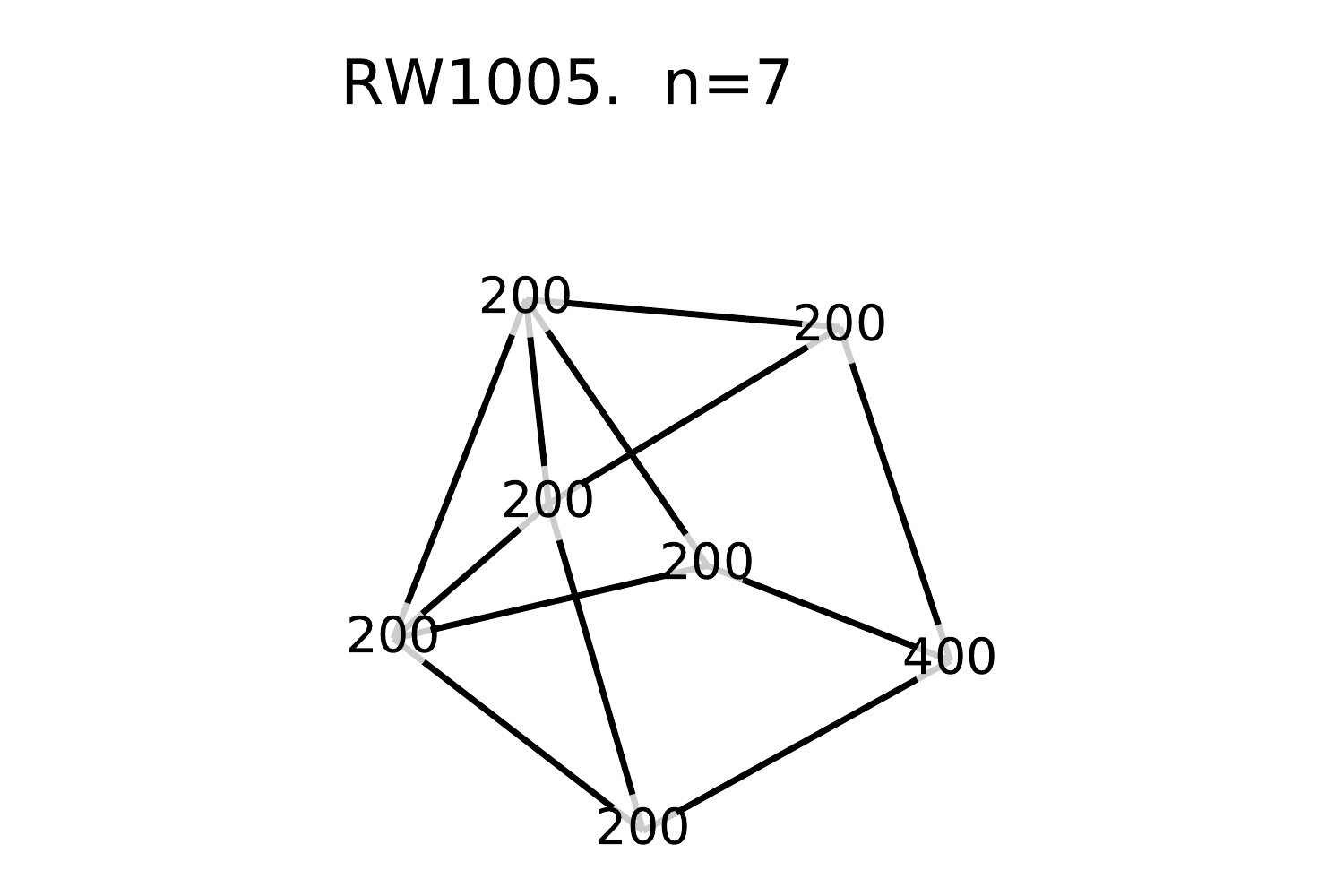} &
    \includegraphics[width=0.22\linewidth]{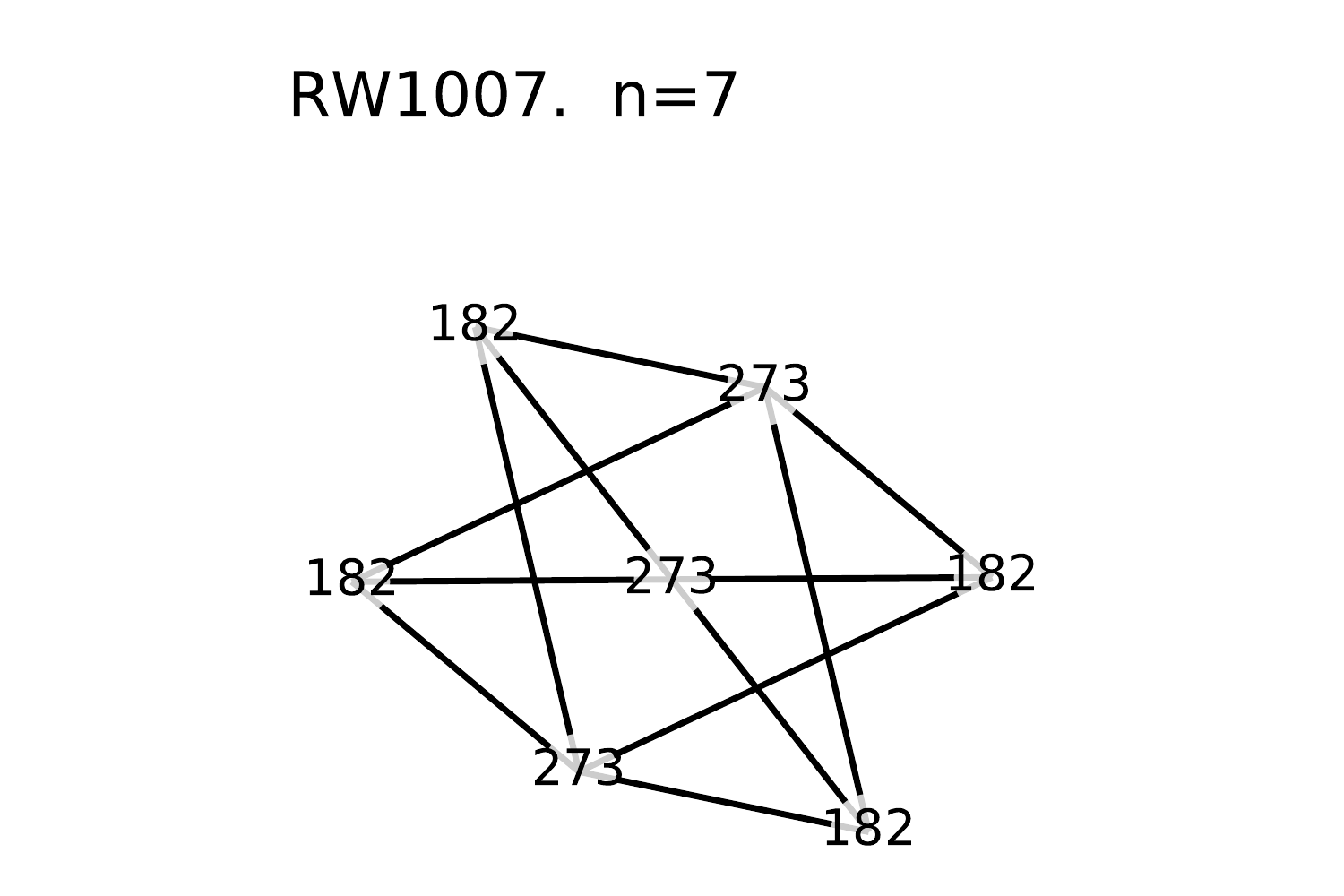}  &
    \includegraphics[width=0.22\linewidth]{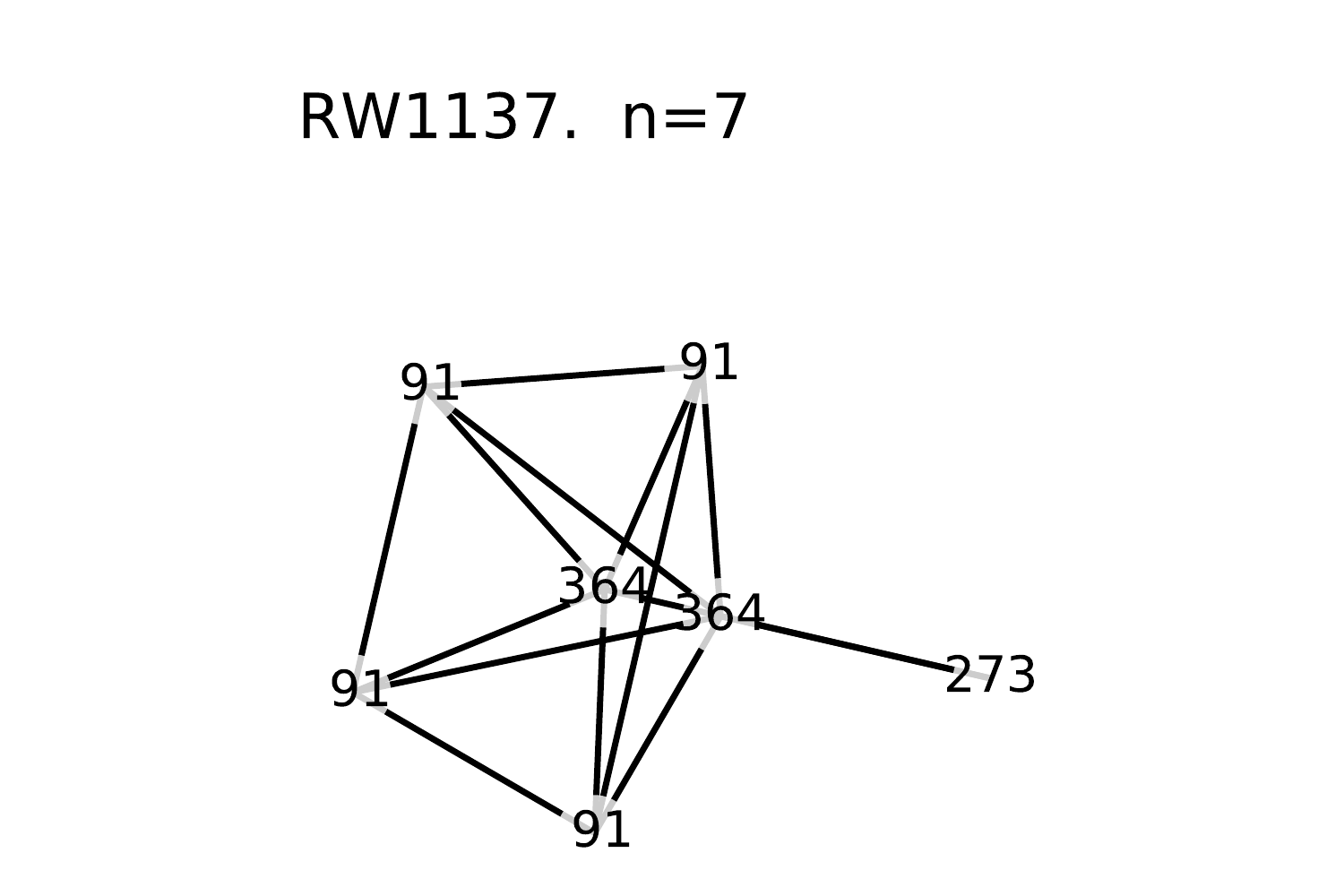}  \\
    \includegraphics[width=0.22\linewidth]{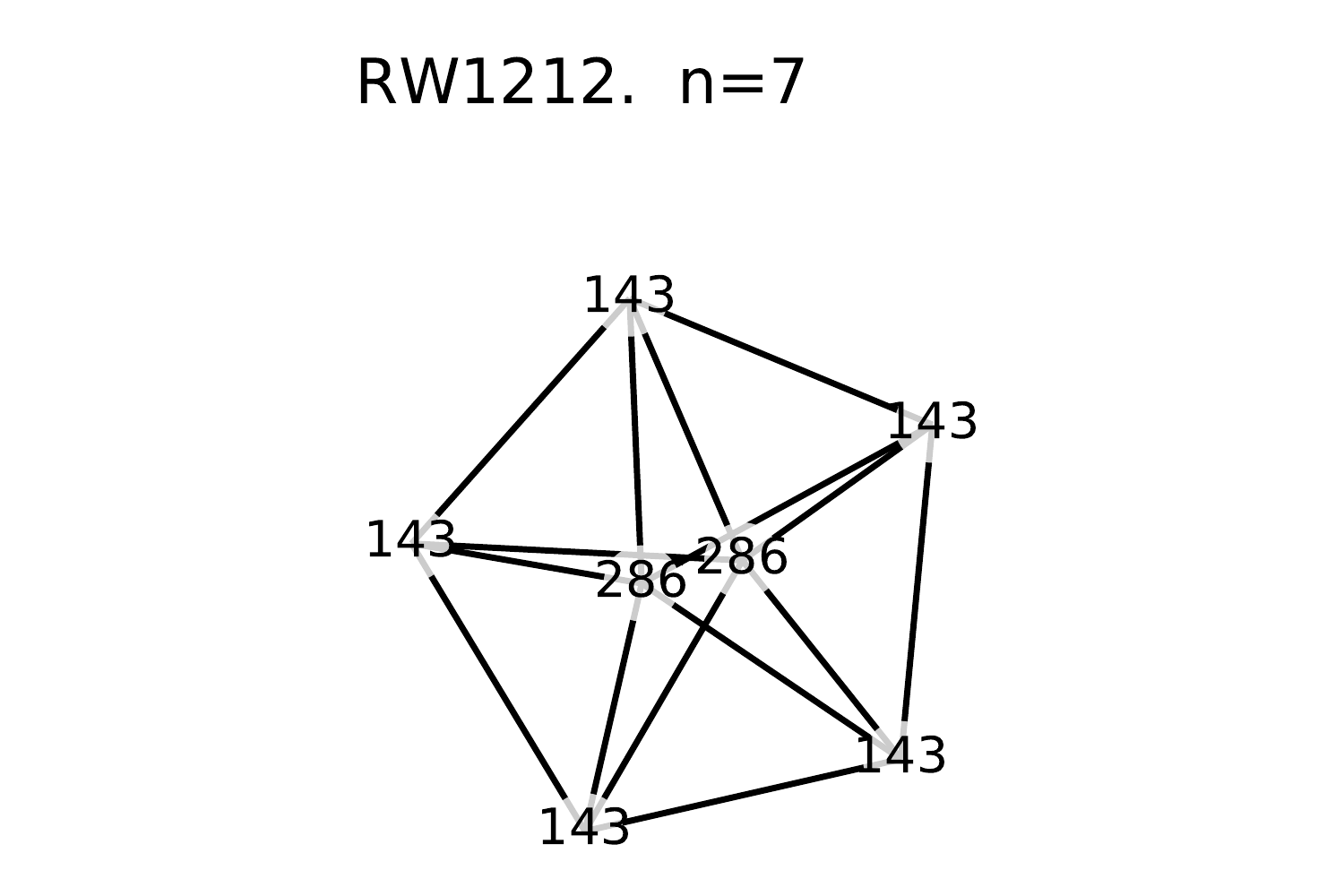} &
    \includegraphics[width=0.22\linewidth]{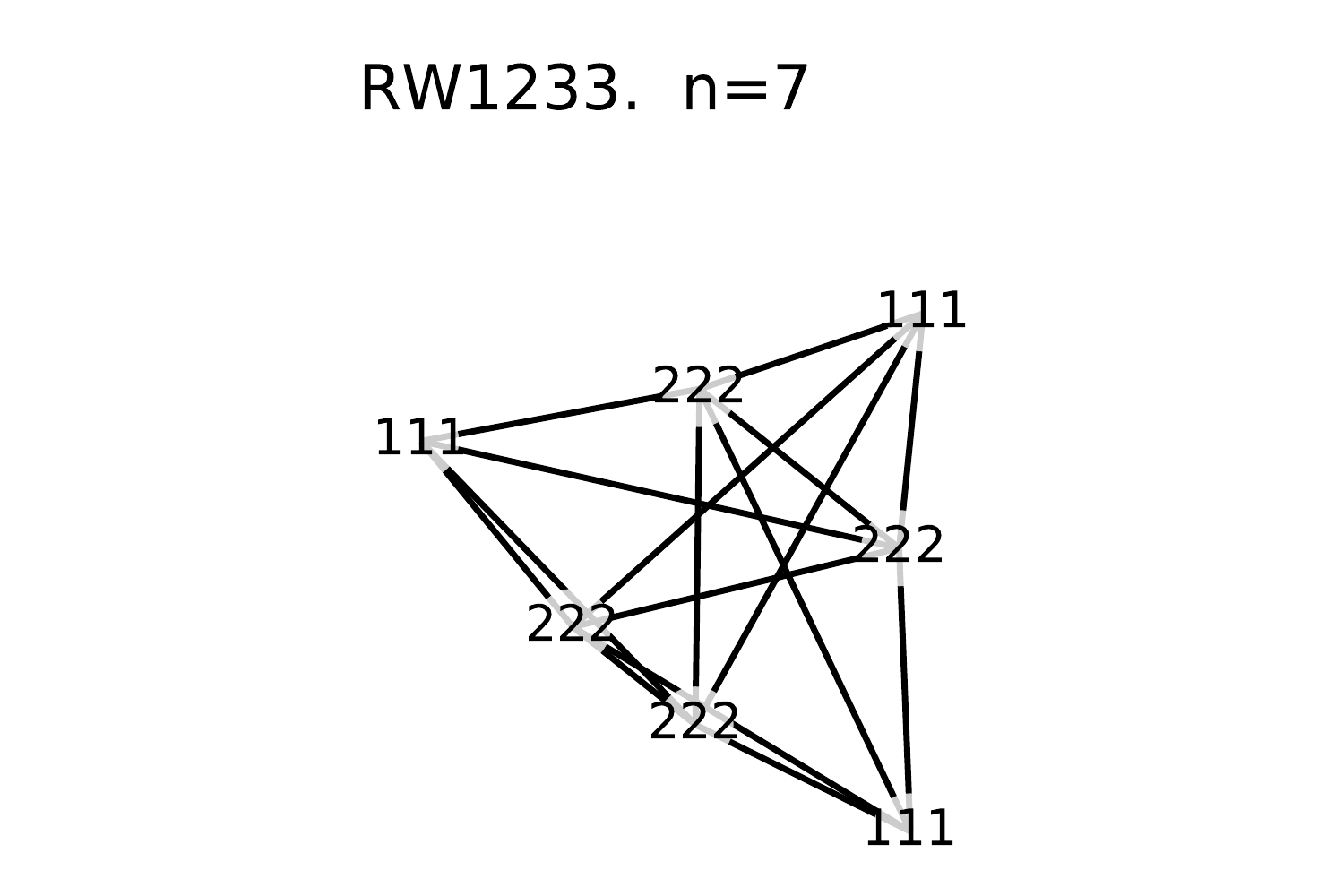}
    \end{tabular}
    \caption{The $22$ non trivial fixed clusters of graph normalization belonging to Read and Wilson atlas of graphs. 
     The numbers on the nodes are the first $3$ decimals of the weights of the solution. 
     For example $286$ represents a weight of $0.286... $ . 
Please note that for a given graph there may be other solutions with different weights as the linear program to find them only exhibits one solution.
Also note that despite what the automatic drawing suggests, $RW1137$ doesn't have a dangling node: the 
node with weight $273$ is connected to the $2$ nodes with weights $364$.  
    }
    \label{fig_RW_fixed_clusters}
\end{figure}

Remark that in accordance with the proposition \ref{ref_propfixedregular}, 
all the non trivial fixed clusters found do not have all-identical weights.

\subsubsection{Stability}
We now look at the dynamics of graph normalization and at the stability of its fixed points, 
first illustrating its behavior with $P_3$. 

The figure \ref{fig_P3_orbits} represents the orbits of normalization on $P3$ when initialized 
at regular barycentric positions on the triangle $\Delta_2$ for different activation functions.
For a linear activation, i.e. no activation, 
one sees that all the trajectories which are not initiated on a fixed point 
converge to $(1,0,1)$ (figure \ref{fig_P3_orbits} a)). 

\begin{figure}[!ht]
    \centering
        \centering
    \begin{tabular}{ccc}    
    \includegraphics[width=0.28\linewidth]{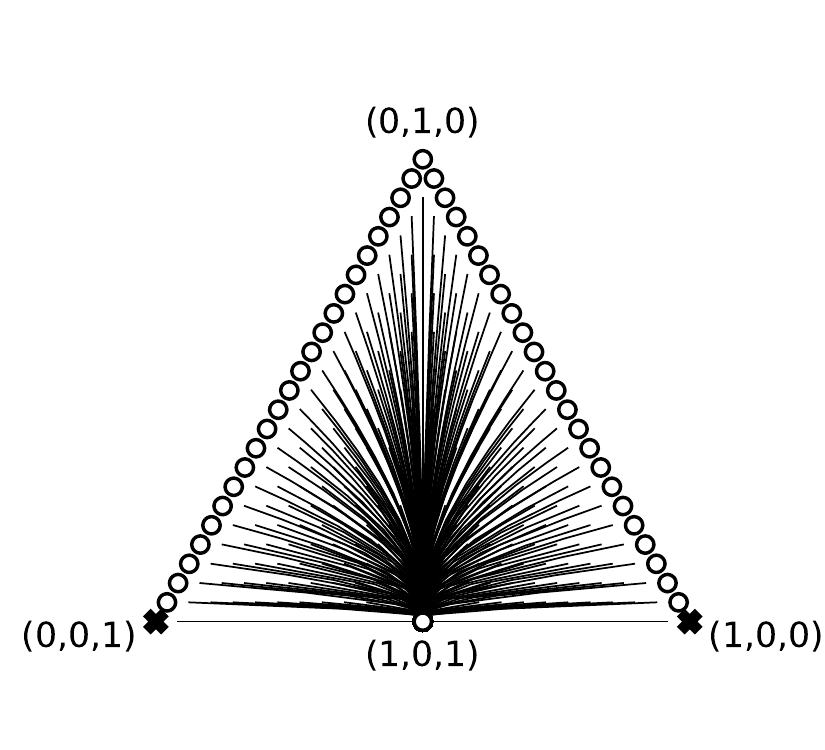} &
    \includegraphics[width=0.28\linewidth]{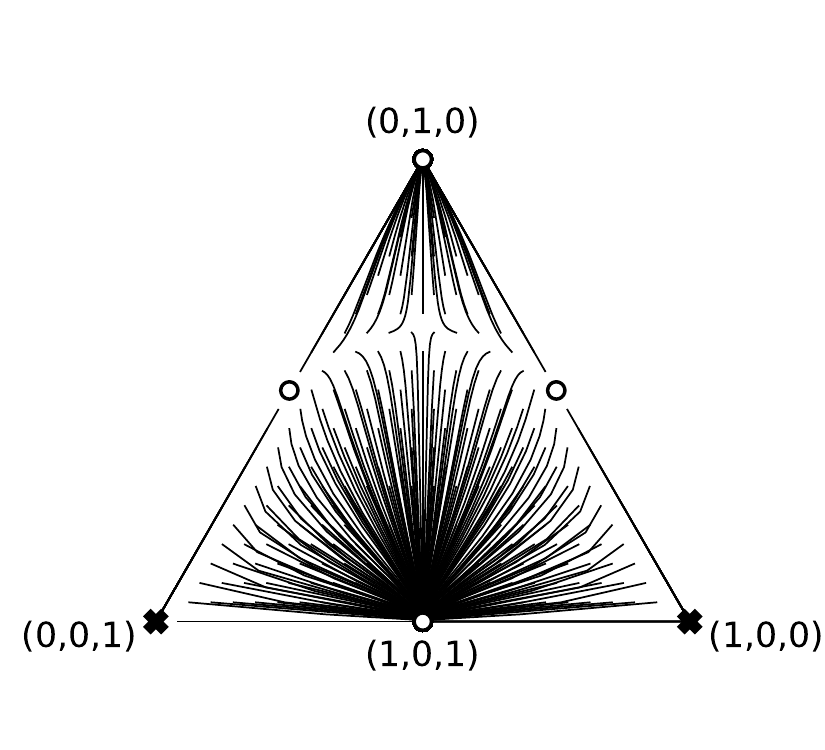} &
    \includegraphics[width=0.28\linewidth]{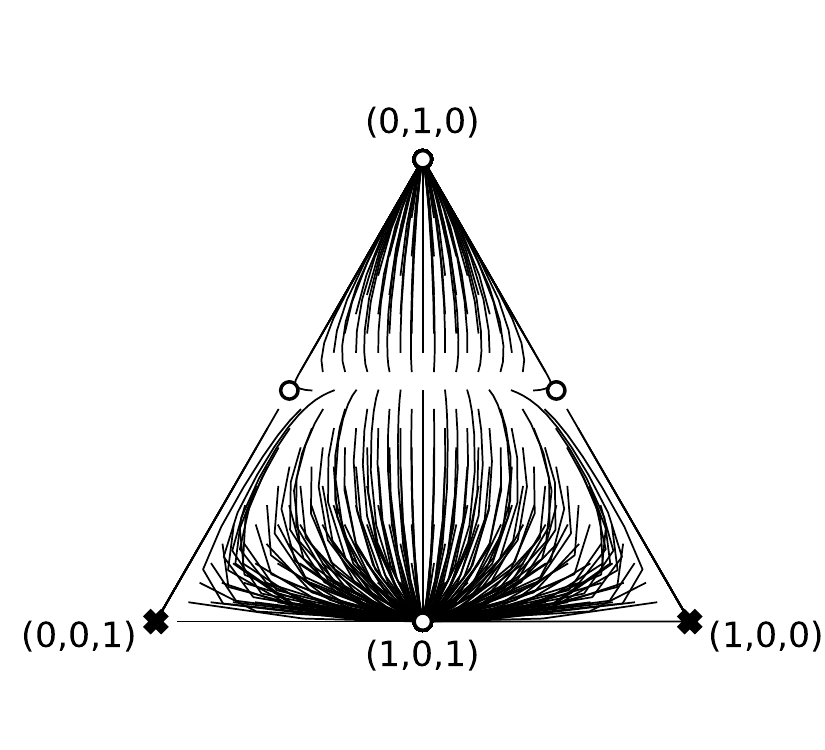} \\
    (a) $h=p_{1.0}$ &
    (b) $h=p_{1.2}$ &
    (c) $h=p_{1.5}$ \\
    \includegraphics[width=0.28\linewidth]{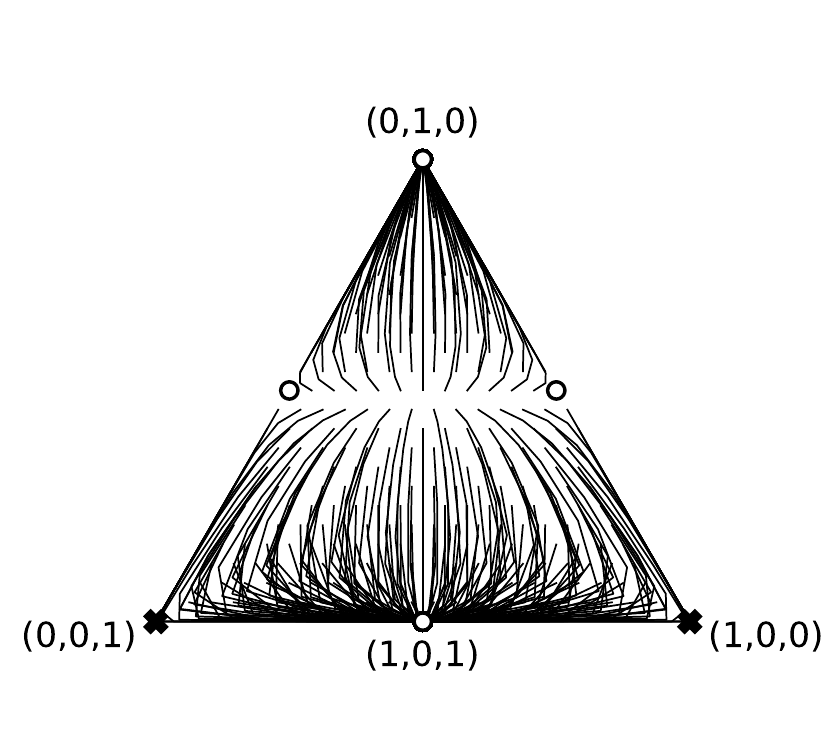} &
    \includegraphics[width=0.28\linewidth]{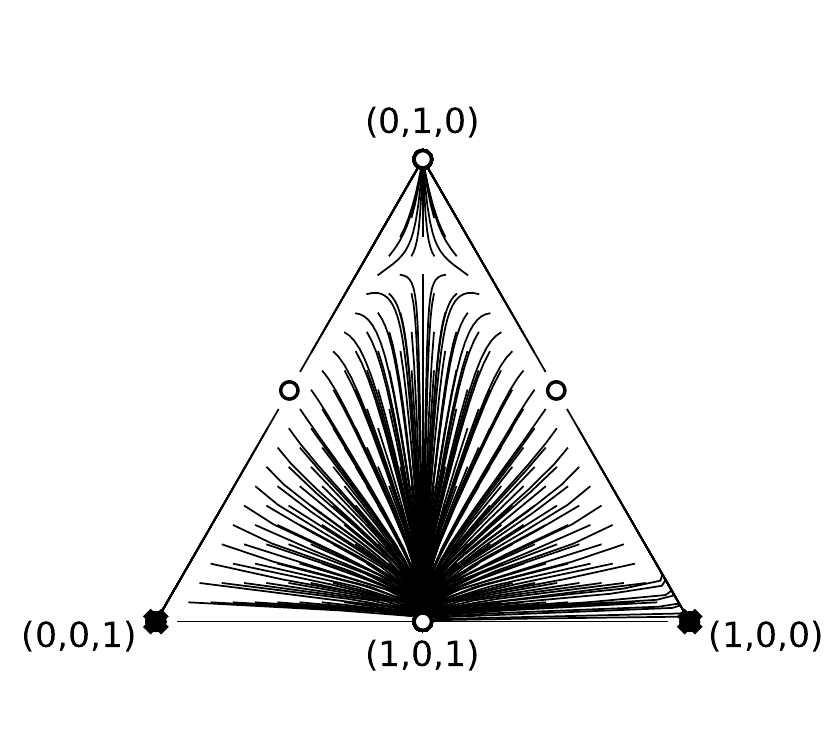} &
    \includegraphics[width=0.28\linewidth]{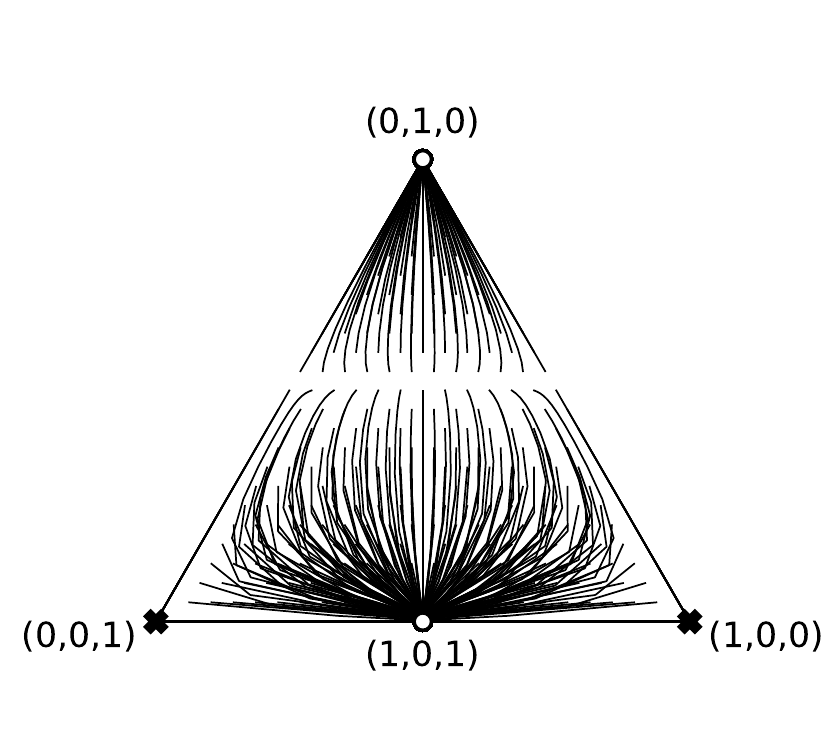} \\
    (a) $h=p_{2.0}$ &
    (b) $h=s_{2.0}$ &
    (c) $h=s_{5.0}$ \\
    \end{tabular}
    \caption{Orbits of the normalization of $P_3$ for different activation functions. 
    For each plot, we initialize normalization on the points of a $2D$ grid 
    defined by regular barycentric coordinates on the triangle $\Delta_2$, which thus samples the 
    projective space $\mathbb{P}_2(\mathbb{R}^+)$. 
    The $3D$ trajectory of each weight vector $x$ is then reprojected 
    on $\Delta_2$ by dividing by $\ONE^t x$ (regular $L_1$-norm normalization) and transformed to 
    a $2D$ reference system attached to the plane containing $\Delta_2$ and centered at $(1,0,1)$.
    For each initialization we perform $500$ iterations to be sure to get close to convergence.
    The black curves represent the trajectories of the iterations from each initialization (orbits). 
    A white circle is drawn at each final point of a trajectory, hence at each numerical fixed-point of normalization.
    The black crosses represent the 2 points numerically tested as non normalizable ($(1,0,0)$ and $(0,0,1)$). See discussion in the text.
    }
    \label{fig_P3_orbits}
\end{figure}

Apart from the special case of $P_3$, 
let us now look at the attractivity of the fixed points of normalization for general graphs.
It is well know that the spectral radius of the Jacobian of a discrete dynamical system 
governs the attractivity of its fixed points.
One easily verifies by calculus that for a graph $G=(A,x)$ and with the notation $B=A+I$,
the components of the Jacobian $J(x)$ of $\Rh$ at $x$ are given by
\begin{eqnarray*}
\forall i \quad
J_{ii}(x) = 
\frac{ \partial {\Rh}_i } {\partial x_i} (x) 
& = &
h\prime \left( \frac{x_i}{\sum B_{ij}x_j}\right) 
\frac{ \sum A_{ij}x_j } { \left( \sum B_{ij}x_j \right) ^2}\\
\forall j\ne i \quad 
J_{ij}(x) = \frac{ \partial {\Rh}_i } {\partial x_j} (x) 
& = & 
h\prime \left( \frac{x_i}{\sum B_{ij}x_j}\right) 
\frac{ -x_i A_{ij}} { \left( \sum B_{ij}x_j \right) ^2} 
\end{eqnarray*}

The structure of $J(x)$ is quite interesting.
It can be written as:
\begin{equation}
J(x) = H(x) \times ( W(x) - \DIAG(x) A)
\end{equation}
where:
\begin{itemize} 
\item$H(x)$ is the diagonal matrix with diagonal entries:
\[
H_{ii}(x) = h\prime \left( \frac{x_i}{\sum B_{ij}x_j}\right)  /  \left( \sum B_{ij}x_j \right) ^2
\]
\item$W(x)$ is the diagonal matrix with diagonal entries:
\[
W_{ii}(x) = \sum A_{ij}x_j
\]
\item$\DIAG(x)$ represents the diagonal matrix with the $x_i$ on its diagonal.
\end{itemize}

$J$ can be interpreted 
as a weighted version of the Laplacian of $G$, which is defined by L$(G) = D - A$, 
where $D=\DIAG(A\ONE)$ is the diagonal matrix which has the degrees of the nodes on its diagonal. 

Evaluating the Jacobian $J$ on the all-ones vector, we get: 
\[
J(\ONE) = H(\ONE) \times L(G)
\]
with
\[
H_{ii}(\ONE) = h\prime \left(\frac{1}{\DEG(i)+1}\right) / (\DEG(i)+1)^2
\]
hence $J(\ONE)$ is the Laplacian whose lines are weighted by the $H_{ii}(\ONE)$ which are only
functions of the degrees of the nodes.

On an arbitrary point $x$, $W_{ii}(x)$ is the sum of the weights of the neighbors of $i$, 
hence can be thought of as the weighted degree of the node $i$. 
The same, the off-diagonal terms of the line $i$ of $J$ are weighted by the weight $x_i$.
This intriguing relationship between the Jacobian of the graph normalization transform 
and the Laplacian of the graph certainly deserves more investigation.

Now, for maximal independent sets, the spectrum of the Jacobian can be found explicitely:
\begin{restatable}{theorem}{thmisspectralradius}\label{ref_thmisspectralradius}
If $x$ is a maximal independent set of $G$ then the spectrum of the 
Jacobian $J(x)$ of $\Rh$ is 
\begin{equation}
\Lambda(J(x)) = \{0\} \cup \left\{ \frac{h\prime(0)}{\sum_j A_{ij} x_j} ; i\notin \SUPP(x) \right\}
\end{equation}
and its spectral radius is
\begin{equation}
\rho(J(x)) =  \frac{h\prime(0)}{dens_{G}(x)}.
\end{equation}
\end{restatable}

\begin{corollary}\label{ref_coro_attraction}
\begin{enumerate}
\item A maximal independent set $x$ such that $h\prime(0) < dens_{G}(x)$ is an attractive fixed point of $\Rh$.
\item In particular, if $h$ is contractive at $0$ ($h\prime(0) < 1$) then all the maximal independent sets of any graph are attractive.
\item Furthermore, if $h\prime(0) = 0$ then $\forall G, \forall x\in \MIS(G)$, the iterations of $\Rh$ converge quadratically to $x$ around $x$.
\end{enumerate} 
\end{corollary}
Let's look back at the special case of $P_3$ and see how this theorem explain its normalization 
dynamics.
Remark that $(1,0,1)$ has a density of $2$ and $(0,1,0)$ has a density of $1$.
Hence $\rho(J((1,0,1)) = h\prime(0)/2$ and $\rho(J((0,1,0)) = h\prime(0)$.
With a linear activation $h\prime(0) = 1$, hence $(1,0,1)$ is an attractive fixed point 
and $(0,1,0)$ is just at the boundary between attractivity and repulsivity (it is neutral or Lyapunov stable). 
In practice, we found it repulsive, but it becomes attractive as soon as $h\prime(0)<1$, 
as illustrated in figure \ref{fig_P3_orbits}b)-f).

What about the attractivity of the non maximal independent sets of $G$, 
which are $(1,0,0)$ and $(0,0,1)$ for $P_3$?
Remember that they are not normalizable as they have 
a density of $0$. However, they can still be 
attractors of the normalization dynamics.

The following theorem gives a sufficient condition on $h$ so that they are repulsive.
\begin{restatable}{theorem}{thnmisrepulsive}\label{ref_thnmisrepulsive}
If $\forall y\in[0,1] : h\prime(y)\ne0$ then any non maximal independent set $x$ of $G$ 
is a repulsive point for $\R_h$. 
\end{restatable}

Unfortunately, this condition on the repulsivity of the non maximal independent sets  
doesn't allow simultaneously quadratic convergence on the maximal ones.
One must thus make a compromise between convergence speed and quality of the 
solution found, as a non maximal solution is obviously sub-optimal.
This will be detailed in the experimental section below.

\subsubsection{Basins of attraction}\label{sec_basins}
Beyond the local attractivity of the maximal independent sets of a graph, 
it is interesting to look at the spatial extension of the basin of attraction of a MIS.

It is then useful to consider graph normalization as a coupled system of two sequences of vectors: 
the graph weights and the sum of the weights on neighborhoods, and to decompose the iterations as:
\begin{eqnarray*}
x^0 &=& x \\
r^k &=& Ax^k \\
x^{k+1} &=& \frac{x^k}{x^k + r^k}
\end{eqnarray*}

For a node $i$, $r^k_i$ is the sum of the weights of its neighbors at iteration $k$.
We call $r^k_i$ the \emph{complementary weight} of the node $i$ 
and $r^k$ the vector of complementary weights of the graph.

Let us first consider a node of weight $x$ and complementary weight $r$, 
and assume that \emph{$r$ is fixed}, 
i.e. assume that the weights of all the other nodes 
of the graph are \emph{not} updated during the iteration.

$x$ then evolves according to:
\begin{eqnarray}
x^0 &=& x \\
x^{k+1} &=& \frac{x^k}{x^k + r}
\end{eqnarray}

The function $f_r(x) = x/(x+r)$ has 2 fixed points $x=0$ and $x=1-r$.
Indeed:
\[
x = \frac{x}{x+r} \quad\Leftrightarrow\quad x(x+r-1) = 0
\]

The figure \ref{fig_f_r} shows the graphs of $f_1$ and $f_{0.3}$.

\begin{figure}[!ht]\center\begin{tabular}{cc}
    \includegraphics[width=0.4\linewidth]{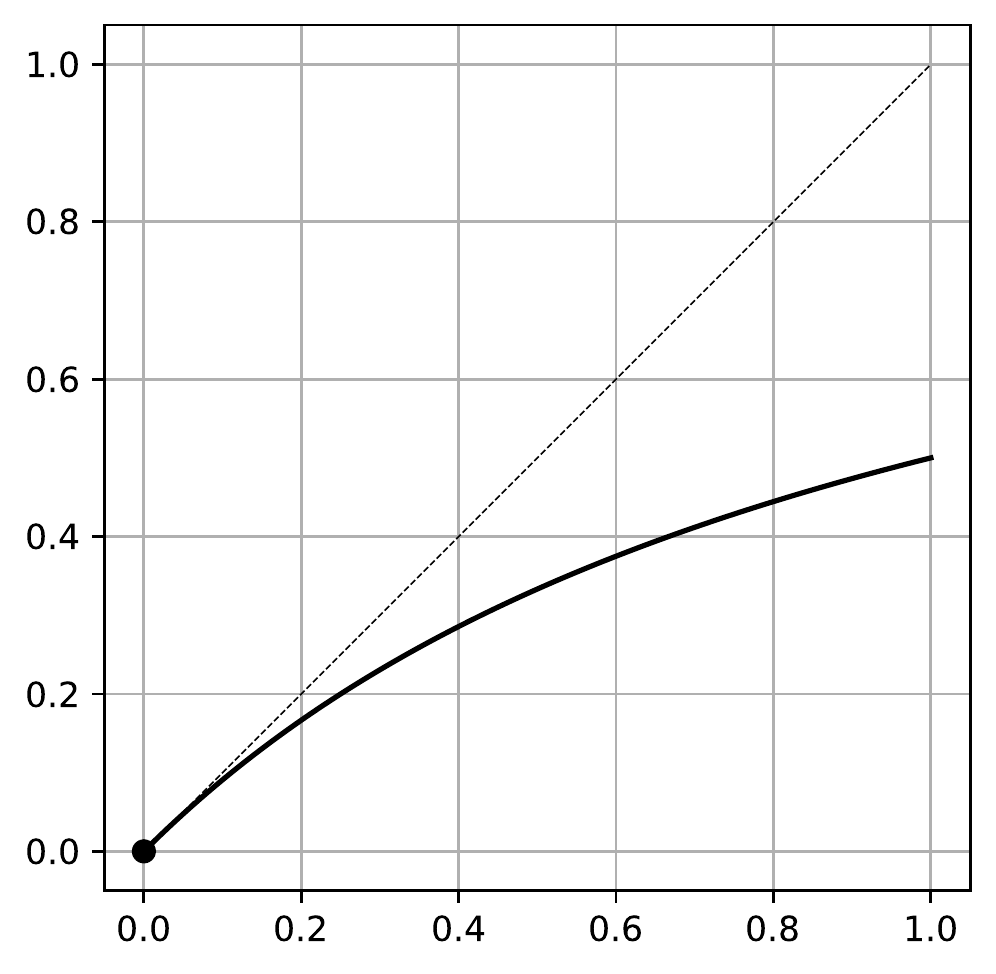}&
    \includegraphics[width=0.4\linewidth]{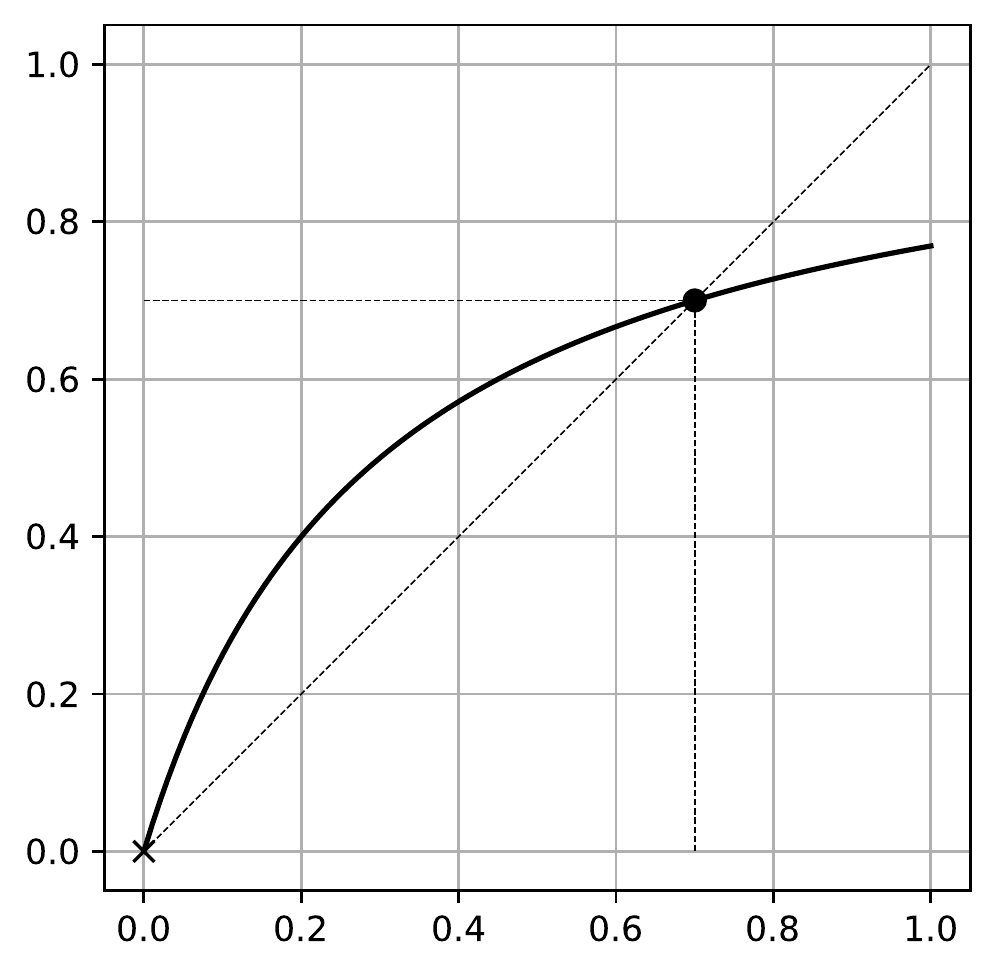}\\
    (a) $f_1$ & (b) $f_{0.3}$
\end{tabular}
    \caption{Graphs of the function $f_r$ for $r=1$ and $r=0.3$.}
    \label{fig_f_r}
\end{figure}

\begin{itemize}
\item If $r\ge 1$ then $0$ is the only fixed point of $f_r$. It is attractive. $x^k$ monotonically decreases to $0$.
\item If $r<1$ then $1-r$ is the only attractive fixed point. If $x<1-r$ then $x^k$ monotonically increases to $1-r$, otherwise it monotonically decreases to $1-r$.
\end{itemize}

For a general graph with coupled nodes dynamics through the $r_i$, one sees that there are only two 
possible combinations of couples $(x_i, r_i)$ at a node when $x$ is a fixed point:
either $x_i = 0$ and then $r_i$ can take any value, or $x_i>0$ and then $r_i = 1-x_i$.

Let us introduce the vector $l$ whose components are defined by:
\begin{equation}
l^k_i = \max \{ 0, 1-r^k_i \} \label{eqn_lik}
\end{equation}

At each iteration, $l^k_i$ is the limit to which $x_i$ would converge if its complementary weight was fixed.
Hence, as normalization is done in parallel, $x^{k+1}_i$ makes a step from $x^{k}_i$ towards $l_i^k$, i.e.
\[
\exists \alpha \in (0,1[ \quad x_i^{k+1} = (1-\alpha) x_i^{k} + \alpha l_i^{k}
\]

Intuitively, the equation \ref{eqn_lik} shows that a weight $x_i$ competes with its neighbors. 
If the sum $r_i^k$ of the weights of the neighbors of a node $i$ is large then $l_i^k$ is small and conversely, 
and the same symmetrically for the neighbors of $i$.

Now, there are configurations of weights such that these dynamics are collaborating instead of competing,
in the sense that all the $x_i$ and $r_i$ evolve monotonically in a coherent direction and thus converge. 
These configurations correspond to fast convergence regions within 
the basins of attraction of attractive fixed points.

The following theorem characterizes such a region around the
maximal independent sets of density at least $2$ of a graph:

\begin{restatable}{theorem}{thmisbasin}\label{ref_thmisbasin}
For any graph $G$ and any $S\in\MIS(G)$ such that $\DEN{S} \ge 2 $,
if $x=\R(y)$ for some (normalizable) $y$ verifies
\begin{eqnarray}
\forall i \in S & : & x_i > 1/2 \label{eqn_mis_basin_1} \\
\forall i \notin S & : & x_i < 1/(2 d_S) \label{eqn_mis_basin_2}
\end{eqnarray}
where
\[d_S = \max_{i \in S} \DEG(i)\]
then 
$\R^k(x)$  converges monotonically to $ind(S)$ (componentwise).
\end{restatable}

To illustrate this result, let us consider the MIS $(1,0,1)$ on $P_3$.
It has a density of $2$ and $d_{(1,0,1)}=1$, 
hence any point $(u,v,w)$ on the Taco such that $u>1/2$, $v<1/2$ and $w>1/2$ is 
in its basin of attraction. Whenever a point enters this region then $u$ and $w$ monotonically converge to $1$, and $v$ monotonically converges to $0$.

\section{Approximation of the maximum weight independent set problem}\label{sec_kako}
Assuming that iterating $\R_h$ always converges to the indicator vector 
of a maximal independent set of a graph, 
the question is then how well this dynamical scheme approximates {\MWIS}.

{\IGN} is closely related to a greedy approximation algorithm of {\MWIS} by Kako et al. \cite{kako2005approximation}.
Given a weighted graph $G=(A,x)$, the authors define the \emph{weighted degree} of a node $i$ by
\[
d_x(i)=\frac {\sum_{j\in \N{i}} x_j} {x_i}=\frac {\sum A_{ij}x_j} {x_i}
\]

They then greedily build a maximal independent set $S$ of $G$ by iteratively adding to $S$ 
the node of minimum weighted degree, removing this node from the graph as well as all its neighbours, 
and repeating until the graph is empty. They call this algorithm {\KAKO}.
Kako et al.'s algorithm of 2005 is a relatively recent extension to weigthed graphs 
of Hochbaum's algorithm of 1983 for the maximum independent set problem \cite{hochbaum1983efficient}, 

One sees that the 'weighted degree' defined by Kako et al. is actually a \emph{relative} weighted degree,
and it corresponds up to a constant to the inverse of the weights obtained after one iteration of normalization of the graph:
\[
(\R_i(x))^{-1} = 1 + d_x(i) 
\]

Obviously, ranking the nodes at each iteration according to 
$\frac{x_i + \sum A_{ij}x_j}{x_i} = (\R_i(x))^{-1}$ in \KAKO would give 
the same solution.
Iterative graph normalization can be thus thought of as a soft version of $\KAKO$, 
in which instead of picking the vertex of maximum relative weighted degree at each step, 
the vertices are reweighted according to the inverse 
of their relative weighted degree, hence pushing up the nodes which have a small relative weighted degree.

Kako et al. have shown that their algorithm achieves an approximation ratio  
related to the weighted inductiveness of the graph and that this bound is tight.
The weighted inductiveness of a graph is however difficult to compute in practice 
as it involves enumerating all the subgraphs of the graph.
We refer the reader to \cite{kako2005approximation} for details.

We compared the solutions found by {\IGN} and {\KAKO} on binomial random graphs $G_{n,p}$, where 
$n$ is the number of nodes and $p$ the probability of edge creation.
The weights are drawn uniformly in $(0,1]$.
Let $w(S_{\IGN})$ and $w(S_{\KAKO})$ be the weights of 
the solutions found by resp. {\IGN} and {\KAKO} on a given graph.
We compute the gap between {\IGN} and {\KAKO} by:

\begin{equation}
g = \frac{ w(S_{\IGN}) - w(S_{\KAKO})}{w(S_{\KAKO})} \label{eqn_gap}
\end{equation}

For each experiment, we sample $10 000$ random graphs.

One has to be careful on the convergence criterion of $\IGN$, 
because in some conditions, the speed $\max_i | x_i^{k+1} - x_i^k |$ can be very slow.
If we stop based on a threshold $\epsilon$ on this speed, 
we can end up with a solution which did not converge
hence is not close to a binary solution or is not maximal.
We thus add two other conditions, 
which are $\max_i \min \{ x_i, 1 - x_i \} \le \alpha$ and 
$\min_i x_i \ge 1-\alpha$, for a small $\alpha$.
The first condition expresses that $x$ must be close to binary and 
the second that it must have at least one component close to $1$.
The second condition is needed for dense large graphs for which the first iteration of 
renormalization can bring all the weights close to $0$.
In our experience, all $3$ conditions are required because 
the weights can also come close to a binary point but still finally converge to another one, 
typically come close to a non maximal independent set and then eventually saturate it.

In all the experiments reported here, we have set $\epsilon = 10^{-6}$ and $\alpha = 10^{-2}$.
With these settings, $\IGN$ always converged to an independent set and converged to a 
maximal independent set when we used a non linear activation $h$ such that $h\prime(0)>0$,
in accordance with the theorem \ref{ref_thnmisrepulsive}.
This is the reason why when using a power function for activation, 
we translate it away from $0$, to avoid a zero derivative at zero.
We otherwise get non maximal solutions, which seems to indicate that the condition 
of theorem \ref{ref_thnmisrepulsive} is necessary and sufficient.

The figure \ref{fig_kako_histo} shows histograms of the gaps found for a power activation 
of parameters $(a,t)=(2, 0.01)$
and different sizes of graphs of medium density ($G\sim G_{n, 0.5}$). 
For small graphs, {\IGN} and {\KAKO} almost always find the same solution, which corresponds to the peak 
in the histogram bin centered on $0$.
For larger graphs ($n=512$), one sees that there is still a peak at $0$ 
but that another distribution builds up, 
made up of gaps which are both positive and negative, which 
means that $\IGN$ finds worse solutions than $\KAKO$ on some instances and better solutions on other instances.

\begin{figure}[!ht]\center\begin{tabular}{ccc}
    \includegraphics[width=0.3\linewidth]{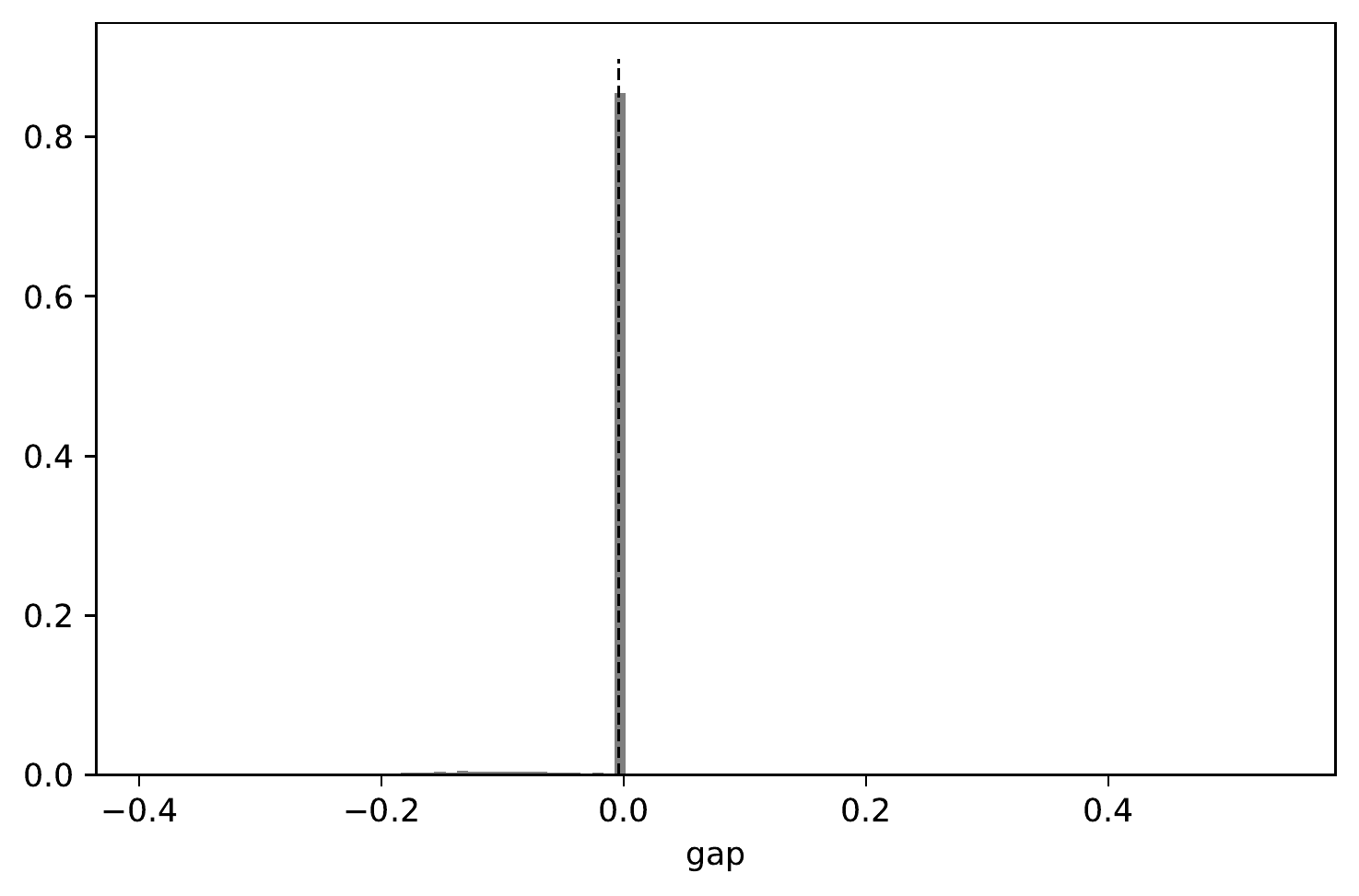}&
    \includegraphics[width=0.3\linewidth]{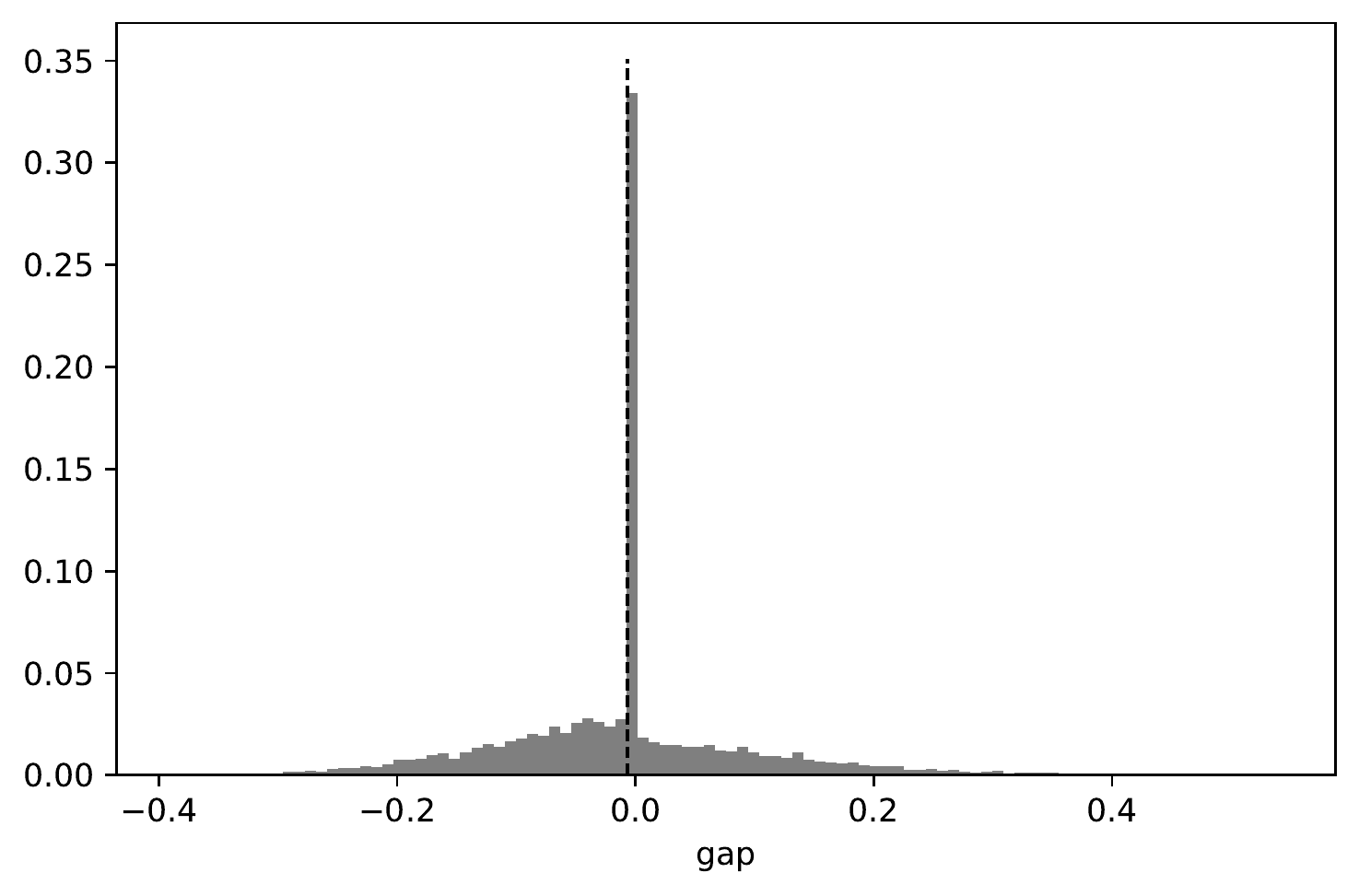}&
   \includegraphics[width=0.3\linewidth]{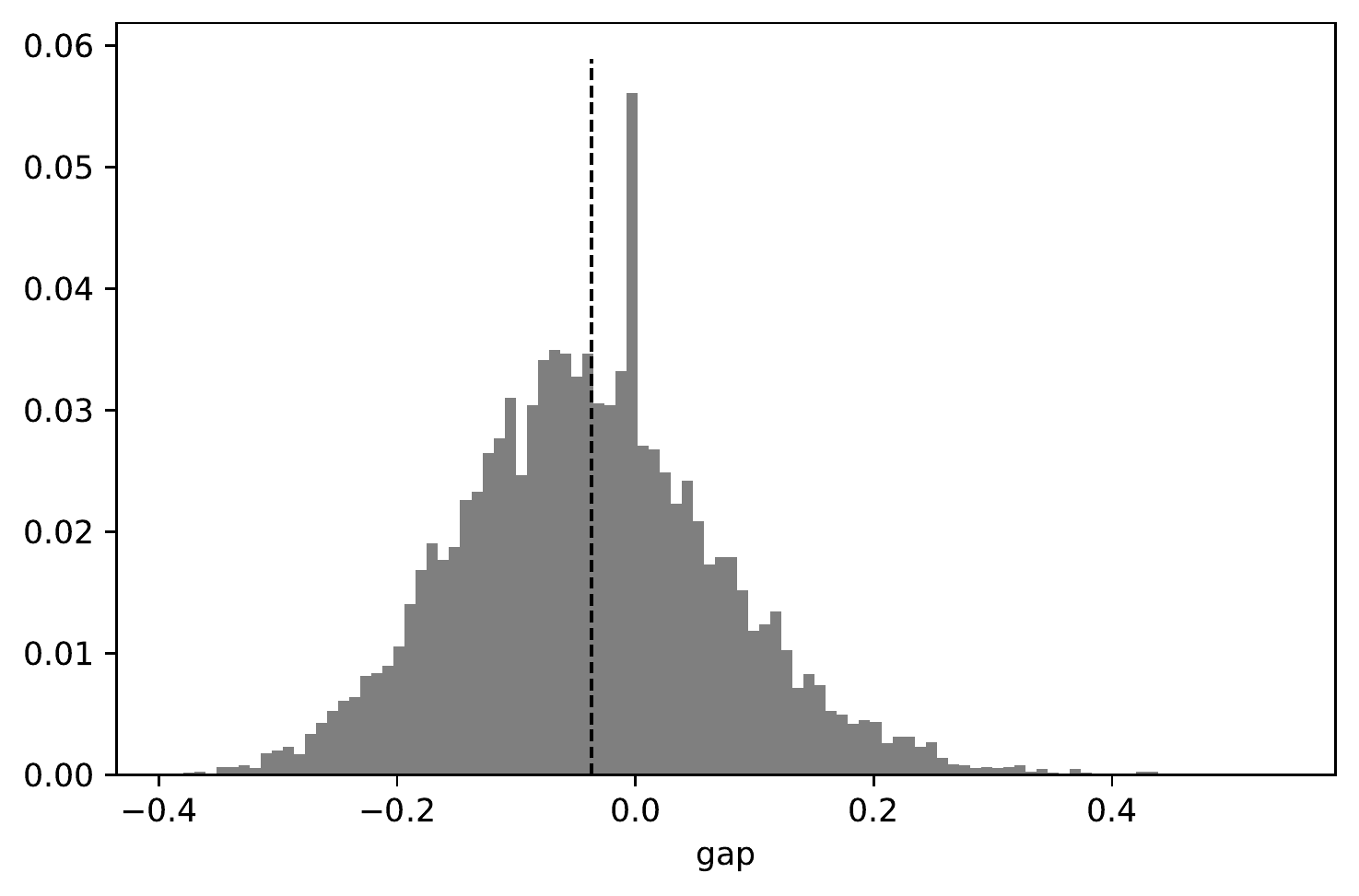}\\
   a) $n=8$ & 
   b) $n=64$ & 
   c) $n=512$ 
\end{tabular}
    \caption{Histograms of the gaps between {\IGN} and {\KAKO} for different sizes of graphs drawn from $G_{n, 0.5}$ for a power activation of parameters $(a,t)=(2, 0.01)$. The dashed vertical line represents the average gap.}
    \label{fig_kako_histo}
\end{figure}

In order to sumarize the performance of $\IGN$ w.r.t. $\KAKO$, we compute three 
statistics: the average gap, the median gap and the proportion of positive gaps.
The average gap indicates the expected difference of quality between $\IGN$ and $\KAKO$.
The median gap is a robust version of it, better accounting for the fact that the distribution of gaps has multiple modes.
The proportion of positive gaps corresponds to the chance that $\IGN$ finds a solution at least as good as $\KAKO$.

\begin{figure}[!ht]\center\begin{tabular}{ccc}
    \includegraphics[width=0.3\linewidth]{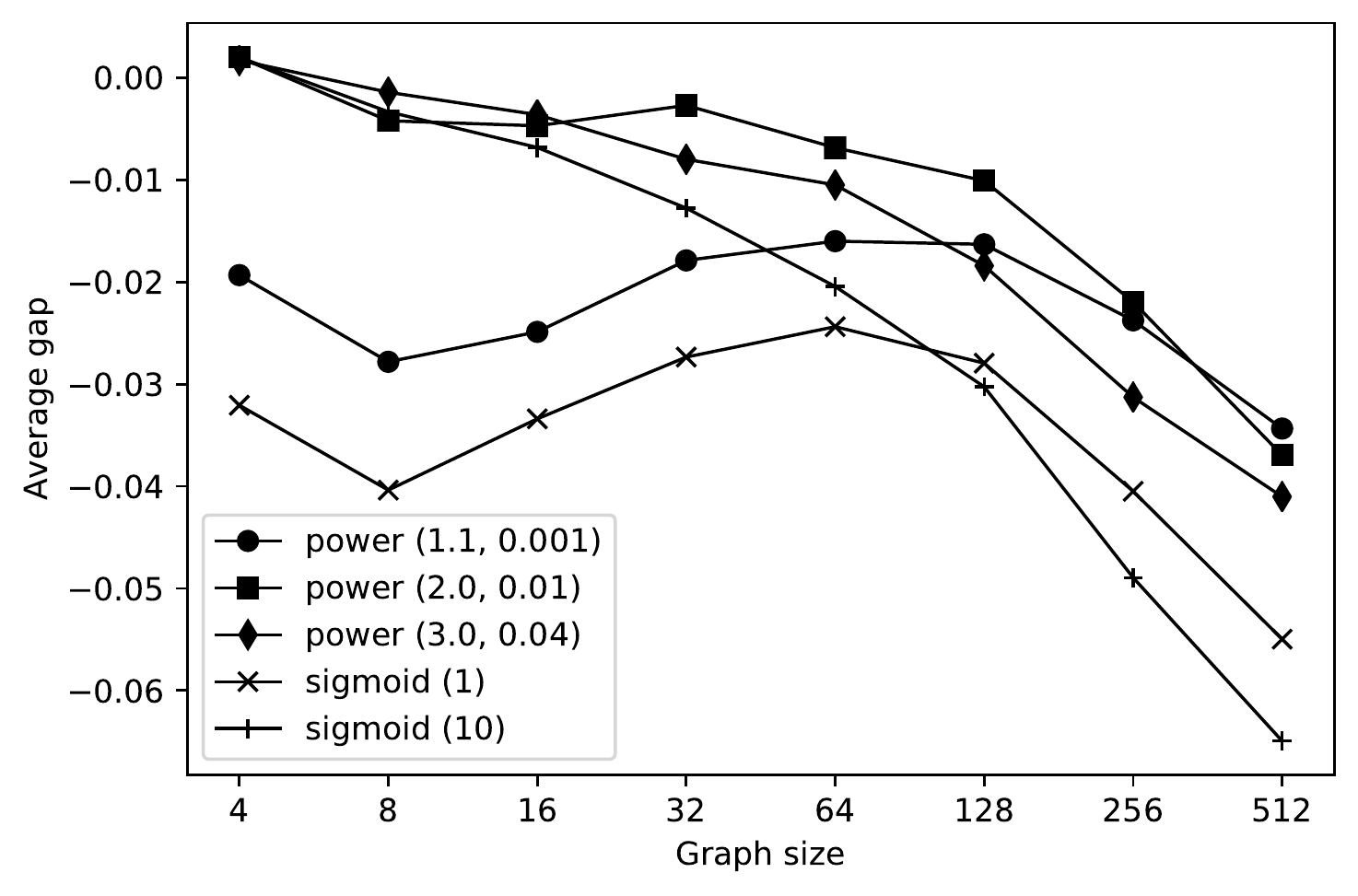}&
    \includegraphics[width=0.3\linewidth]{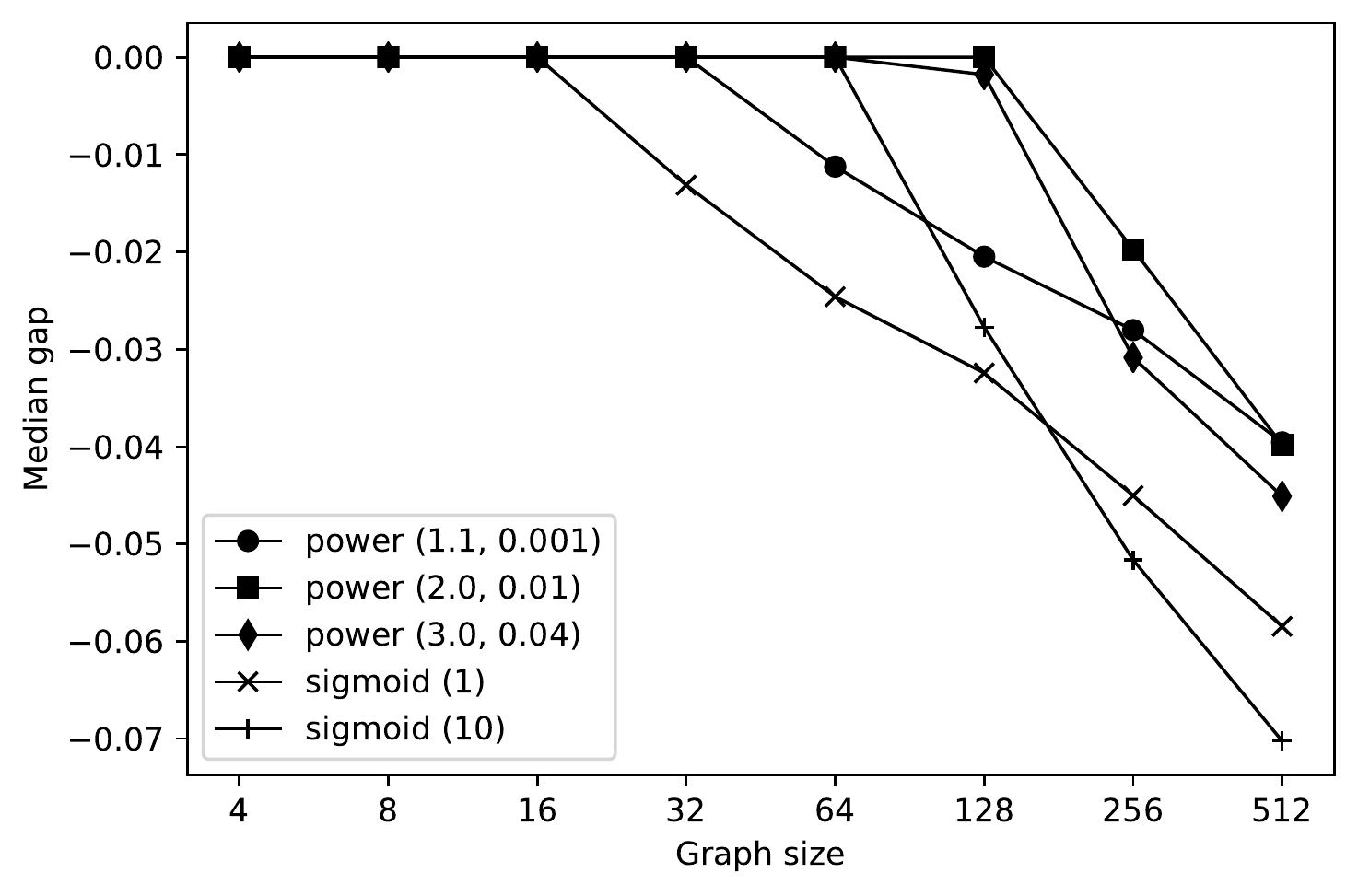}&
    \includegraphics[width=0.3\linewidth]{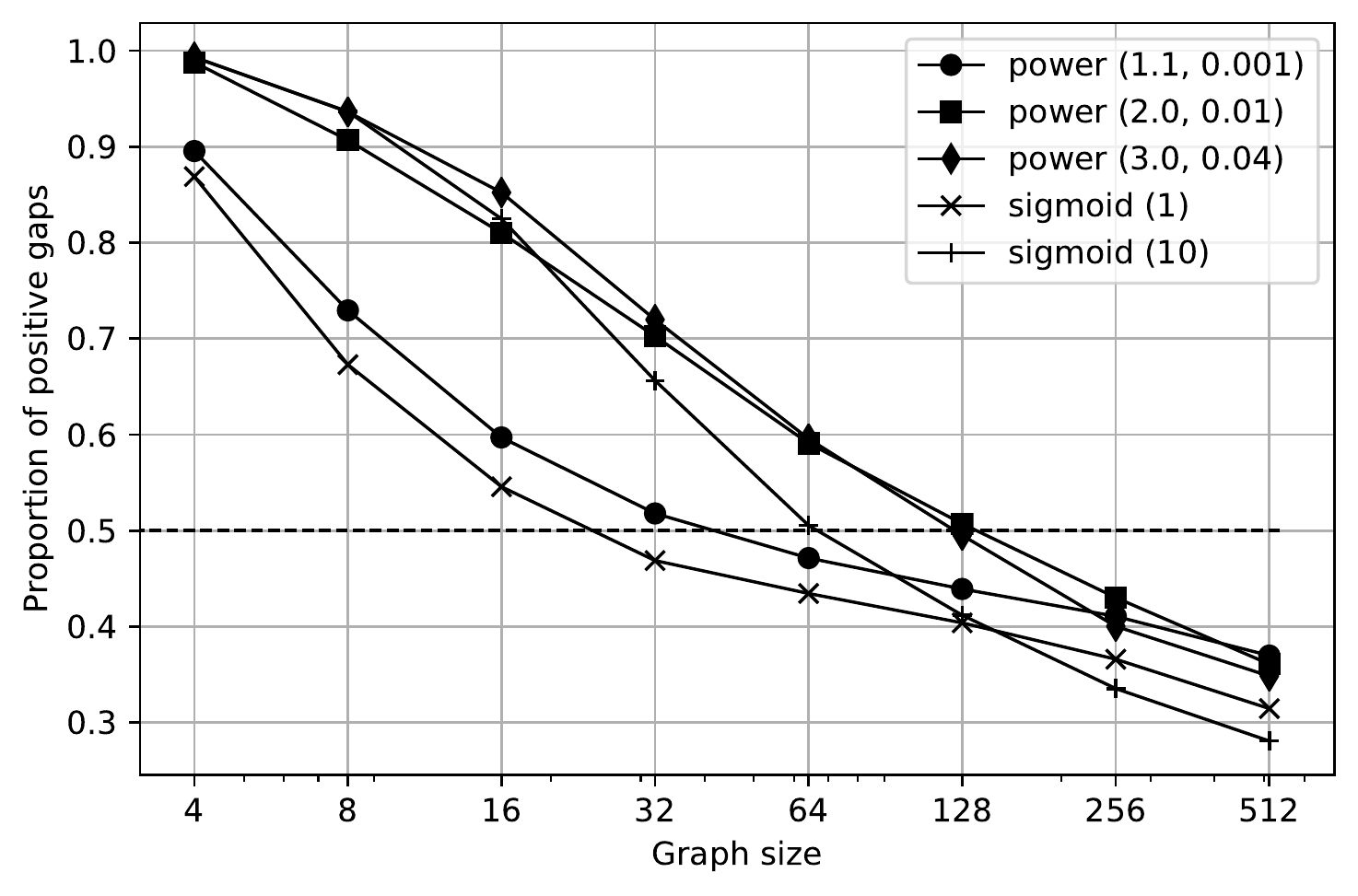}\\
   a) Average gap & b) Median gap & c) Proportion of positive gaps
\end{tabular}
    \caption{Comparison between the solutions found by $\IGN$ and $\KAKO$ for various activation functions.}
    \label{fig_kako_perf}
\end{figure}

One sees in figure \ref{fig_kako_perf} that 
$\IGN$ overall provides a slightly worse solution than $\KAKO$ \emph{on average}  
($<5\%$ smaller weight for $n\le 512$ and a suitable activation). 
However, up to graphs of $128$ nodes, 
it gives a better solution \emph{in probability} when a power activation with an exponent greater than $2$ is used.

Remark also that sigmoid activations seem to perform poorer in general than power activations.

These experiments correspond to binomial graphs of medium density, i.e. drawn according to $G_{n, 0.5}$.
We found however that the gap between $\IGN$ and $\KAKO$ almost doesn't depend on the graph density, 
as is illustrated in figure \ref{fig_kako_density}. Note that it doesn't mean that $\IGN$ and $\KAKO$ 
provide approximations of \MWIS whose quality is independent of the graph density, 
but that  $\IGN$ and $\KAKO$ behave very similarly for all graph densities, 
hence confirming the fact that they are closely related.
 
\begin{figure}[!ht]\center\begin{tabular}{ccc}
    \includegraphics[width=0.3\linewidth]{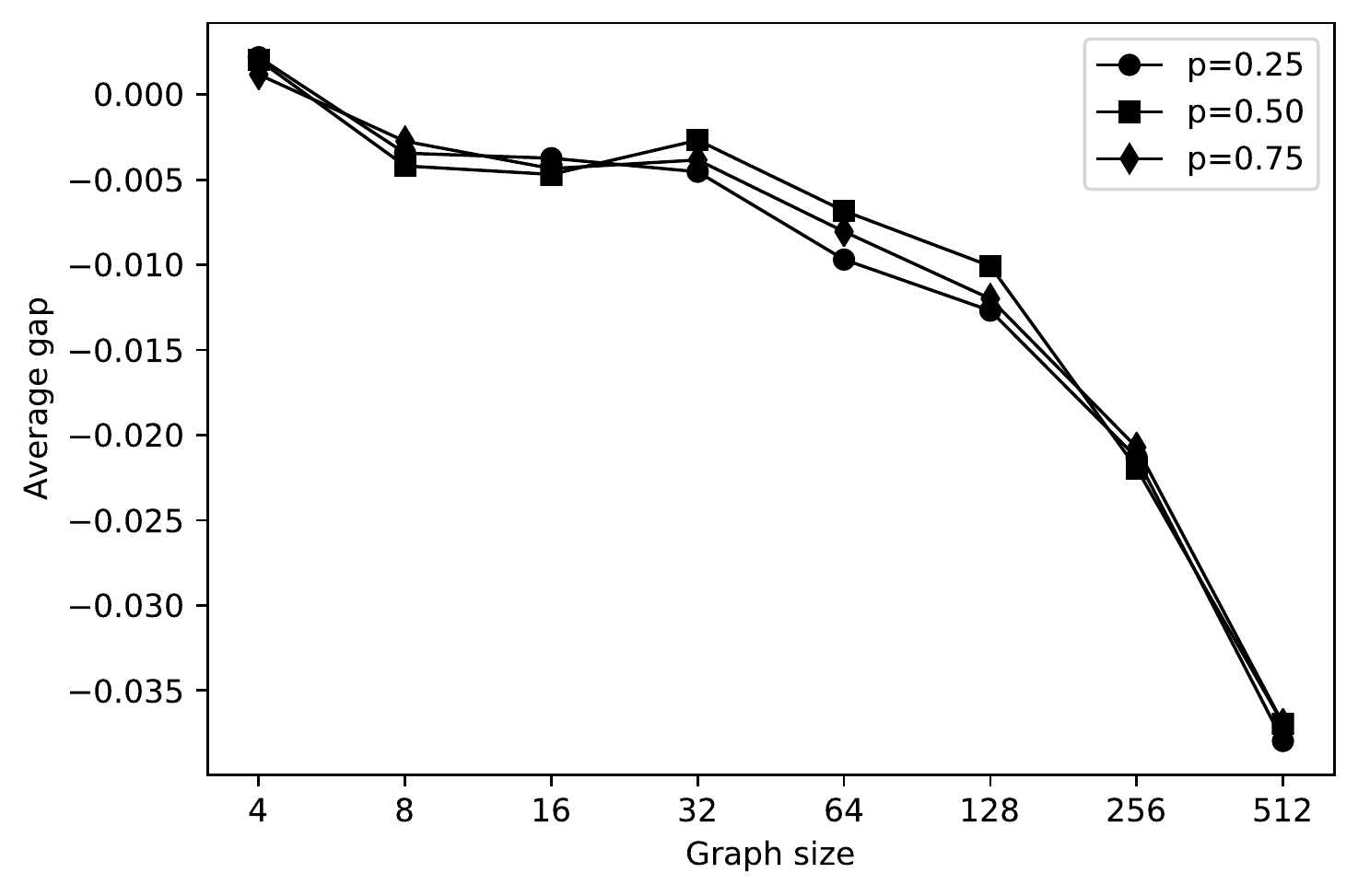}&
    \includegraphics[width=0.3\linewidth]{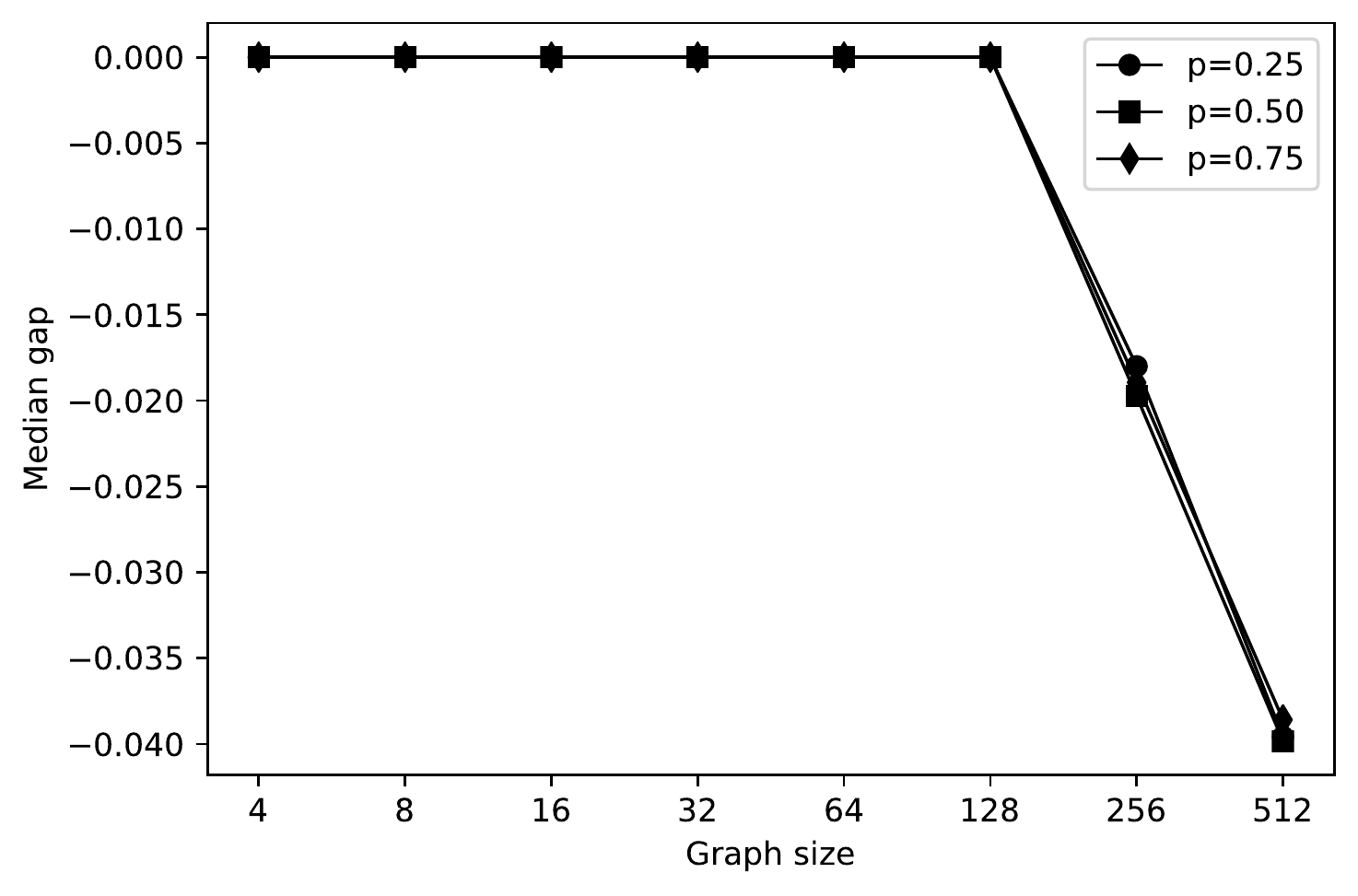}&
    \includegraphics[width=0.3\linewidth]{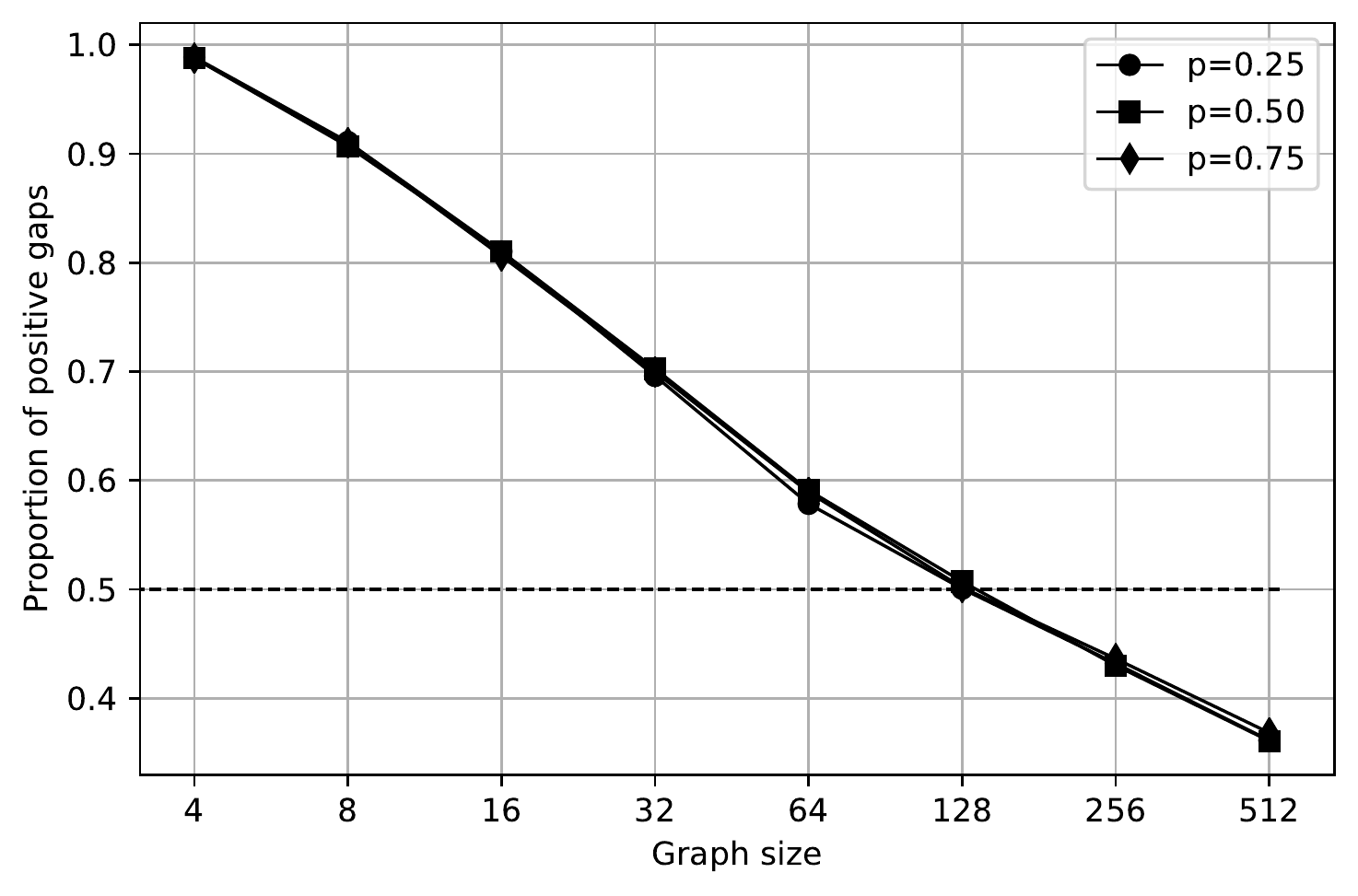}\\
   a) Average gap & b) Median gap & c) Proportion of positive gaps
\end{tabular}
    \caption{Gap between $\IGN$ and $\KAKO$ for different graph densities. Power activation $(2, 0.01)$. }
    \label{fig_kako_density}
\end{figure}

The speed of convergence of $\IGN$ according to the criteria described above is provided in 
figure \ref{fig_kako_iter} for different power activations.
One sees that increasing non linearity from a power $1.1$ to $2$ significantly speeds up convergence 
and also significantly reduce the variance of the number of iterations needed.
Comparatively, going from a power $2$ to $3$ doesn't improves much the convergence speed. 

\begin{figure}[!ht]\center\begin{tabular}{cc}
    \includegraphics[width=0.4\linewidth]{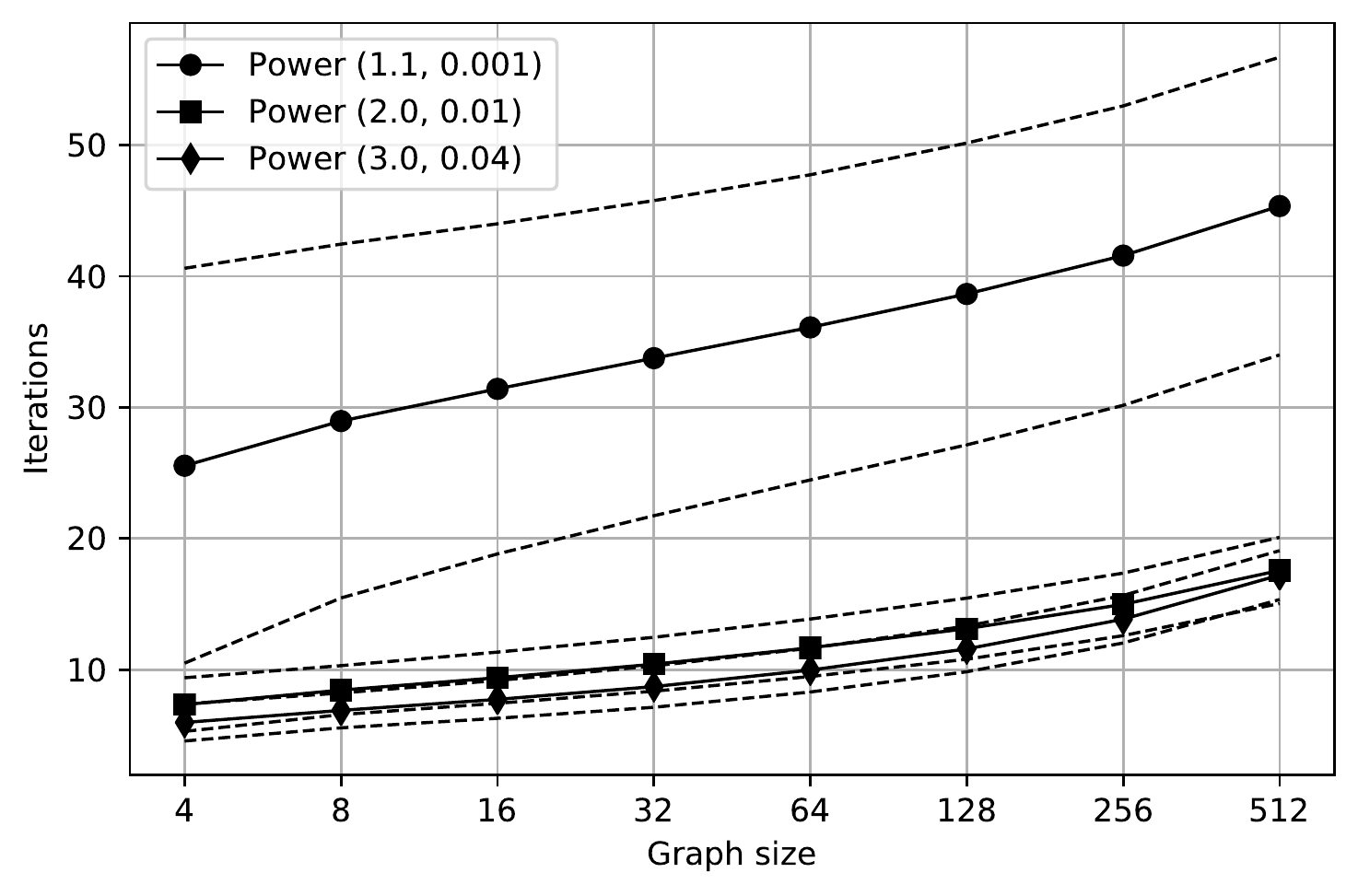}&
    \includegraphics[width=0.4\linewidth]{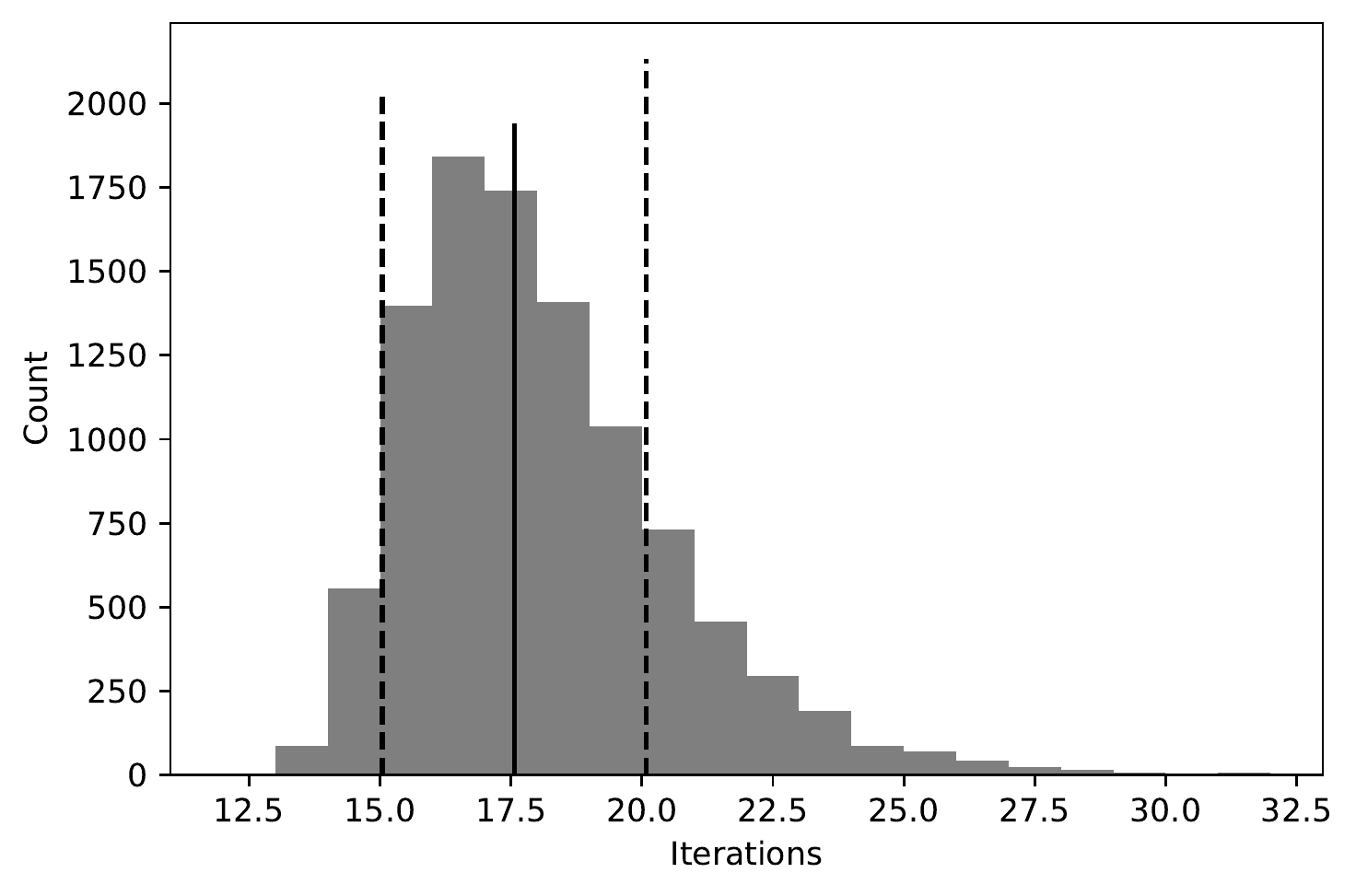}\\
   a) & b) 
\end{tabular}
    \caption{(a) Number of Iterations to converge within $10^{-2}$ distance to a binary solution for different power activations. The solid curves represent the average on $10 000$ instances in $G_{n, 0.5}$ and the dashed lines correspond to $1$ standard deviation.
    (b) Histogram of the number of iterations for the activation parameters $(2.0, 0.01)$ and graphs of size $512$.}
    \label{fig_kako_iter}
\end{figure}

 All the experiments and the figures of this section can be reproduced with the script 'kako.py'.
 
\section{The assignment problem, Sinkhorn-Knopp and Softassign}\label{sec_SK}
Let $B=(V_1,V_2,X)$ be an edge-weighted bipartite graph where $V_1=\{1\dots n\}$, $V_2=\{1\dots m\}$ and 
$X=(X_{ij})$ is a $n\times m$ matrix of nonnegative weights on the edges between $V_1$ and $V_2$.
A matching of $B$ is a set of independent edges in $B$, i.e. a set $M$ of edges without common nodes: $M=\{ (a_i,b_i) \}_{i=1\dots p}$  with $a_i \in V_1$ and $b_i \in V_2$ and $\forall i\ne j$ : $a_i \ne a_j$ and $b_i \ne b_j$.
The Maximum Weight Bipartite Graph Matching problem, or Assignment Problem (\AP), is to find a matching $M$ of maximum weight, 
i.e. maximizing $\sum_{(a,b)\in M} X_{a,b}$ .

As mentioned earlier, \AP is a particular instance of \MWIS in the weighted graph $G$ which has one node per edge in $B$ and one edge between two nodes if the corresponding two edges in $B$ share a common node,
and it can be solved in polynomial time, for example by the Hungarian algorithm \cite{kuhn1955hungarian}.

Please remark that any \AP problem such that $m<n$ can be turned into a balanced problem in 
which $m=n$ by adding $n-m$ dummy nodes to $V_2$, each connected with a zero weight to all the nodes 
of $V_1$. In the following, we thus assume $m=n$. 
In this case, $X$ is a square matrix and $B$ admits perfect matchings, 
i.e. matchings made up of edges which hit all the nodes of the graph.
As previously, let us assimilate sets and indicator vectors. 
In the bipartite graph matching context, the $n\times n$ binary matrices are the indicators of sets of edges.
The indicator matrix of a matching cannot have two $1$ on the same row or on the same column, 
because two $1$ on the same row (resp. column) represent $2$ edges sharing the same node in $V_1$ (resp. $V_2$). A perfect matching is a maximal matching whose indicator matrix has exactly one $1$ per row and per column, hence is a permutation matrix.

We denote by $\PERM_n$ the set of the permutation matrices of size $n$ 
and by $\BIRK_n$ the Birkhoff polytope, i.e the set of doubly stochastic matrices, 
which are the positive matrices normalized to $1$ both in rows and columns.
The Birkhoff–Von Neumann theorem \cite{birkhoff1946three, von1953certain} says 
that $\BIRK_n$ is the convex hull of $\PERM_n$. 

Let's define the following operators on a square nonnegative matrix $A$:
\begin{center}
\begin{tabular}{lcl}
Row normalization &:& $ \RoN(A) =  A \HD A \ONE \ONE^t $ \\
Column normalization &:& $ \CoN(A) =  A \HD \ONE \ONE^t A$ \\
Cross normalization &:& $ \CrN(A) =  A \HD ( A \ONE \ONE^t + \ONE \ONE^t A - A )$ \\
\end{tabular}
\end{center}
 
 Each one with its own domain of definition.
 We see in particular that the set of the cross-normalizable matrices is larger than the 
 set of either row or column normalizable matrices. 
We say that $A$ is cross normalized if $\CrN(A)=A$, i.e. if it is a fixed point of cross normalization. 
We also denote by $\CrN_h$ the activated cross normalization, i.e. $\CrN$ followed by 
a non-linear activation as previously.
 
The Sinkhorn-Knopp algorithm (\SK) corresponds to the iteration of $\CoN\circ \RoN$ (or $\RoN\circ \CoN$), 
and we call Iterative Cross Normalization (\ICN) the iteration of $\CrN_h$ on a matrix.

Let $vec(A)$ denote the vectorization of a matrix $A$, obtained by stacking its columns in a vector of size $n^2$.
Vectorizing cross normalization gives:
\[
vec(\CrN(A)) = vec(A) \HD \left[ \left(\ONE \ONE^t \otimes I_{n} + I_{n} \otimes \ONE \ONE^t - I_{n^2}\right) vec(A)\right]
\]
where $\otimes$ is the Kronecker product.

Which shows that matrix cross normalization is a special case of graph normalization: 
\[
vec(\CrN(A)) = vec(A) \HD (C+I_{n^2})vec(A)= \R_C(vec(A))
\]
where $C = \ONE \ONE^t \otimes I_n + I_n \otimes \ONE \ONE^t - 2I_{n^2}$ 
is the $n^2 \times n^2$ adjacency matrix of the dual of the complete bipartite graph $K_{n,n}$, 
in the hypergraph node/edge duality sense.
$C+I_{n^2}$ is a symmetric block circulant matrix, with circulant $n \times n$ blocks:
\begin{center}
$\begin{pmatrix}
\ONE \ONE^t & I & \cdots & I \\
I & \ONE \ONE^t & & \vdots \\
\vdots & & \ddots & I \\
I & \cdots & \ONE \ONE^t & I
\end{pmatrix}$
\end{center}

Remark that the graph $C$ is $(2n-2)$-regular.

We can now specialize our general results on graph normalization to the assignment problem:

\begin{itemize}
\item Matrix cross normalization commutes with pre or post multiplication by a permutation matrix:
$\forall (M,N)\in\PERM^2 : \CrN(MAN) = M\CrN(A)N$. Indeed, any permutation of $V_1$ or $V_2$ gives an automorphism of the graph.
\item A matching is cross normalizable if and only if it is maximal, i.e. is a permutation matrix.
\item All the maximal matchings, i.e. all the permutation matrices, are attractive fixed points of cross normalization, even for a linear activation. It follows from the corollary \ref{ref_coro_attraction} and the fact that the associated density of a permutation matrix is $2$, in the sense of definition \ref{def_density}.
\item For any $P\in \PERM$, any matrix $X$ such that
\begin{eqnarray*}
\textrm{whenever} \quad P_{ij} = 1 & \textrm{then} & X_{ij} > \frac{1}{2}, \\
\textrm{and whenever} \quad P_{ij} = 0 & \textrm{then} & X_{ij} < \frac{1}{4n-4}
\end{eqnarray*}
is in the attraction bassin of $P$. $\CrN^k(X)$ converges monotonically to $P$ componentwise (theorem \ref{ref_thmisbasin}). Furthermore, for the case of the assignment problem, $P$ corresponds to 
the optimal assignment associated with $X$, hence $\ICN$ exactly solves $\AP$ for such an initialization
(obviously, any other permutation than $P$ provides a smaller total weight).
\end{itemize}

A relevant algorithm in our context is the Softassign algorithm ($\SA$),
introduced by Kosowsky and Yuille \cite{kosowsky1994invisible}.
$\SA$ corresponds to applying $\SK$ to an exponentiated version of the initial weight matrix:
\[
\SA_\tau(X) = \SK(\exp\left(X / \tau)\right)
\]
Using a statistical physics approach, the authors show that as $\tau$ goes to $0$, 
$\SA_\tau$ converges to the optimal solution of $\AP$. 
$\tau$ can be interpreted as a temperature parameter in a deterministic annealing framework. 
Relating it to the auction algorithm of Bertsekas \cite{bertsekas1992auction} and making an 
economic interpretation of it, Kosowsky and Yuille called their algorithm the "invisible hand" in their introducing paper. It has then been renamed to "Softassign" in later work \cite{gold1996softmax, rangarajan1997convergence}.
It is obviously related to the Softmax operator, discussed earlier. 

Mena et al. have recently proven that $\SA_\tau$  corresponds to the solution of 
an entropy-regularized assignment problem, which is very insightful\cite{mena2018learning}.
Given a doubly stochastic matrix $D$, define its entropy 
by $h(D) = -\sum D_{ij} \log D_{ij}$. 
Mena et al. then show that 
\[
\SA_\tau(X)= \arg\max_{D\in \BIRK}  \quad Tr(D^t X) -  h(D) / \tau.
\]
and that $\SA_\tau(X)$ converges almost surely to the optimal assignment as $\tau$ goes to $0$ 
when the $X_{ij}$ are i.i.d. samples from a distribution that is absolutely continuous with respect to the Lebesgue measure in $\mathbb{R}$\footnote{The authors actually state this result with an opposite sign on the entropy term, which is wrong. They apparently made a sign error in their paper.}.

This result makes total sense remarking 
that the entropy is null if and only if $D$ is binary and that it is strictly positive otherwise.
Regularizing the maximum assignment by the \emph{opposite} of the entropy of the solution thus
penalizes non binary solutions
and when the weight of the regularization term goes to infinity, 
it enforces a binary solution.
On its side, the $\SK$ operator projects on the Birkhoff polytope, i.e. ensures double stochasticity.
Both constraints are thus forcing to pick a solution from the set of permutations.
Although we didn't prove it, we think that any regularizer with a similar property than the entropy 
would yield a similar result, such as a quadratic penalty $q(D) = -\sum D_{ij} ( 1-D_{ij})$, 
which has for example been used (component-wise) to enforce a binary solution in 
some quadratic programming formulations of \MWIS.

A problem though with the $\SA$ algorithm is that it provides an optimal solution 
in the form of a limit on a temperature/entropy regularization parameter, 
which is not accessible in practice. 
Setting $\tau$ large enough to avoid numerical overflow, i.e. of the order of $0.01$ for 
initial weights in $[0,1]$, which gives exponentiated weights up to $\exp(100)\sim 10^{43}$, 
one still gets soft assignments after $\SK$ convergence (see figure \ref{fig_SK_ICN}).
Furthermore, the solution is not guarantied to be close to binary, 
and contains in practice competing values, close to $0.5$, in rows or columns.
Most of the weights have converged to values close to zero however a number of 
ambiguous matches remain.
All authors mention having to apply  a final "clean-up" procedure 
in order to get to a binary solution.
In \cite{gold1996graduated}, Gold and Rangarajan mention that ”a clean-up heuristic is necessary because the algorithm does not always converge to a permutation matrix”.
In \cite{zanfir2018deep}, Zanfir and Sminchisescu use SK iterations, in what they call a Bi-Stochastic Layer, 
in order to perform deep learning for the graph matching problem, 
and end up with a crisp matching using a voting scheme.
Similarly, Wang et al. \cite{wang2019learning} use what they call Sinkhorn layers for a similar task 
and they mention that ”for testing, Hungarian algorithm is performed as a post processing step to discretize output into a permutation matrix”. 

\begin{figure}[!ht]\center\begin{tabular}{lllll}
	\includegraphics[width=0.17\linewidth]{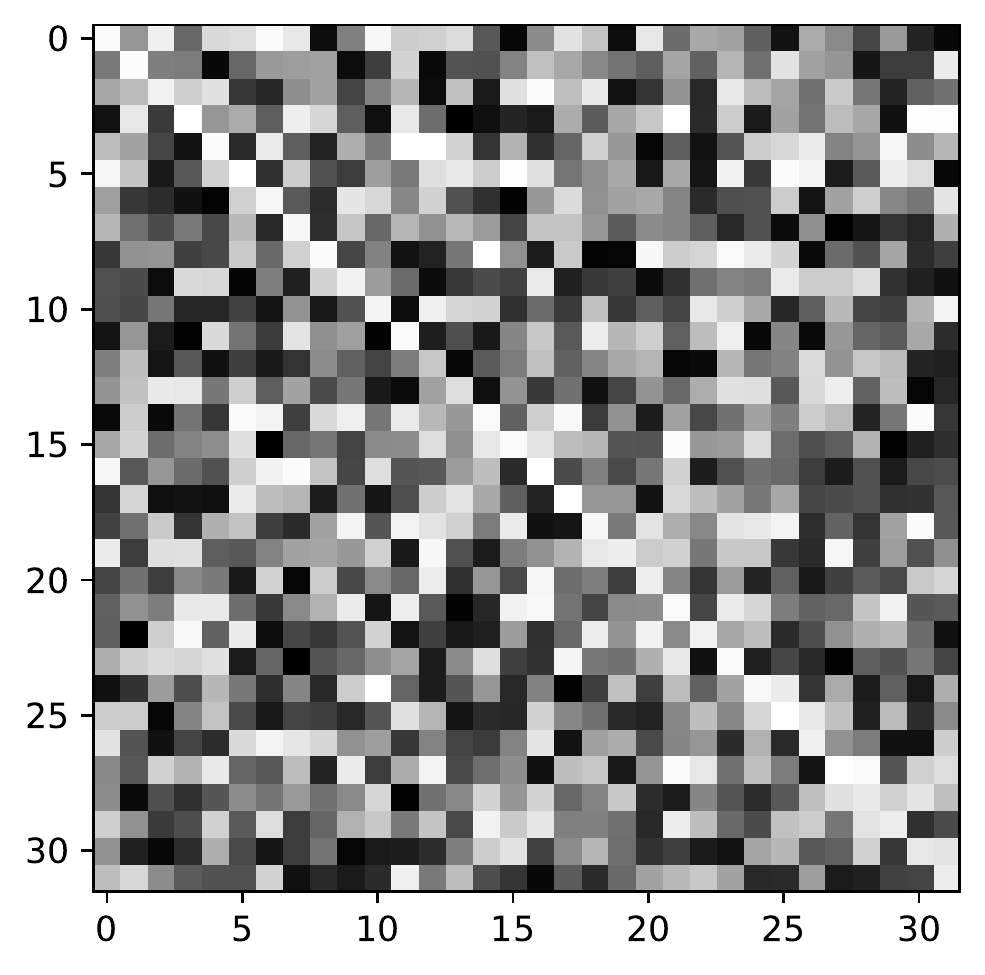}& 
	\includegraphics[width=0.17\linewidth]{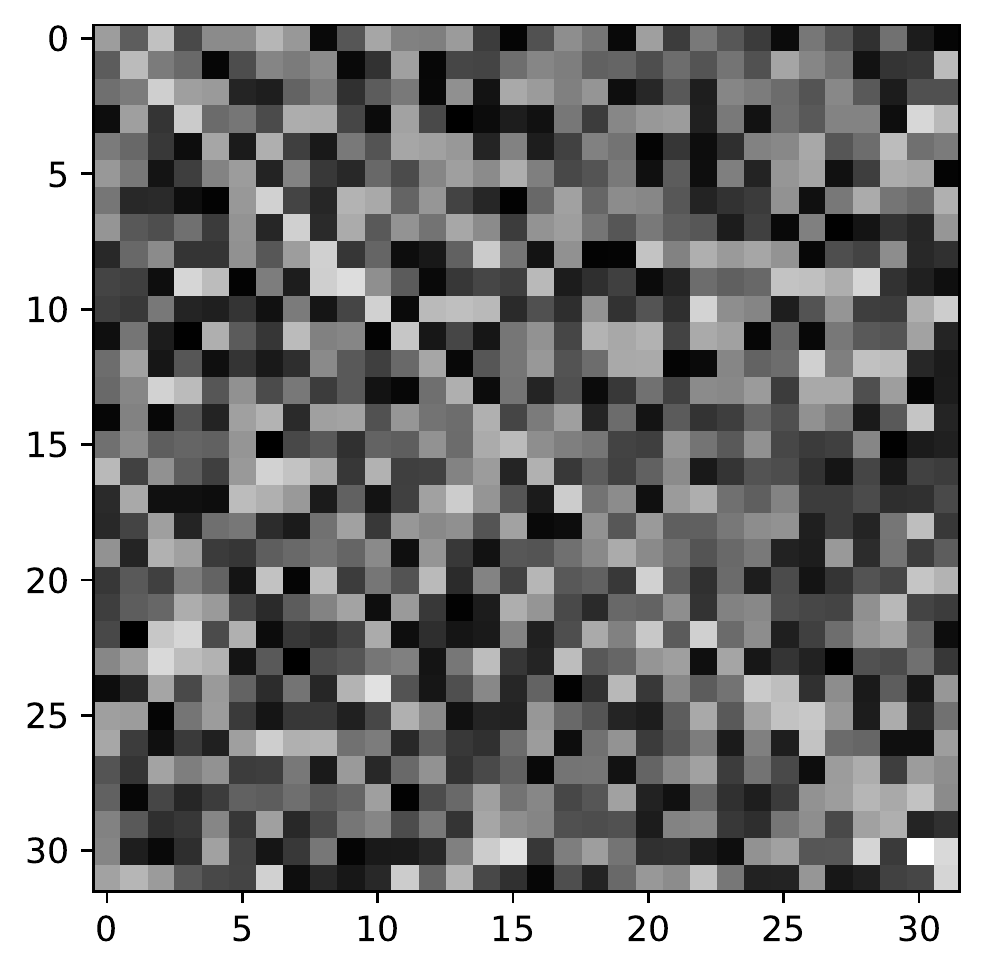}& 
 	\includegraphics[width=0.17\linewidth]{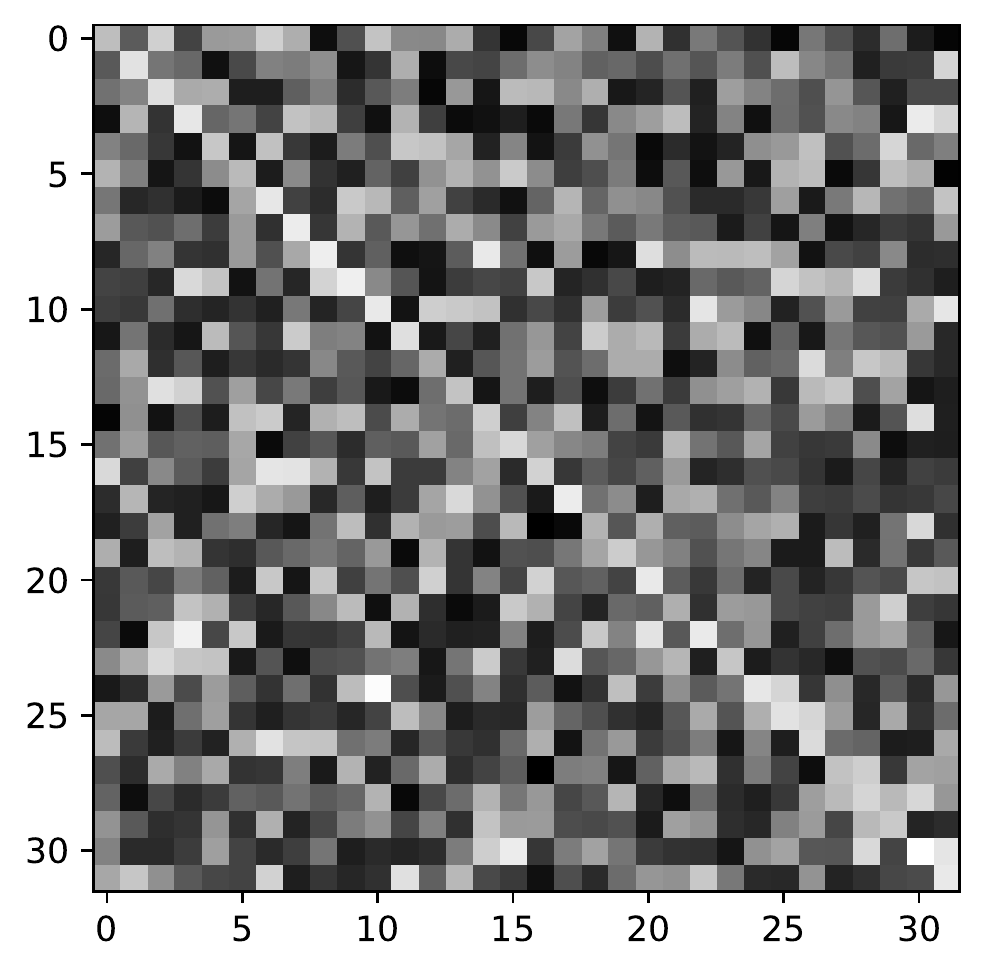}& 
 	\includegraphics[width=0.17\linewidth]{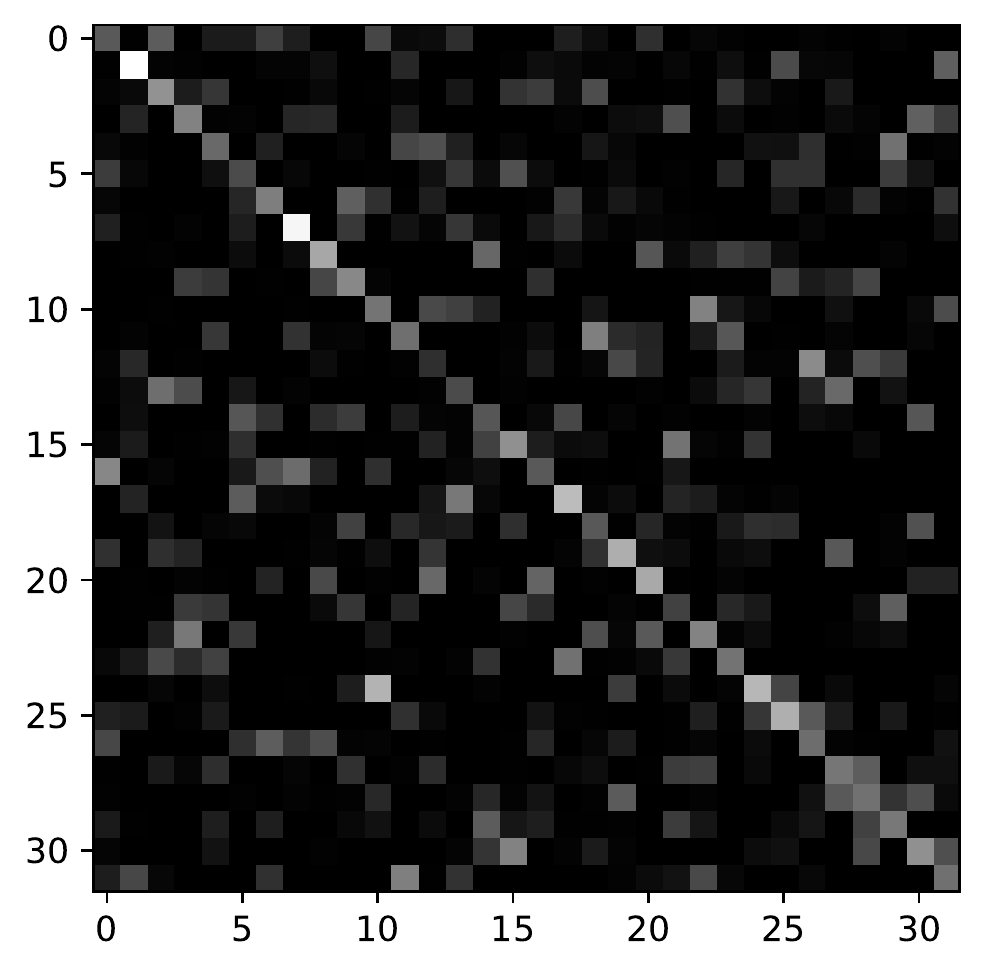}& 
 	\includegraphics[width=0.17\linewidth]{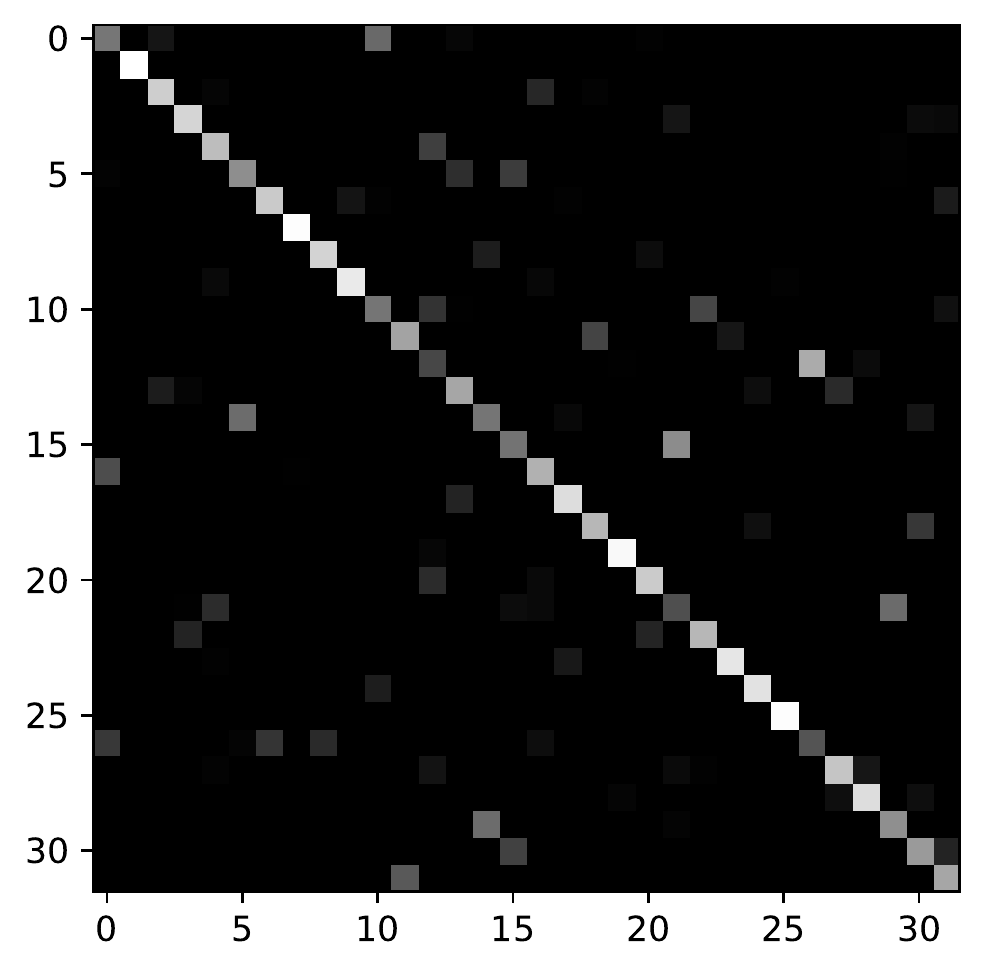}\\
	$(a) \quad A$ & $(b) \quad \SK(A)$ & $(c) \quad \SA_1(A)$ & $(d) \quad \SA_{0.1}(A)$ & $(e) \quad \SA_{0.01}(A)$
	\end{tabular}
    \caption{Typical results of Softassign as a function of $\tau$. (a) Initial weight matrix. The optimal assignment is diagonal. 
    (b) Result of $\SK$ (no exponentiation). (c)-(e) Result of Softassign ($\SA_\tau$) for different values of $\tau$.
    The iterations were stopped when the matrix was doubly-stochastic to a precision of $10^{-2}$.} 
    \label{fig_SK_tau}
\end{figure}

In our interpretation, these reasons explain why we have never seen any quantitative data on 
the optimality gap of an end-to-end assignment solution based on $\SA$, 
which must includes the necessary final "clean-up" step. 
Indeed, any final clean-up which does not ensure to get a permutation matrix,
 such as taking the maximum weight assignment per row, is invalid and can give 
 a "solution" which has a larger total weight than the optimal assignment.
On another hand, using an optimal assignment algorithm as the final step, 
such as the Hungarian algorithm, in order to evaluate the optimality gap of $\SA$ also
doesn't make sense as just skipping $\SA$ would always provide the optimal solution.

As a conclusion, in practice, $\SA$ alone doesn't provide a valid end-to-end answer to $\AP$
which can be rigorously benchmarked against the optimum.

On the contrary, if our conjectures are correct, 
$\ICN$ with a non-linear activation always provides a permutation matrix, 
by iterating a numerically stable discrete dynamical system.

A possibility is then to chain both algorithms: 
run $SA_\tau$ with a small $\tau$ to get to a doubly stochastic matrix which is close to a permutation
and then crisp the assignment by $\ICN$ with a non-linear activation.
We compared both approaches: running $\ICN$ from the initial matrix or running $\SA_\tau$ on 
the initial matrix, followed by $\ICN$, which we denote by $\SA_\tau + \ICN$ (see figure \ref{fig_SK_ICN}).

\begin{figure}[!ht]\center\begin{tabular}{ccc}
    \begin{tabular}{l}\includegraphics[width=0.17\linewidth]{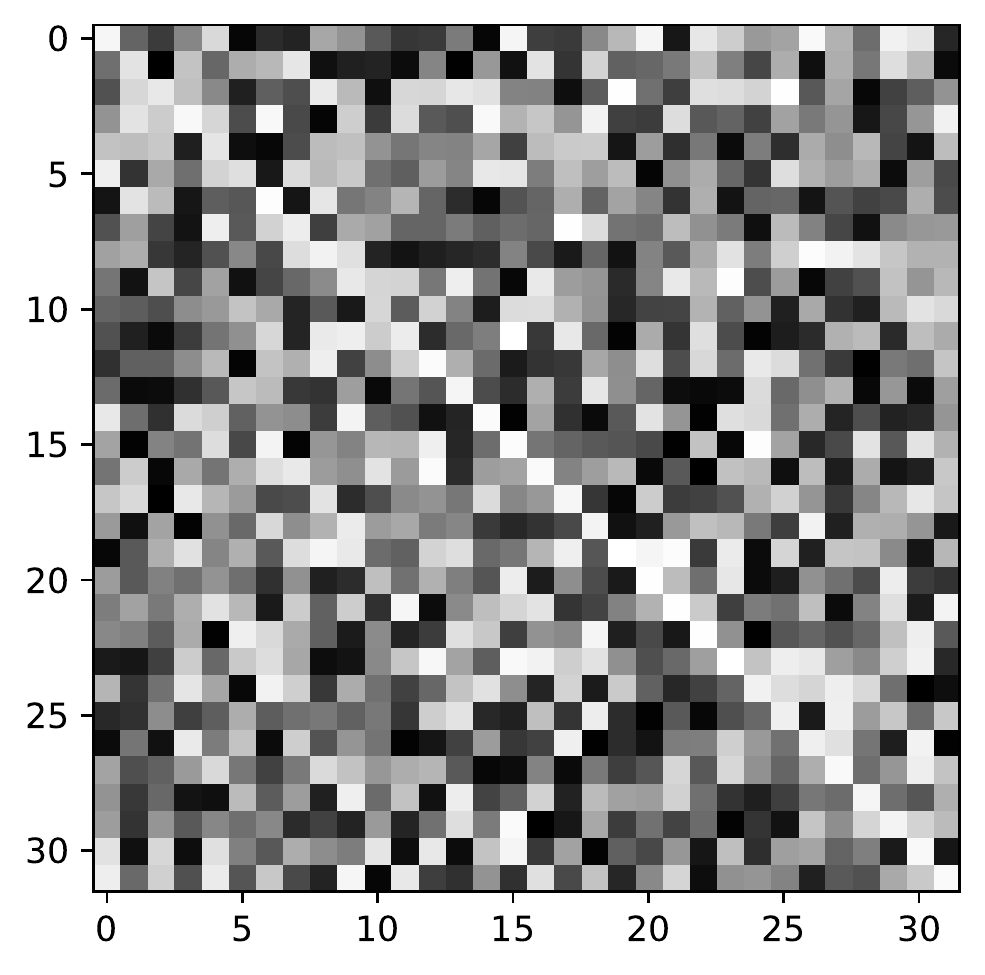}\end{tabular}& 
    \begin{tabular}{c}$\SA_\tau$ \\ $\xrightarrow{\makebox[0.08\linewidth]{}}$
    \end{tabular} &
    \begin{tabular}{l}\includegraphics[width=0.17\linewidth]{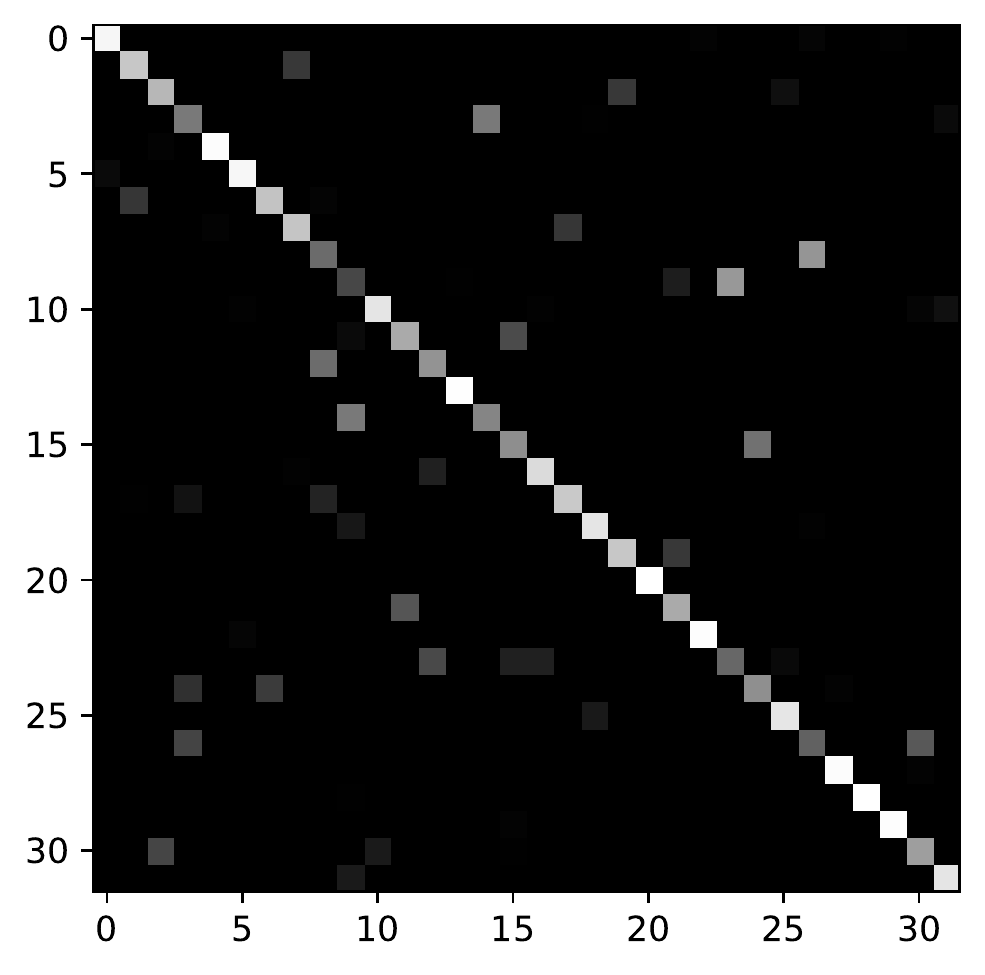}\end{tabular}\\
    \begin{tabular}{l}$\ICN\left\downarrow\rule{0cm}{0.04\linewidth}\right.$\end{tabular}&& 
    \begin{tabular}{l}$\ICN\left\downarrow\rule{0cm}{0.04\linewidth}\right.$\end{tabular}\\
    \begin{tabular}{l}\includegraphics[width=0.17\linewidth]{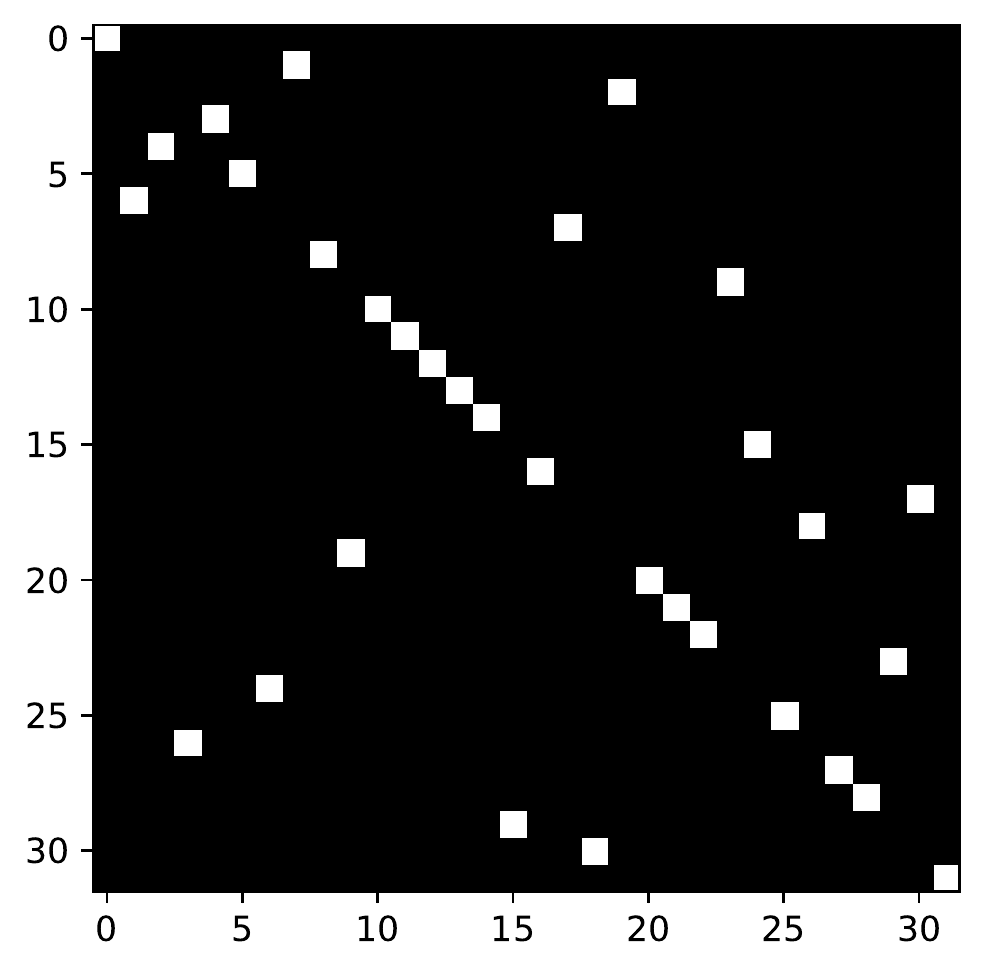}\end{tabular}& &
    \begin{tabular}{l}\includegraphics[width=0.17\linewidth]{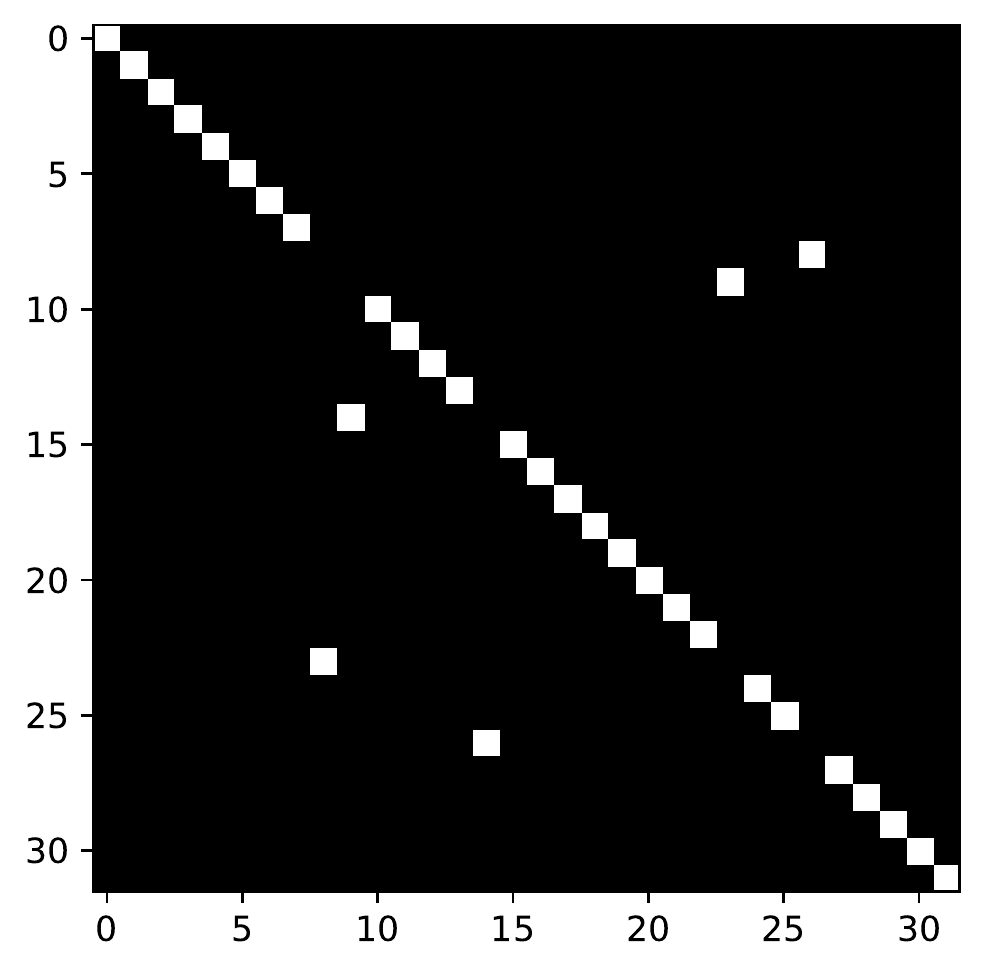}\end{tabular}
\end{tabular}
    \caption{Two solutions to get an end-to-end differentiable algorithm for the assignment problem.
    1) Directly apply \ICN to the original weight matrix (left). 
    2) Chain \SA and \ICN (right).
    Our experiments show that chaining the two algorithms performs 
    better for small graphs (such as the one in this figure) than directly applying \ICN to the initial weights 
    but at the expense of more iterations,
    and that both algorithms perform equally well for large graphs but that chaining 
    saves an order of magnitude of iterations.}
    \label{fig_SK_ICN}
\end{figure}

In all the experiments reported here, we sample the initial weight matrices by drawing 
each entry uniformly and independently in $[0,1]$.
$\SK$ iterations are stopped when the matrix is doubly-stochastic to a precision of $10^{-2}$, 
i.e. when $\forall i : |\sum_j A_{ij}-1| < 10^{-2}$ and $\forall j : |\sum_i A_{ij}-1| < 10^{-2}$.
$\ICN$ is stopped as soon as binarizing it by thresholding at $0.5$ gives a permutation matrix.
We sample $1000$ random matrices for each experiment.
We compute the optimal assignment by the Hungarian algorithm and we compute the 
optimality gap as in equation \ref{eqn_gap}.

The figure \ref{fig_SK_ICN_tau} shows the median and average gaps obtained for various 
problem sizes and values of $\tau$ for a sigmoid activation of parameter $a=5$.

\begin{figure}[!ht]\center\begin{tabular}{cc}
	\includegraphics[width=0.47\linewidth]{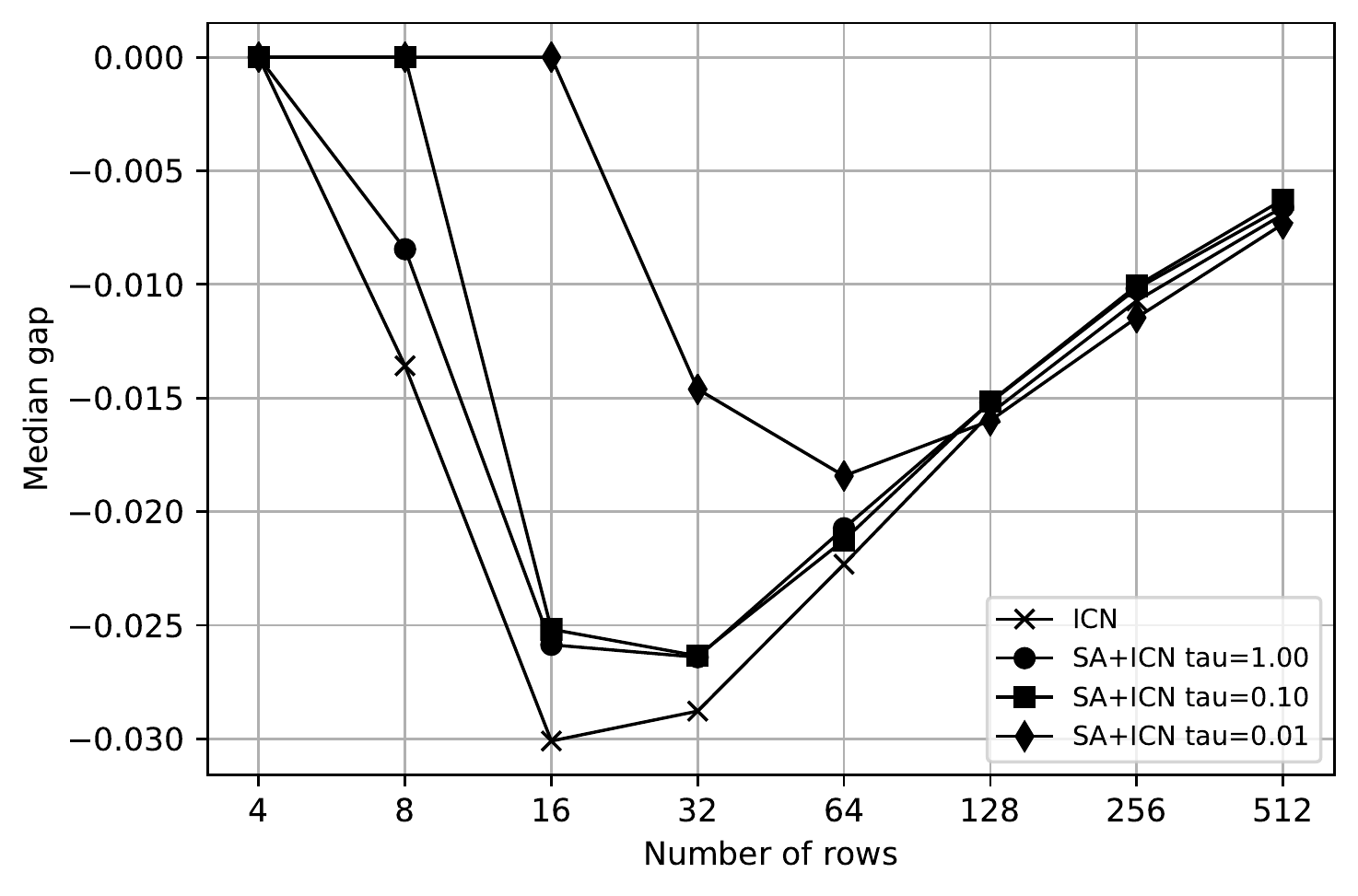}& 
	\includegraphics[width=0.47\linewidth]{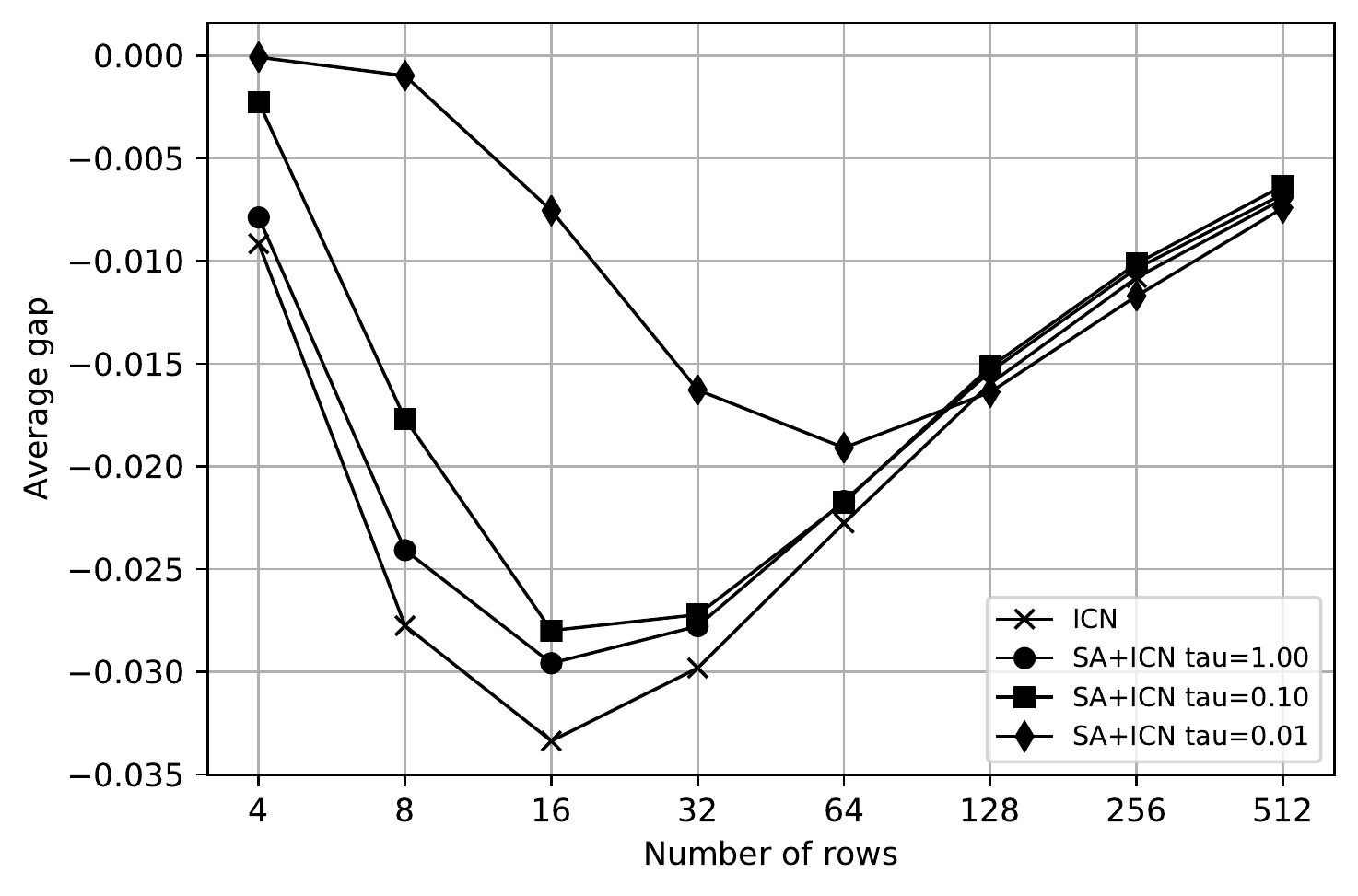}\\ 
	(a) Median gap & (b) Average gap
	\end{tabular}
    \caption{(a) Median and (b) average optimality gap for $\ICN$ and $\SA_\tau$ followed by $\ICN$ for different values of $\tau$. Sigmoid activation of parameter $a=5$.} 
    \label{fig_SK_ICN_tau}
\end{figure}

For matrices of size $n\le 64$, a smaller $\tau$ gives a better solution, as expected.
In particular, for $n\le 16$ and $\tau=0.01$, $\SA+\ICN$ with this activation exactly 
solves the assignment problem with a probability larger than $0.5$ (the median gap is zero).
Overall, the median optimality gap is lower than $3\%$ and the average optimality gap lower than $3.5\%$.
It is quite interesting also that for any $\tau$ the absolute gap first increases with the graph size but then decreases. Furthermore, for graph of size $128$ and more, 
all the median and average gaps are very close, independently of $\tau$ and 
for both $\ICN$ alone or $\SA+\ICN$.

The figure \ref{fig_SK_ICN_iterations} provides
the corresponding number of iterations of $\ICN$ alone, of $\SA_{\tau}$ and of $\ICN$ after $\SA_{\tau}$.

\begin{figure}[!ht]\center\begin{tabular}{c}
	\includegraphics[width=0.6\linewidth]{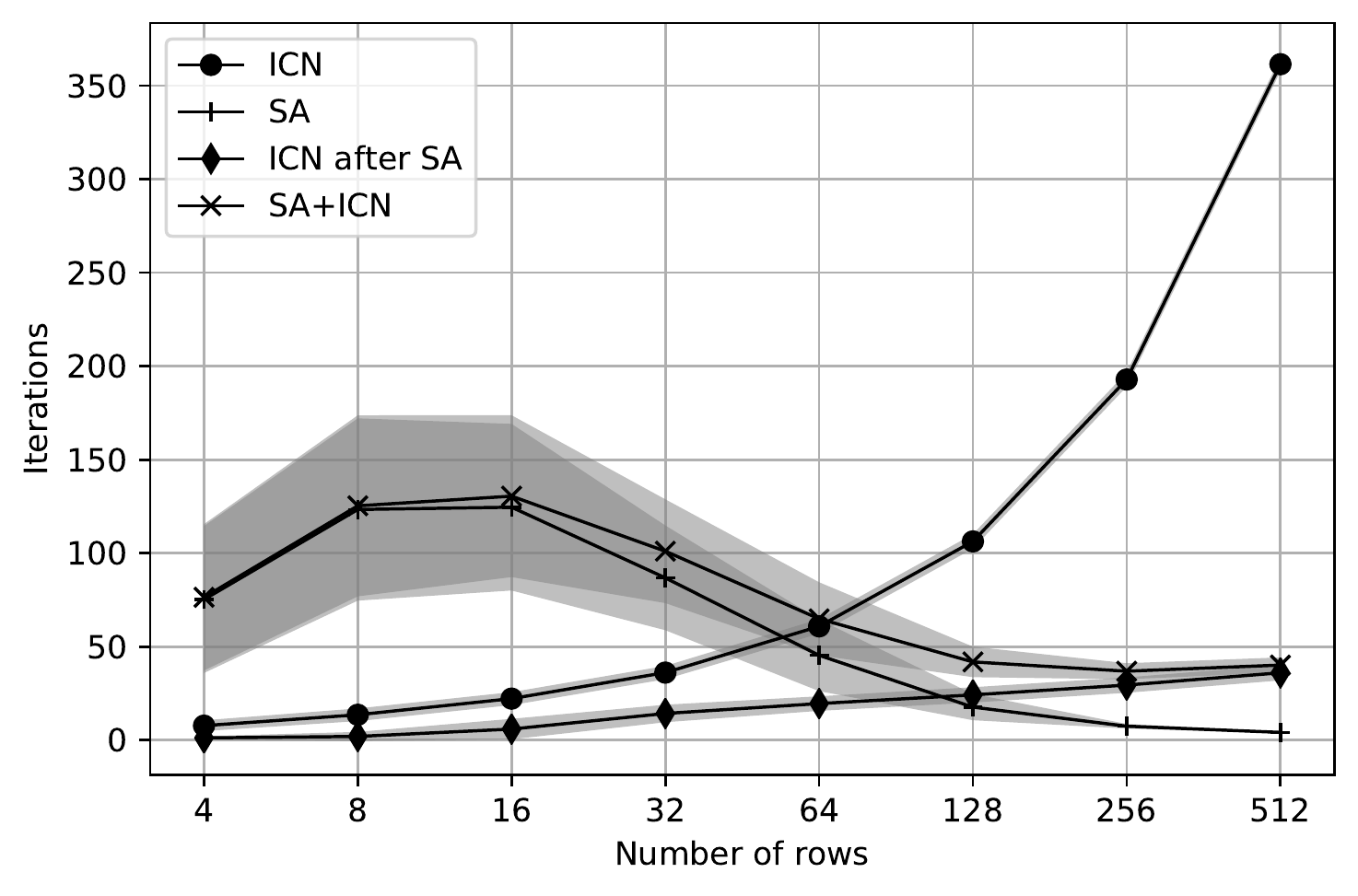}
	\end{tabular}
    \caption{Average number of iterations of $\ICN$, $\SA_{0.01}$, $\ICN$ after $\SA_{0.01}$ and 
    $\SA_{0.01}+\ICN$ for a sigmoid activation ($a=5$), 
    as a function of the problem size. 
    The grayed areas correspond to $1$ standard deviation around the average.} 
    \label{fig_SK_ICN_iterations}
\end{figure}

One sees that for uniformly sampled weights and our stopping criteria: 
\begin{itemize}

\item When running from the initial weight matrix, the number of iterations of $\ICN$ grows almost linearly vs. the problem size $n$ (the x-axis is logarithmic, doubling $n$ approximately doubles the number of iterations). It has a very small standard deviation, even when the average number of iterations is large 
($361$ iterations on average for $n=512$ with a standard deviation of $3.7$ iterations).

\item $\SA_{0.01}$ on its side first increases up to $n=16$ and then decreases to reach a few iterations ($4.1$ iterations on average for $n=512$). The standard deviation is large when the number of iterations is also large (e.g. $125$ iterations on average for $n=16$ with a standard deviation of $45$ iterations).

\item Starting from the nearly doubly stochastic matrix provided by $\SA$, $\ICN$ requires an order of magnitude less iterations to crisp the assignment than starting from scratch ($36$ vs $361$ iterations on average for $n=512$).

\item All-in-all, comparing the number of iterations required to the gaps obtained in figure \ref{fig_SK_ICN_tau}, 
for small sizes ($n\le 64$) $\ICN$ provides a less accurate solution than $\SA_{0.01}+\ICN$ but converges faster, 
and for larger sizes, $\SA_{0.01}+\ICN$ is faster and provides equivalent accuracy, hence should be preferred.
\end{itemize}

In figure \ref{fig_SK_ICN_iterations_sigmoid}, we have represented the average number of iterations of $\SA_{0.01}$ followed by $\ICN$ for increasing values of the parameter $a$ of a sigmoid activation.
One sees that for $a\le 3$, increasing $a$ increases the number of iterations hence reduces the 
convergence speed. The standard deviation is also large for these settings. 
But a transition occurs between $a=3$ and $a=4$, 
after which the number of iterations drops by an order of magnitude.
Convergence speed then increases with $a$ and the standard deviation has reduced importantly. 
Remark that for $a=4$, the average number of iterations has reduced below $50$ for $n=512$ but 
that the standard deviation is still large, so that we think that $a=4$ is close to the transition point.

\begin{figure}[!ht]\center\begin{tabular}{ccc}
	\includegraphics[width=0.3\linewidth]{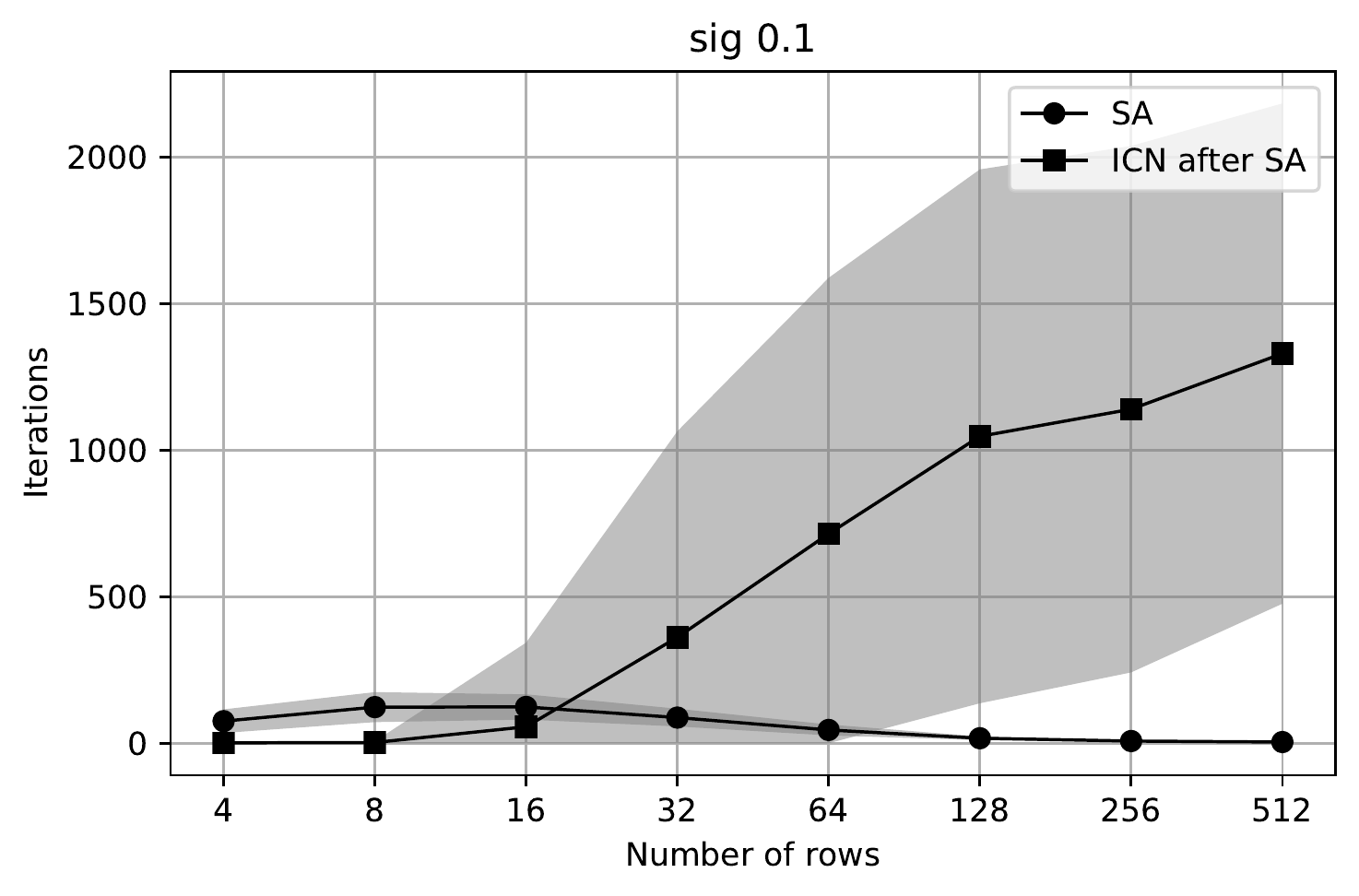}&
	\includegraphics[width=0.3\linewidth]{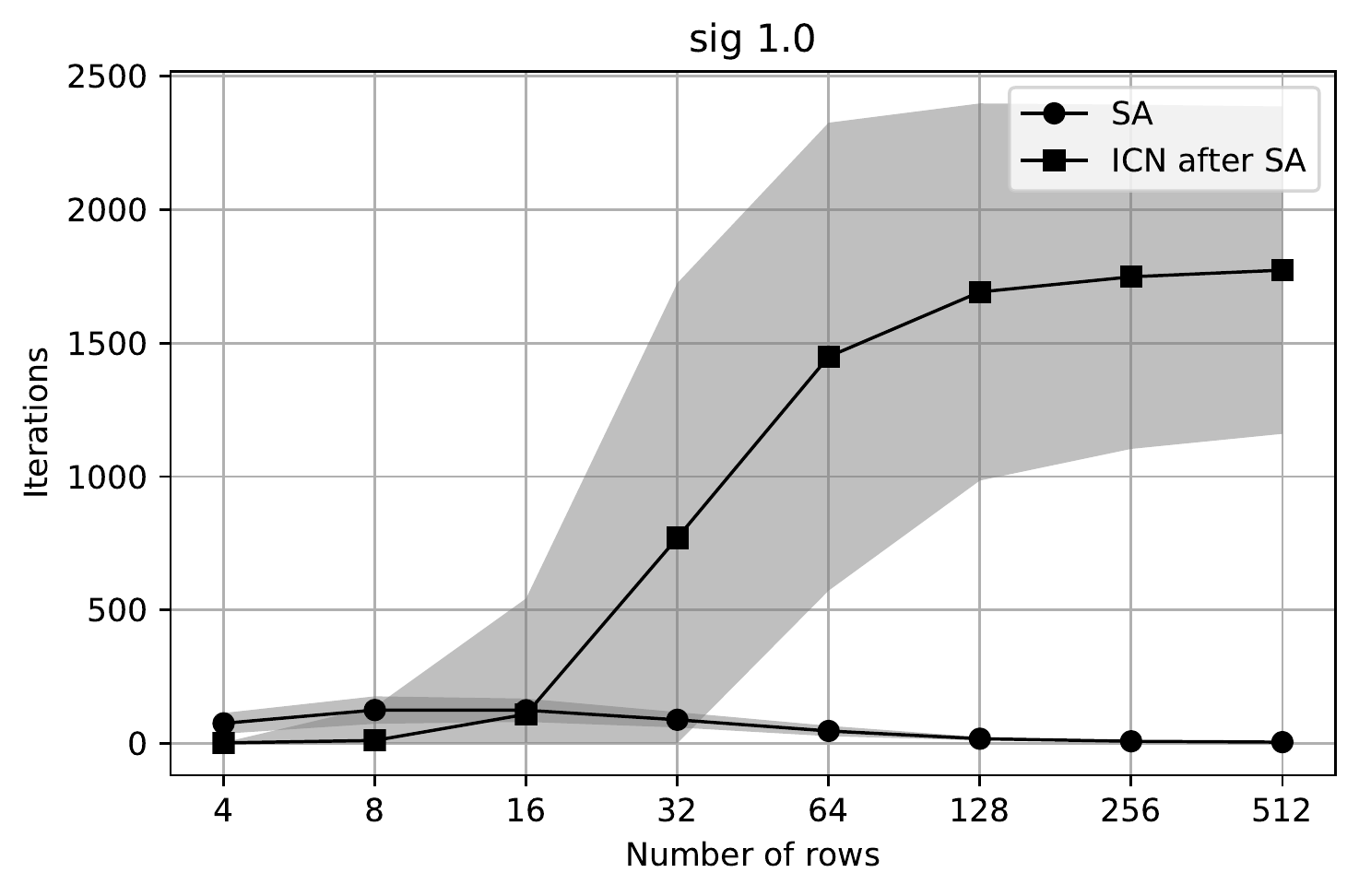}&
	\includegraphics[width=0.3\linewidth]{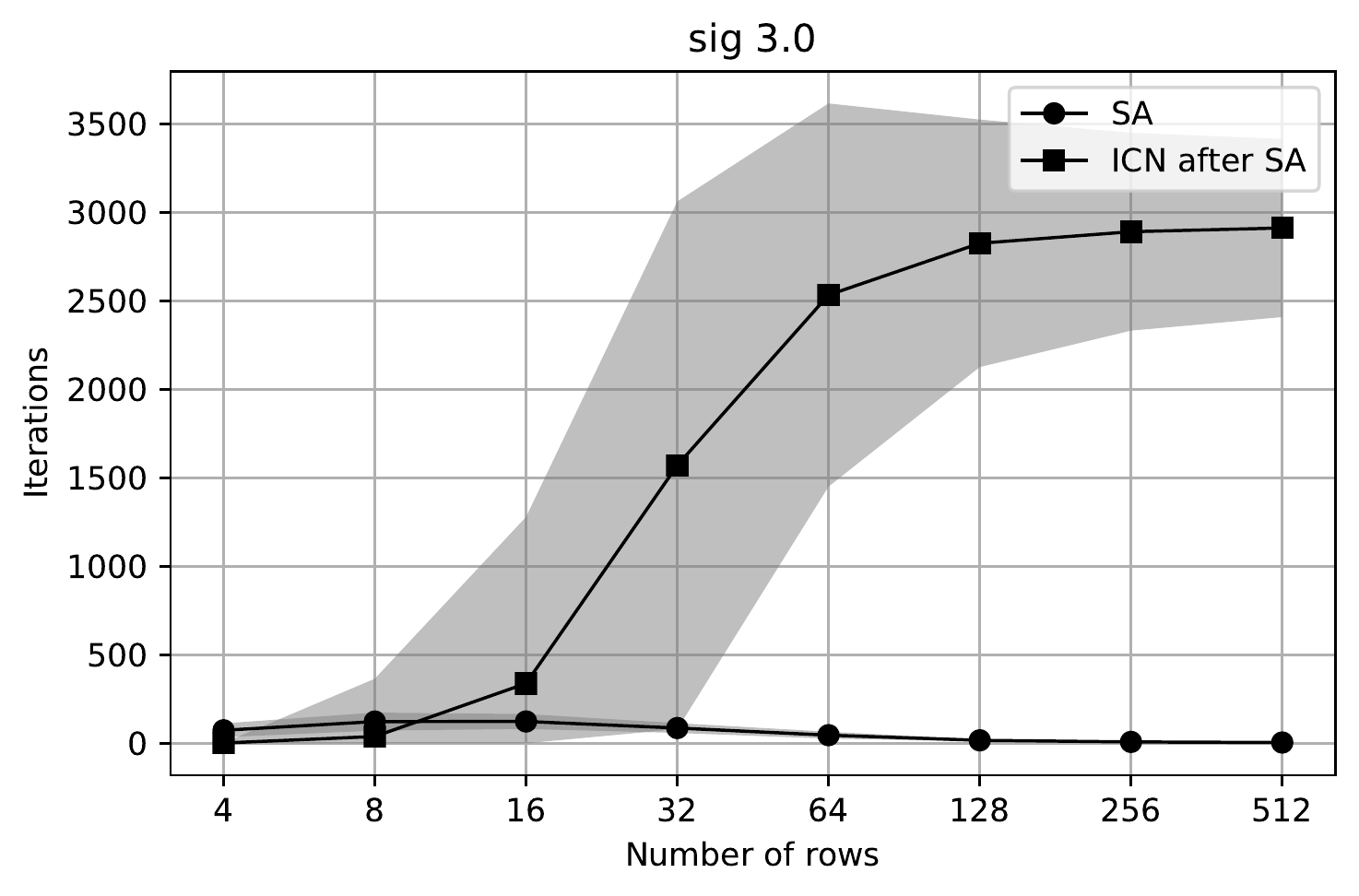}\\
	(a) $a=0.1$ & (b) $a=1.0$ & (c) $a=3.0$ \\
	\includegraphics[width=0.3\linewidth]{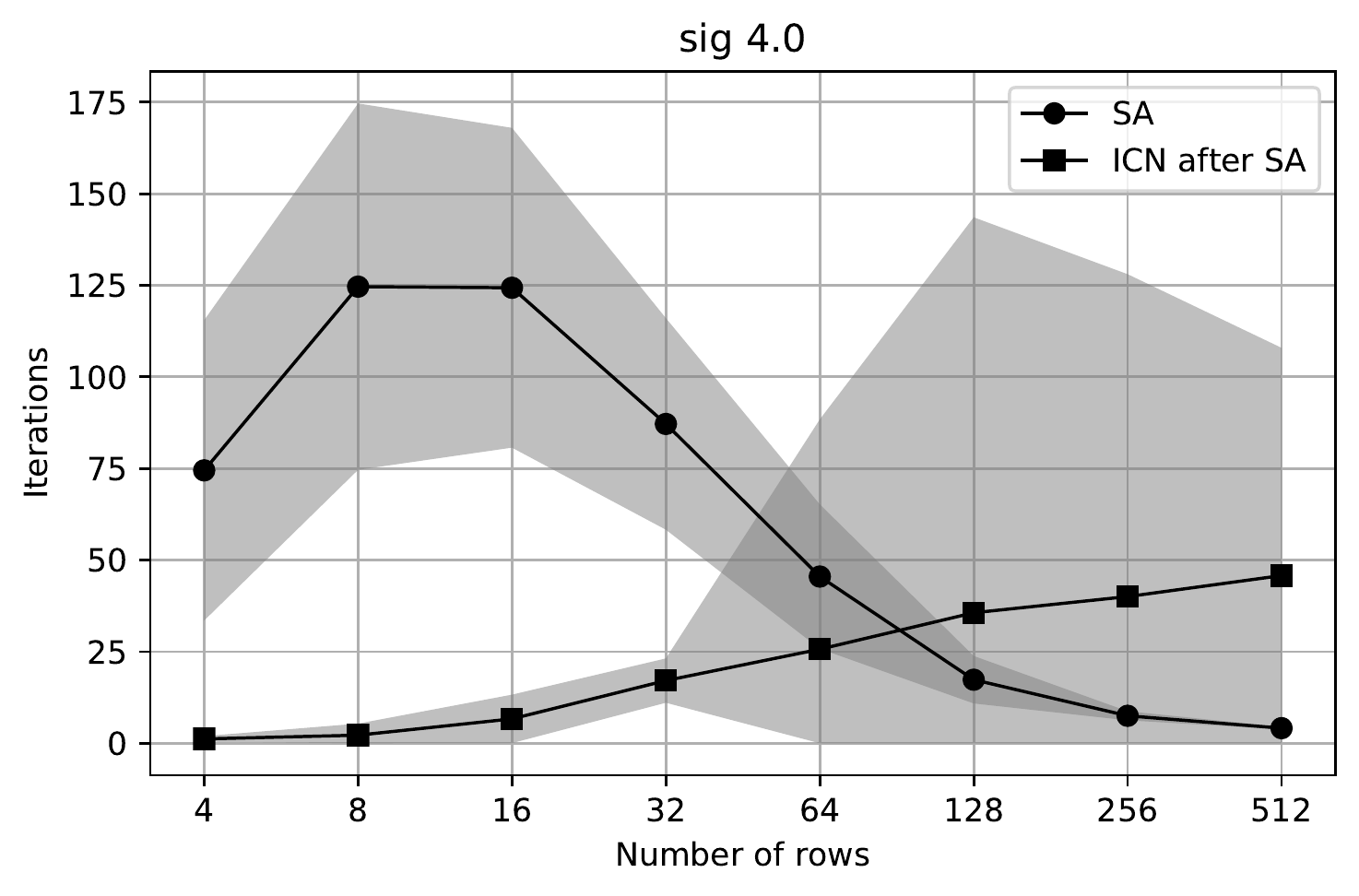}&
	\includegraphics[width=0.3\linewidth]{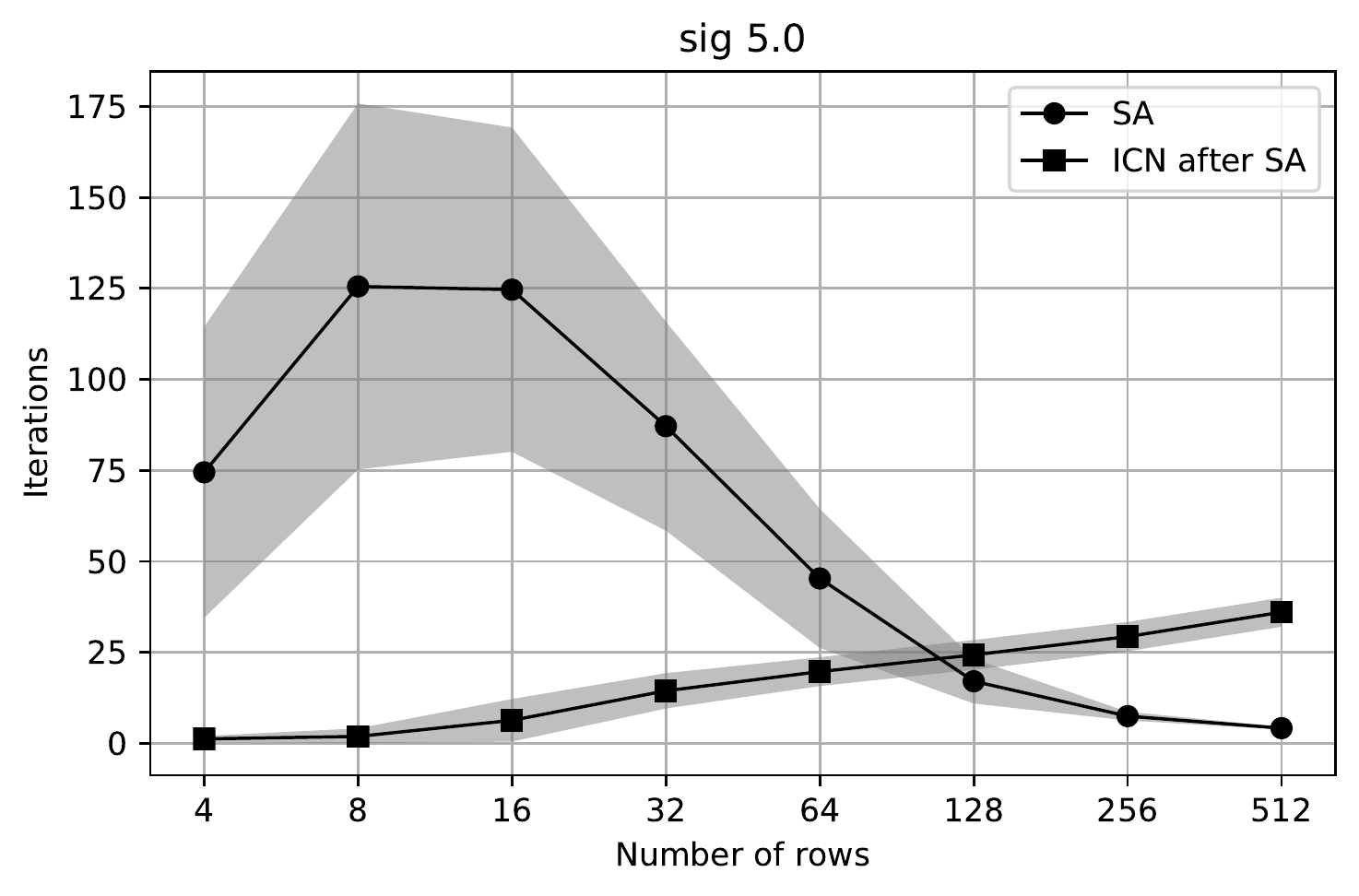}&
	\includegraphics[width=0.3\linewidth]{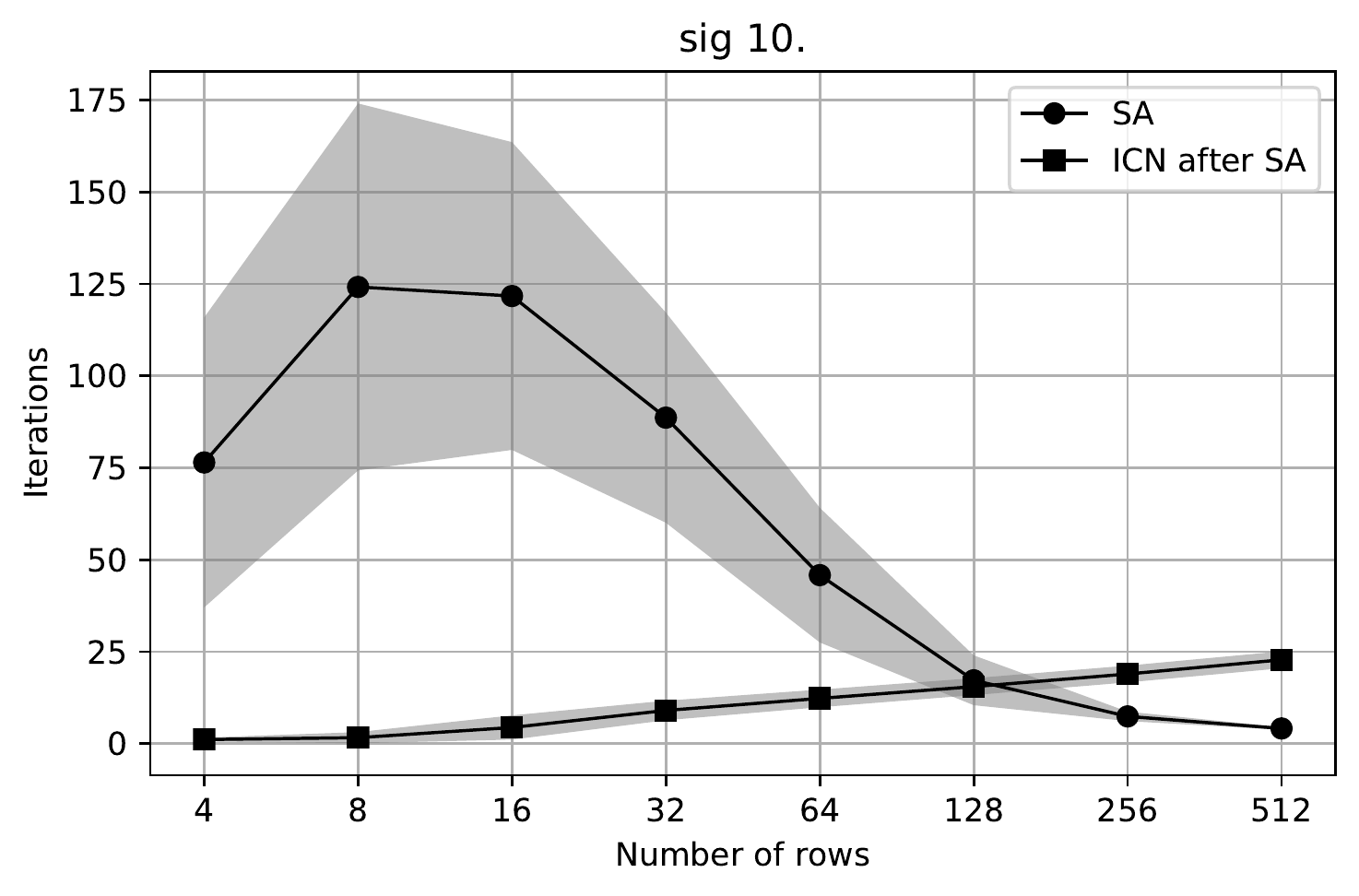}\\
	(d) $a=4.0$ & (e) $a=5.0$ & (f) $a=10.0$ 
	\end{tabular}
    \caption{Average number of iterations of $\SA_{0.01}$ followed by $\ICN$ for increasing non-linearity 
    of a sigmoid activation.}
    \label{fig_SK_ICN_iterations_sigmoid}
\end{figure}

In figure \ref{fig_SK_ICN_gap_vs_iterations}, we look at the optimality gap of $\SA+\ICN$ 
vs convergence speed for various activation functions.
The results show that increasing the exponent of a power activation seems 
to systematically degrade both the optimality gap and the convergence speed 
(compare $a=1.01$ to $a=1.1$ and $a=1.5$).
However, reducing the shift at zero, which amounts to reducing the derivative of the activation function at zero, 
speeds up convergence while only slightly increasing the gap (compare $a=1.01, t=10^{-1}$ to $a=1.01, t=10^{-4}$). 
The behaviour of the algorithm for sigmoid activations is different, as noticed previously by looking at the convergence speed w.r.t. non-linearity.
Increasing the non-linearity of the sigmoid from $a=0.1$ to $a=3$ both degrades the optimality gap and the 
convergence speed, however, starting from $a=4$, the convergence speed dramatically increases 
and the optimality gap regresses. Both performances for lvalues of $a$ larger than $4$ then seem to stabilize.

As a conclusion, there is a trade-off between optimality gap and convergence speed.
Increasing non-linearity increases convergence speed but degrades accuracy, which 
is intuitively related to the range at which information propagates in the graph through the 
speed at which the weights are binarized.
We observe a sort of phase transition phenomenon at a certain point, 
in the sense that past a certain amount of linearity, the dynamics get faster by an order of magnitude 
but seem to freeze at this level.

\begin{figure}[!ht]\center\begin{tabular}{c}
	\includegraphics[width=0.6\linewidth]{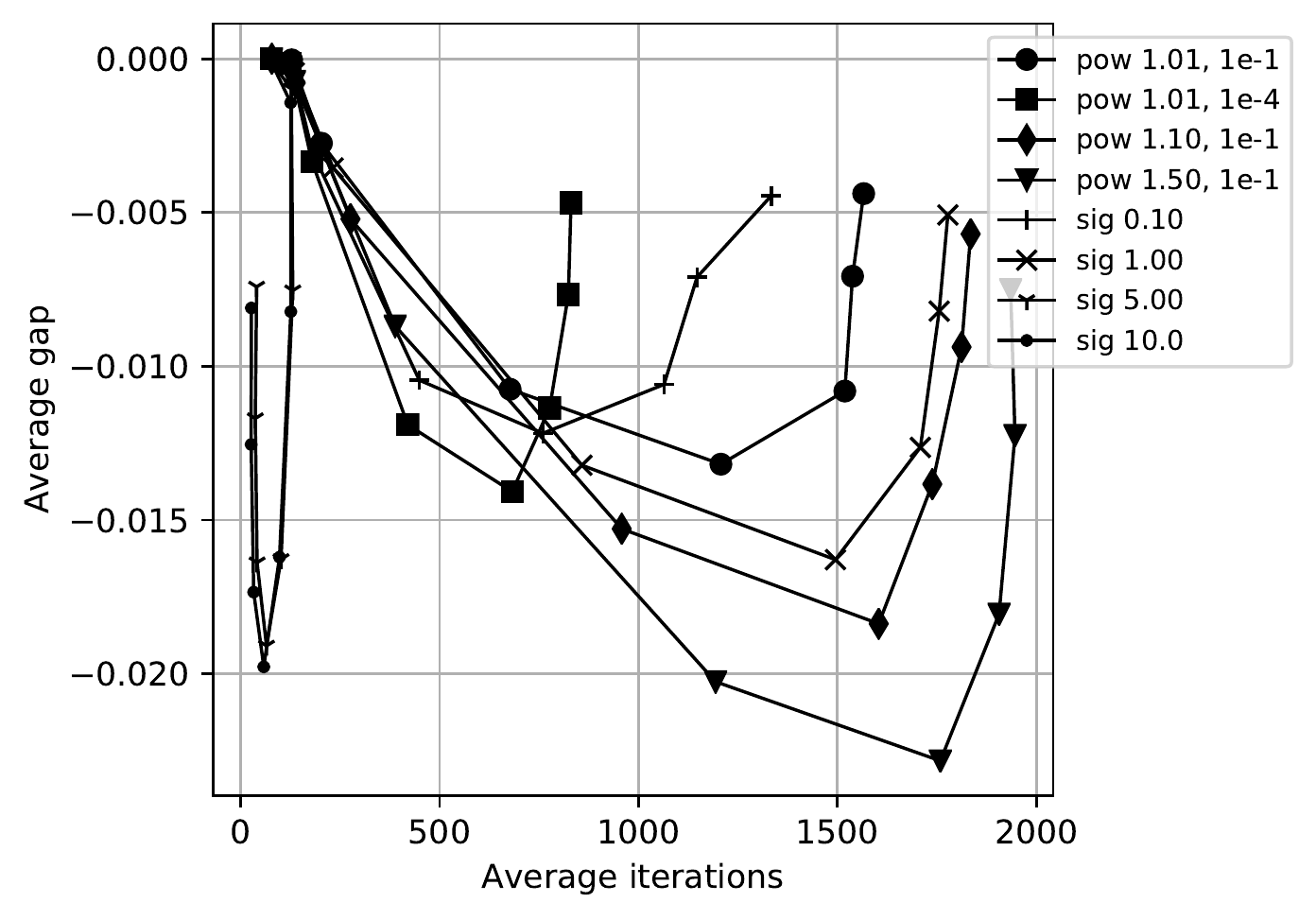}
	\end{tabular}
    \caption{Average optimality gap vs average number of iterations of $\SA_{0.01}+\ICN$  for various activations.
    The points on each curve correspond to different problem sizes, from $n=4$ to $n=512$.} 
    \label{fig_SK_ICN_gap_vs_iterations}
\end{figure}
\section{Additional remarks and discussion}\label{sec_discuss}
\subsection{Elements of general convergence}
Despite the results above on convergence around maximal independent sets, 
we weren't able to prove convergence for any initialization.
We always verified it in practice and we conjecture that iterative graph normalization converges in general.

We also remarked another interesting property: 
in all our numerical experiments with a linear activation, 
the sum of the weights of the graph was systematically strictly increasing (except on a fixed point of course).
As the weights $x$ are positive, this sum also corresponds to the $L_1$ norm of $x$.
We related the increase of the $L_1$ norm by graph normalization 
to a more fundamental linear algebra property about adjacency matrices of simple connected graphs 
which we also verified through extensive numerical experiments.
The following property is our main conjecture upon which the general convergence of graph normalization rests:

\begin{conjecture}\label{conj_QL1} If $A$ is the adjacency matrix of a connected simple graph then 
\begin{eqnarray}
 \forall x \in {\mathbb{R}^{+*}}^{n} && 
 y^t (A+I) y \le \norm{y}_1 \\
 \textrm{where} && y = x \HD (A+I) x  \nonumber
\end{eqnarray}
\end{conjecture}

If this conjecture holds then the next theorem shows that the sum of the weights of a graph systematically increases by normalization and converges:

\begin{restatable}{theorem}{thLoneincreases}\label{ref_thLoneincreases}
Let $G=(A,x)$ be a connected simple graph with positive weights $x>0$.
Let $x^0 = x$ and $x^{k+1} = \R(x^k)$.
If $x^0$ is not a fixed point and if the conjecture \ref{conj_QL1} holds then $\norm{x^k}_1$ is a strictly increasing sequence with respect to $k$ which converges as $k$ goes to infinity.
\end{restatable}

Of course, convergence of the $L_1$ norm of the weights doesn't imply convergence of the weights themselves 
\footnote{Convergence of the $L_1$ norm of $x^k$ isn't convergence \emph{in} $L_1$ norm, 
i.e. convergence of $\norm{x^k - x^*}_1$ for some point $x^*$, which would imply 
component-wise convergence as we are in a finite space.}.
However, as we have proved, the image of normalization is a hypersurface which
is in a 1-1 mapping with the unit simplex. 
Such a variety is very regular and increasing $L_1$ norm amounts to 
climbing it, eventually ending on one of its maxima in $L_1$ norm.
We thus further conjecture that convergence of the $L_1$ norm of the weight vector implies 
component-wise convergence.

\subsection{Independent sets and Rayleigh quotients}
The independent sets of a graph $A$ are intimately related with the Rayleigh Quotients of $A+I$.

A binary vector $x\ne 0$ is the indicator of an independent set $S$ if and only if it is orthogonal to $Ax$.
Indeed if $i\in S$ then $\forall j\in \N{i} : j\notin S$ hence $x_i = 1 \Rightarrow (Ax)_i = 0$ and
so $x^t Ax = 0$.
Conversely, if $x^t Ax = 0$ then whenever $x_i = 1$ then $(Ax)_i$ must be null.
This equation corresponds to the constraint used in quadratic 
programming formulations of \MWIS.

We thus have:
\begin{eqnarray*}
\forall x \in \{0,1\}^n \ne 0_n : x \in \IS(A) &\Leftrightarrow & x^t Ax = 0 \\
&\Leftrightarrow & x^t Ax  + x^t x= x^t x \\
&\Leftrightarrow & x^t (A+I)x = x^t x 
\end{eqnarray*}
And thus
\begin{eqnarray*}
\forall x\in \{0,1\}^n \ne 0_n : x \in \IS(A)  &\Leftrightarrow & \frac{x^t (A+I)x}{x^t x} = 1
\end{eqnarray*}
Which means that a binary vector $x$ is an independent set of $A$ if and only if the Rayleigh Quotient at $A+I$ on $x$ is equal to $1$.

Now remark that \emph{all} the terms of the sum $x^t Ax$ are null and not only the whole sum, 
hence the following vector equation also holds:
\begin{eqnarray}
\forall x\in \{0,1\}^n \ne 0_n : x \in \IS(A) &\Leftrightarrow &  x \HM Ax = 0_n \nonumber \\
&\Leftrightarrow & x \HM (A+I) x = x \HM x\label{eqn_ray_squared}
\end{eqnarray}
As $x$ is binary $\forall i: x_i^2 = x_i$, hence $x \HM x = x$,  and so
\[
x \HM (A+I) x = x
\]

which is a fixed point equation for the independent sets of $A$.

Graph normalization can thus be viewed as the corresponding fixed point algorithm:
\[ 
x_{n+1} = x_n \HD (A+I)x_n
\]

Note that this derivation suggests that any non-linear function $h$ which is stable on $0$ and $1$, 
either applied on the numerator of the Hadamard division only, 
or on $x_n$ in both the numerator and the denominator - which amounts to alternating 
normalization and non-linear activation - also gives a fixed point equation.

Further note that these fixed point equations also have non binary fixed points in general, and they have much in practice as we have seen before, but which - as we conjecture - can be avoided by introducing a non-linearity within the iterations.

Interestingly, the fixed point algorithm based on the squared $x_i$ at the numerator, 
i.e. corresponding to the equation \ref{eqn_ray_squared}, diverges empirically.
We explain it by remarking that the equation $x \HM A x = 0_n$  also holds for any multiple 
of the indicator vector of an independent set.
Using $x_i^2 = x$ then restricts it to binary vectors. 

\subsection{Matrix cross normalization in the Fourier domain}
An interesting property of cross normalization of a matrix 
comes from the fact that the matrix of the cross sums, 
i.e. $A\ONE\ONE^t + \ONE\ONE^tA - A$ can be expressed as 
a 2D circular convolution (denoted by $*$):
\[
A\ONE\ONE^t + \ONE\ONE^tA - A \quad= \quad \Gamma * A
\]
with 
\[
\Gamma = 
\begin{pmatrix}
1 & 1 & \cdots & 1 \\
1 & 0 & \cdots & 0 \\
\vdots & \vdots & & \vdots \\
1 & 0 & \cdots & 0 
\end{pmatrix}
\]
The convolution theorem applied to a fixed point $X$ of cross normalization then gives:
\begin{eqnarray*}
X \HM (\Gamma * X) = X &\Leftrightarrow&
F\left[ X \HM (\Gamma * X)\right] = F(X) \\
&\Leftrightarrow& 
F(X) * F(\Gamma * X) = F(X) \\ 
&\Leftrightarrow& 
F(X) * \left[ F(\Gamma)\HM F(X) \right] = F(X)
\end{eqnarray*}
where $F$ represents the Discrete Fourier Transform. 
The roles of the pointwise multiplication (Hadamard product) and the convolution are thus swapped
in the spectral domain.

Cross normalization could thus be related to a form of deconvolution in the spectral domain.
\subsection{Various extensions}
\subsubsection{Extension to other fields}
The normalization equation is valid for weights belonging to any field, e.g. to the field of complex numbers.
Graph normalization also empirically converges on graphs valued by complex weights.
Furthermore, applying a non-linear activation on both the real and imaginary part of the weights after 
each normalization also empirically yields convergence to binary solutions, 
i.e. projects the complex weights on the two points $0$ and $1$ of the real axis, 
like in the real-valued weights case.

\subsubsection{Extension to continuous time dynamics}
Let $A$ be the adjacency matrix of a graph and $B=A+I$.
The continuous time ordinary differential equation corresponding to $\IGN$ is
\[
\frac{dx}{dt} = h ( x \HD Bx ) - x 
\]

Numerical integrations of this ODE for $P_3$ are represented in figure \ref{fig_P_3_continuous}.

\begin{figure}[!ht]\center\begin{tabular}{cc}
	\includegraphics[width=0.4\linewidth]{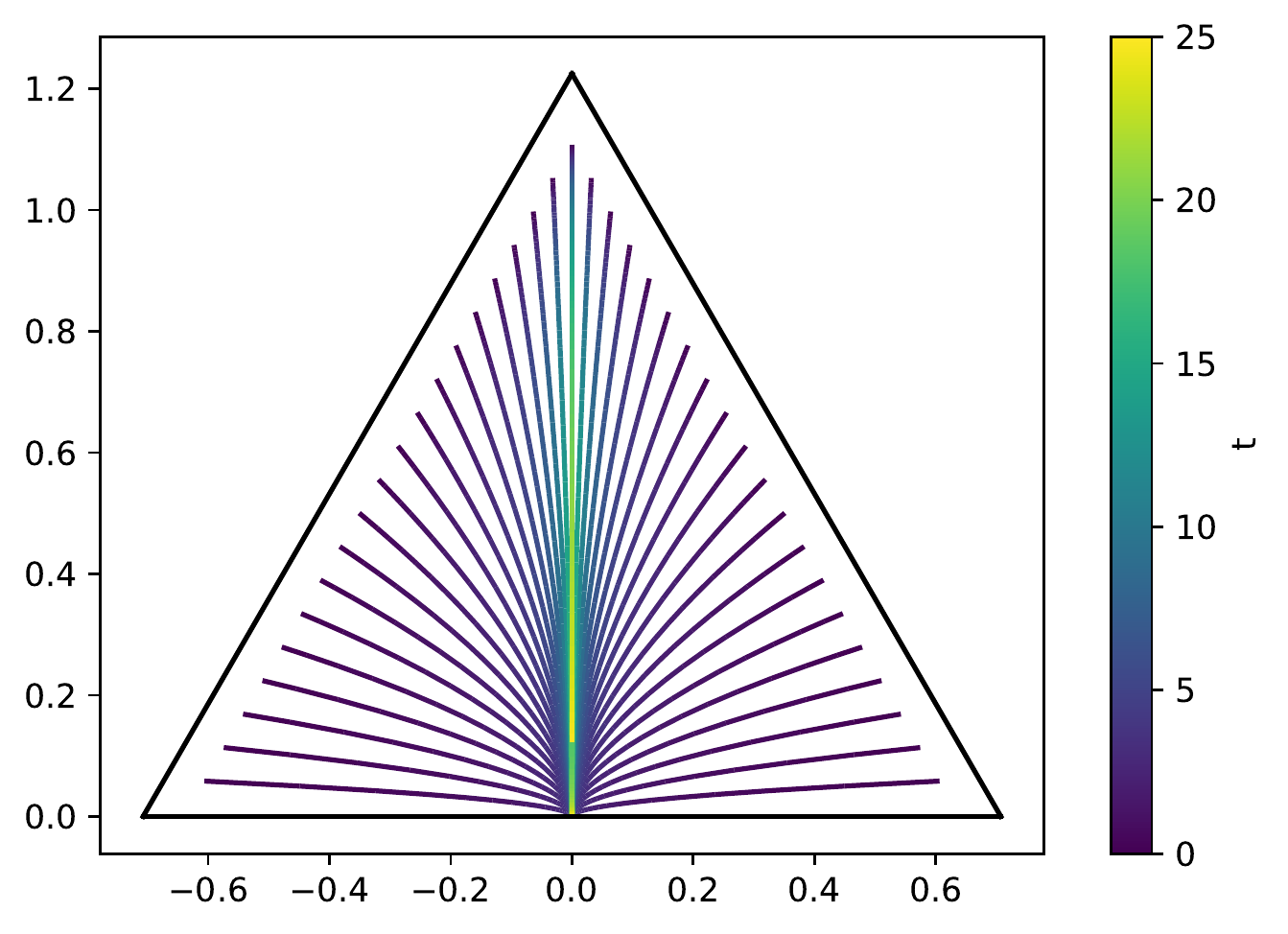} &
	\includegraphics[width=0.4\linewidth]{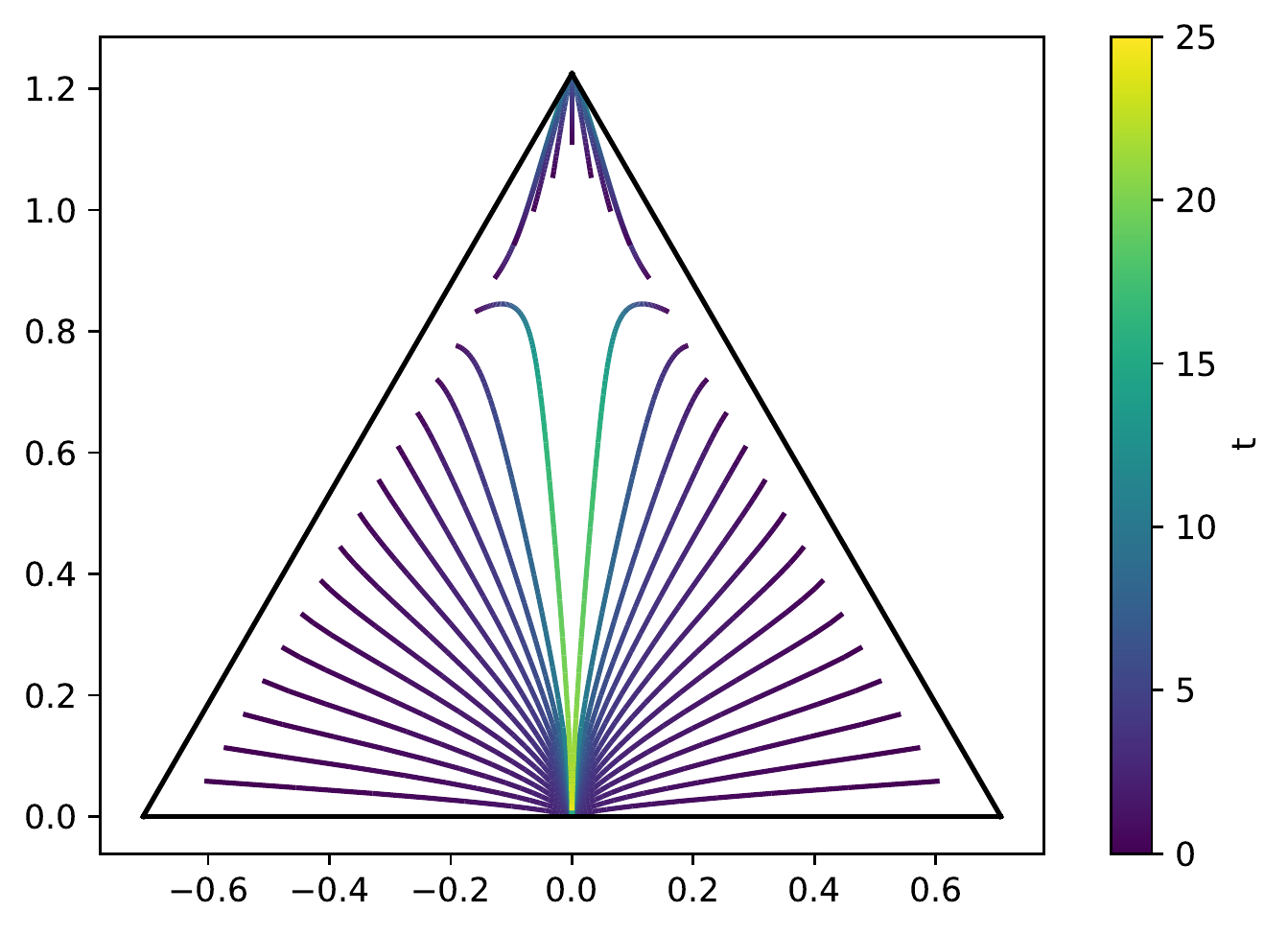} \\
	(a) & (b)	
	\end{tabular}
    \caption{Flows corresponding to the continuous normalization of $P_3$ projected on the simplex.
    (a) Linear activation. 
    (b) Power activation with $a=1.1$.} 
    \label{fig_P_3_continuous}
\end{figure}

\subsubsection{Extension to edge-weighted graphs and graph learning}
A natural question which arises is: what happens if the matrix $A$ - the adjacency matrix - is not binary?
$A$ would then correspond to an edge-weighted graph, hence to adding some transition weights between nodes.
Empirically, graph renormalization with a non binary transition matrix also seems to always converge with a non-linear activation function.
The solution are still characterized by their sparsity, in that they contain multiple zeros, but 
the non zero weights are not $1$ anymore.
As $\IGN$ is also differentiable w.r.t $A$, the edge weights can be learned in order to 
minimize a final loss function.
$\IGN$ with learnable edge weights thus corresponds to a new class of neural networks on graphs.

\subsubsection{Extension to more general topologies}
We studied $\IGN$ for discrete topologies defined by finite graphs.
Another subject of interest would be the study of iterative normalization on other topological spaces.
One could for example consider the functions from $\mathbb{R}^n$ to a field $F$ and a system of weighted neighborhoods on the points of space. 
For each point $x \in \mathbb{R}^n$, define its neighborhood as a function $N(x) : \mathbb{R}^n \rightarrow \mathbb{R}$.
 The normalization of a function $f : \mathbb{R}^n \rightarrow F$ would then be another function $\R(f)$ 
 whose values would be defined by:
\[
\R(f)(x) = \frac{f(x)} { \int_{\mathbb{R}^n} N(y) f(y) dy }
\]

Under proper conditions on the functions and the neighborhoods.

\section{Conclusion}
If our conjectures are correct, iterative graph normalization would be 
a systematically convergent differentiable approximation 
of the maximum weight independent set problem in a graph, 
thereby unlocking various end-to-end machine learning applications.

We have proved convergence for complete graphs and local convergence around 
maximal independent sets in general graphs but couldn't prove convergence in general.
We are certain that better mathematicians than we are will be able to prove or disprove 
our conjectures.
There are also a number of intriguing relationships between $\IGN$ and other objects, 
such as the graph Laplacian or Rayleigh Quotients that are worth exploring.
The extension to other numerical fields, topological spaces, continuous time or to edge-weighted graphs and to graph learning are also totally blank pages.

In practice, we have shown through numerical experiments that $\IGN$ provides close solutions 
to the greedy algorithm $\KAKO$ by Kako et al. \cite{kako2005approximation}.
We have also studied the relationship between $\IGN$ and the Softassign algorithm for the 
assignment problem, a special case of $\MWIS$, and showed 
that chaining both algorithms provided an approximation 
algorithm with a small optimality gap. 
As mentioned in introduction, 
$\MWIS$ encompasses a broad class of combinatorial optimization sub-problems, 
such as max flow or shortest path problems, for which optimal algorithms are known.
It would be interesting to study the empirical accuracy of $\IGN$ on such problems.

Of course, practical application in end-to-end deep learning systems, 
for example for computer vision, are to be explored.

We believe that $\IGN$ opens a whole new field of theoretical investigations 
and practical applications.

\section*{Acknowledgments}
We would like to thank Jean-Luc Guigues for fruitful discussions on matrix cross normalization
at early stages of our journey into $\IGN$, 
his complete enumerations of fixed points for small dimensions 
and a number of results on cross normalization which we didn't include here.
We are grateful to Paul Munger for the finding of the "Taco" and for many insightful discussions.
We also thank Ali Rahimi for his feedback, our rich exchanges and 
the idea of the "vec trick" to transform the matrix formulation of cross normalization into 
the vector formulation corresponding to the associated general $\MWIS$ problem. 
\appendix
\section{Proofs}
\thKnconvergence*
\begin{proof}
Without loss of generality, 
we can assume that the maximum component of $x$ is $x_1$, i.e. that $\forall i>1 : x_i < x_1$. 

For the complete graph, normalization amounts to rescaling all the weights by the same value, hence 
it preserves the order of the weights. As $h$ is increasing it also preserves order, hence $x_1$ remains the 
maximum component throughout the iterations. 

Let us split the activated normalization iteration into two alternating steps:
\begin{eqnarray*}
x^0 & = & x \\
y^{k+1}  & = & \R(x^k) \\
x^{k+1} & = & h(y^{k+1})
\end{eqnarray*}

We have $\forall i>1 : y_i^{k} = x_i^{k} / (x_1^k+x_i^k+\sum_{j\ne 1, j\ne i} x_j^k)$ 
hence as $x_1 > x_i$, $y_i^k < x_i^k / (x_1^k + x_i^k) < x_i^k / (2x_i^k) = 1/2$.
Now $h$ is strictly convex on $(0, 1/2[$ and so if $y_1^k$ is also smaller than $1/2$ then 
the value of $h$ at $y_i^k$ is 
strictly lower than the value at $y_i^k$ on the secant from $(0, 0)$ to $( y_1^k, h(y_1^k))$, i.e.:
\begin{equation}
\label{prop_Kn_eqn_1}
h(y_i^k) < \frac{h(y_1^k)}{y_1^k} y_i^k
\end{equation}
This also holds if $y_i^k > 1/2$ as we required that if the activation function crosses 
the line $y=x$ then it remains over that line, hence if it is the case, the slope $\frac{h(y_1^k)}{y_1^k}$ 
is greater than $1$ and the inequality still holds. 
From equation \ref{prop_Kn_eqn_1} we get:
\begin{equation*}
\frac{h(y_i^k)}{h(y_1^k)} < \frac{y_i^k }{y_1^k}
\quad\Leftrightarrow \quad
\frac{x_i^k}{x_1^k} < \frac{y_i^k }{y_1^k}
\end{equation*}
As noted above, as $x_1^k>0$, normalization of $x_1^k$ can be expressed in terms of the ratios of the other weights to $x_1^k$: 
\begin{equation*}
y_1^{k+1} = \left( 1 + \sum_{i>1} \frac{x_i^k}{x_1^k} \right)^{-1}
\end{equation*}
And thus
\begin{equation}
\label{prop_Kn_eqn_2}
y_1^{k+1} = \left( 1 + \sum_{i>1} \frac{x_i^k}{x_1^k} \right)^{-1}
> \left( 1 + \sum_{i>1} \frac{y_i^k}{y_1^k} \right)^{-1}
= \quad \R_1(y^k)
\end{equation}
Now, $y^k = \R(x^{k-1})$, hence it is already normalized and as graph normalization 
(with a linear activation function) is stable on the complete graph, we obtain $\R(y_k)=y^k$ 
and in particular $\R_1(y^k)=y_1^k$.
We thus get that $y_1^{k+1}  > y_1^k$ and thus 
that $y_1$ is strictly increasing. As it is bounded above by $1$ it converges.
The only way to achieve equality in equation \ref{prop_Kn_eqn_2}, hence convergence,  
is to have $\forall i>1 : h(y_i^k) = y_i^k$ in equation \ref{prop_Kn_eqn_1} which,
as $y_i^k < 1/2$, only happens when $\forall i>1 : h(y_i^k) = y_i^k = 0$.  
As $\sum y_i^k = 1$, the only possible fixed point verifies $y_1=1$ and $\forall i>1 : y_i = 0$ and 
thus $x_1=1$ and $\forall i>1 : x_i = 0$.
\end{proof}
\thprojective*
\begin{proof}
Let $G=(A,x)$ be a normalizable graph.
Let $B=A+I$ and $y=x \HD Bx$ for some $x \in \DOM$.
Assume that another point than $x$ maps on the same line than $y$, 
i.e. assume $\exists z \in \DOM$, $\exists k>0$ such that $ky = z \HD Bz$.
We get $kx\HD Bx = z\HD Bz$ hence $k (x\HM Bz) = z \HM Bx$.
Summing up the components of these vectors, we obtain:
\begin{eqnarray*}
k \sum_i (x\HM Bz)_i &=& \sum_i (z\HM Bx)_i  \\
k x^t B z &=& z^t B x
\end{eqnarray*}
$B$ is symmetric hence $x^t B z = z^t B x$, and thus $k=1$.
\end{proof}
\thtreeinjective*
\begin{proof}
Let 
Consider $x>0$ and $y>0$ such that $x\HD Bx = y\HD By$.
This gives 
\begin{eqnarray}
x\HM By \quad - \quad y \HM Bx &=& 0 \nonumber\\
\Leftrightarrow \quad 
x\HM(A+I)y \quad - \quad y \HM(A+I)x &=& 0 \nonumber\\
\Leftrightarrow \quad 
x\HM Ay \quad - \quad y \HM Ax &=& 0 \label{eqn1_th_homo_tree}
\end{eqnarray}
Without loss of generality, we can assume that $T$ is connected, 
hence that its adjacency matrix $A$ has at least one non zero entry in each line. 
In this case, the last equation is a system of $n$ homogeneous polynomials equations 
of degree $2$ of the $2n$ variables 
$x_1, \dots x_n, y_1 \dots y_n$.

Let $d_{ij} = x_i y_j  -  x_j y_i$.

As $A$ is symmetric, the $i$th equation can be written:
\begin{equation}
\sum_j A_{ij} d_{ij} = 0
\label{eqn_rev}
\end{equation}
 
If a node $i$ is a leave, i.e. has degree $1$, then its corresponding equation reduces to 
$d_{ij} = 0$,  where $j$ is $i$'s unique neighbor.
As the weights are positive, it means that $(x_i,x_j)$ is proportional to $(y_i,y_j)$, hence $\exists k_{ij}>0 | (x_i,x_j) = k_{ij} (y_i,y_j)$.

If $T$ is a tree, one can order the equations in climbing order from the leaves.
Let $V=\{1\dots n\}$ denote the set of nodes of $T$.
Let $L_1 \subset V$ be the set of the leaves of $T$. 
We then define $L_2$ as the leaves of the graph obtained after removing $L_1$ from $T$, 
$L_3$ as the leaves of the graph obtained after removing $L_1$ and $L_2$ from $T$, etc.
Formally, $L_{k+1}$ is defined recursively as the set of the leaves 
of the subgraph induced by $V \backslash \cup_{i=1\dots k} L_i$.
As $T$ is a tree, the $L_i$ form a partition of $V$.
We say that $L_k$ is the $k$-th layer of this partition.
Also because $T$ is a tree, 
any node $i$ in a layer $L_k$, except the node belonging to the 
last non-empty layer, has a unique parent $p(i)$ in a layer $k' > k$.

One can now traverse the equations layer by layer and prove that $y = kx$ by recursion.
As $L_1$ contains the leaves of $T$, we know that $\forall i \in L_1 :  d_{i p(i)} = 0$.
Assume that $\forall m \in \{1\dots k\}, \forall i \in L_m : d_{i p(i)} = 0$.
$\forall i \in L_{k+1}$, the $i$th equation of the system \ref{eqn_rev} 
contains a term $d_{i p(i)}$ corresponding to the edge between $i$ and its parent, 
and all the other terms corresponding to edges with adjacent nodes belonging to a previous layer $k' < k+1$.
As $d_{ij} = -d{ji}$, all these terms are null, hence $d_{i p(i)} = 0$.

We thus get $n-1$ coefficients of proportionality $k_{i p(i)}$ between $(x_i, x_{p(i)})$ and $(y_i, y_{p(i)})$, 
one for each edge of the tree, which are all equal by connectivity of the tree.
\end{proof}
\thbinaryfparemis*
\begin{proof}
Let $x\in \MIS(G)$ and $S=\SUPP(x)$.
$\forall i\in S: x_i=1$ and $\forall j\in \N{i} : x_j=0$ hence $\sum A_{ij}x_j = 0$ 
and thus ${\Rh}_i(x) = h(1) = 1$.
$\forall i\notin S: x_i=0$ and as $x$ is maximal $\exists j\in \N{i} : x_j=1$ hence $\sum A_{ij}x_j >= 1$ 
and thus ${\Rh}_i(x) = h(0) = 0$. $x$ is thus a fixed point of $\Rh$.

Assume now that $x$ is a normalizable binary vector which is not an independent set of $G$.
As $x$ is not independent, $\exists i$ such that $x_i = 1$ and $x_j = 1$ for some $j \in \N{i}$. 
Hence $\sum A_{ij}x_j >= x_j >= 1$ and as $h$ is strictly increasing, 
${\Rh}_i(x) = h(1 / 1 + \sum A_{ij}x_j) <= h(1/2) < 1$. Thus $x$ is not a fixed point of $\Rh$. 

\end{proof}

\propfixedregular*
\begin{proof}
Assume that $G=(A,x)$ is a non regular fixed cluster with all-identical weights equal to some value $a>0$. 
As $G$ is connected, it must contain two adjacent nodes $i$ and $j$ which have different degrees.
The normalization conditions on $i$ and $j$ are respectively:
\begin{eqnarray*}
x_i + x_j + \sum_{k\in \N{i}, k\ne j} x_k &=& 1\\ 
x_j + x_i + \sum_{k\in \N{j}, k\ne i} x_k &=& 1
\end{eqnarray*}
Substracting these two equation we get:
\[
\sum_{k\in \N{i}, k\ne j} x_k = \sum_{k\in \N{j}, k\ne i} x_k
\]
As all weights are equal to $a>0$, it gives: $a (\DEG(i)-1) = a (\DEG(j)-1)$
which contradicts the fact that $i$ and $j$ have different degrees.
\end{proof}

\thmisspectralradius*
\begin{proof}
Let $S \in \MIS(G)$ and $x = ind(S)$.

If $i\in S$ then $x_i=1$ and $\forall j\in \N{i} : x_j = 0$ hence $\sum A_{ij}x_j  = 0$ and $\sum B_{ij}x_j  = 1$.
We thus get:

$\forall i\in S:$
\begin{eqnarray*}
J_{ii}(x) &=& 0 \\
\forall j\ne i \quad J_{ij}(x)&=&-h\prime(1) A_{ij}
\end{eqnarray*}

If $i\notin S$ then $x_i=0$ and thus $\sum A_{ij}x_j = \sum B_{ij}x_j$. We thus get:

$\forall i\notin S:$
\begin{eqnarray*}
J_{ii}(x) &=& \frac{h\prime(0)}{\sum A_{ij}x_j} \\
\forall j\ne i \quad J_{ij}(x)&=&0
\end{eqnarray*}

Let $v$ be an eigenvector of $J$ associated with the eigenvalue $\lambda$.
Writing $Jv = \lambda v$, we obtain:
\begin{eqnarray}
\forall i \in S &:& -h\prime(1) \sum A_{ij} v_j = \lambda v_i  \label{eqn_evines}\\
\forall i \notin S &:& \frac{h\prime(0)}{\sum A_{ij}x_j} v_i = \lambda v_i \label{eqn_evnines}
\end{eqnarray}

If the eigenvector $v$ verifies $\forall i\notin S: v_i = 0$ then the equations \ref{eqn_evnines} are 
all true for any $\lambda$. 
Furthermore, in this case $\forall i\in S : \sum A_{ij}v_j = 0$ because $v_j = 0$ whenever $A_{ij} \ne 0$ in the sum.
The equations \ref{eqn_evines} thus give $\forall i\in S : \lambda v_i = 0$. 
As $v\ne 0$, at least one of the $v_i$ is nonzero, hence $\lambda = 0$.
Hence any vector $v$ such that $\forall i\notin S v_i = 0$ and $\exists j\in S | v_j\ne 0$ is an eigenvector 
of $J$ associated with the eigenvalue $0$.
A basis for this eigenspace is thus the subset of the natural basis vectors 
$\{ ind(\{i\}); i \in S\}$ of size $|S|$.

On the other hand, if $v$ verifies that $\exists i\notin S | v_i \ne 0$ then the equations \ref{eqn_evnines} give:
$\forall i\notin S | v_i\ne 0 : \lambda = \frac{h\prime(0)}{\sum A_{ij}x_j}$. The subsets of 
indices $K\subset \bar{S}$ such that $\forall k \in K | v_k \ne 0$, and such that all the eigenvalue found from all $k$ equations are consistent are the only possible generators of an eigenvector.
The elements in such a subset all have an equal adjacency to elements of $S$. 
One easily verifies that such a solution is compatible with the equations \ref{eqn_evines}.
We thus have as many distinct non zero eigenvalues as values in the set $\left\{ \frac{h\prime(0)}{\sum A_{ij}x_j}; i \notin S \right\}$

Let's denote these eigenvalues by $\lambda_i$, including the $0$ eigenvalue.
The spectral radius is then given by
\begin{eqnarray*}
\rho(J(x)) 
&=& \max_{i\notin S} \{ |\lambda_i| \} \\
&=& \max_{i\notin S}  \frac{h\prime(0)}{\sum A_{ij}x_j} \\
&=& \frac{h\prime(0)}{\min_{i\notin S}  \sum A_{ij}x_j} \\
&=& \frac{h\prime(0)}{dens_{G}(x)} \\
\end{eqnarray*}
as we assume that $h$ is increasing hence that $h\prime \ge 0$.
\end{proof}

\thnmisrepulsive*
\begin{proof}
 We prove that if $\forall y\in[0,1] : h\prime(y)>0$ then 
the Jacobian of $x+\epsilon$ diverges as $\epsilon \rightarrow 0$.
As $x$ is a non maximal independent set of $G$, $\exists i$ such that $x_i = 0$ and 
$\forall j\in \N{i} : x_j = 0$.
The diagonal entry $J_{ii}(x+\epsilon)$ of the Jacobian then becomes:

\begin{eqnarray*}
J_{ii}(x+\epsilon) &=& h\prime\left( \R_i(x+\epsilon)\right) 
\frac{ \sum_{j\ne i} B_{ij}(x_j + \epsilon_j) }{\left(\sum B_{ij}(x_j + \epsilon_j)\right)^2}\\
&=& 
h\prime\left( \R_i(x+\epsilon) \right) 
\frac{ \sum_{j\ne i} B_{ij}\epsilon_j }{\left(\sum B_{ij}\epsilon_j\right)^2} \\
&\ge& 
h\prime\left( \R_i(x+\epsilon) \right) 
\frac{ \sum_{j\ne i} B_{ij}\epsilon_j }{\left(\sum_{j\ne i} B_{ij}\epsilon_j\right)^2} \\
&=& 
h\prime\left( \R_i(x+\epsilon) \right) 
\frac{ 1 }{\sum A_{ij}\epsilon_j}
\end{eqnarray*}
As $\forall y : \R_i(y) \in [0,1]$:
\[
J_{ii}(x+\epsilon)  \ge 
 \frac{\min_{y\in [0,1]} h\prime(y) }{\sum A_{ij}\epsilon_j}
\]
which diverges as $\epsilon \rightarrow 0$ if $\forall y\in[0,1] : h\prime(y)\ne 0$.
Hence $Tr(J(x+\epsilon))$ diverges as $\epsilon \rightarrow 0$ 
and so is the spectral radius which means that $x$ is a repulsive point.
\end{proof}

\thmisbasin*

\begin{proof}
\begin{eqnarray*}
\forall i \in S: x_i^k > 1/2 &\Rightarrow & \forall j\notin S : r_j^k > 1 \Rightarrow l_j^k= 0\\
\forall i \notin S : x_i^k < 1/(2 d_S) &\Rightarrow & \forall j\in S : r_j^k < 1/2 \Rightarrow l_j^k > 1/2
\end{eqnarray*} 
which shows that $\forall i\notin S$, $x_i^{k+1}$ is decreasing towards $0$ and 
$\forall i\in S$, $x_i^{k+1}$ remains greater than $1/2$. 
The conditions of equations \ref{eqn_mis_basin_1} and \ref{eqn_mis_basin_2} thus still hold 
at iteration $k+1$. The weights of the nodes which do not belong to the independent set $S$ 
monotonically decrease to $0$ thus $\forall i \in S$, $r_i$ also monotonically decrease to $0$ and 
$x_i$ converges to $1$.
\end{proof}

\thLoneincreases*

\begin{proof}
We prove that:
\[
\forall x\in  {\mathbb{R}^{+*}}^{n}  \quad
\norm{\R(\R(x))}_1 \ge \norm{\R(x)}_1
\]
with equality if only if $x$ is a fixed point of $\R$.

Let $y = \R(x)$, $z=\R(y)=\R(\R(x))$ and $\Delta = z - y$.
Let $B = A+I$.

$z = y \HD By$ is equivalent to $y = z \HM By$. Hence:

\begin{eqnarray*}
\Delta &=& z - z \HM By \\
&=& \frac{y}{By} - \frac{y}{By}\HM By\\
&=& \frac{1}{By}\HM \left(  y -  y \HM By \right)
\end{eqnarray*}

And thus
\[
\sum_i \Delta_i = 
\sum_i \left[ \frac{1}{By}\HM \left(  y -  y\HM By \right) \right]_i
\]

Remark that each term of the righthand sum is not necessarily positive, and is not in general as some weights increase and other decrease by normalization.
However, whenever a term $y_i - y_i (By)_i$ is positive then $(By)_i$ is smaller than $1$ and 
thus $y_i - y_i (By)_i \ge 0 \Rightarrow \left(  y_i -  y_i (By)_i \right) / (By)_i \ge y_i -  y_i (By)_i$. 

In the same way, $y_i - y_i (By)_i \le 0 \Rightarrow (By)_i \ge 1 \Rightarrow \left(  y_i -  y_i (By)_i \right) / (By)_i \ge y_i -  y_i (By)_i$.

We thus get:
\begin{eqnarray}
\sum_i \Delta_i 
&\ge& \sum_i \left[ y -  y\HM By \right]_i \label{th_LoneIneq}\\
&=& \sum_i y_i -  \sum_i \left[ y \HM By \right]_i \nonumber\\
&=& \sum_i y_i -  \sum_i y_i \left[ By \right]_i \nonumber\\
&=& \sum_i y_i -  y^t By \nonumber
\end{eqnarray}

And thus if the conjecture \ref{conj_QL1} holds then:
\[
\sum_i \Delta_i  \ge 0
\]

As $\norm{\R(x^k)}_1 \le n$, the $L_1$ norm of the weights is increasing 
and bounded above, hence it converges. One easily verifies that equality in 
equation \ref{th_LoneIneq} 
only occurs on fixed points.
\end{proof}

\bibliography{IterativeGraphNormalization}{}
\bibliographystyle{plain}

\end{document}